\documentclass[11pt]{amsart}
\usepackage[a4paper,margin=2cm]{geometry}

\usepackage{graphicx}
\usepackage[T1]{fontenc}
\usepackage{lmodern}
\usepackage[cp1250]{inputenc}
\usepackage[breaklinks]{hyperref}

\newcommand{\Sref}[1]{(\hyperref[SP #1]{SP #1})}
\newcommand{\Lref}[1]{(\hyperref[L #1]{L #1})}
\newcommand{\Qref}[1]{(\hyperref[Q #1]{Q #1})}
\newcommand{\SBref}[1]{(\hyperref[SB #1]{SB SB #1})}
\newcommand{\Eref}[1]{(\hyperref[E #1]{E#1})}

\usepackage[ruled]{algorithm}
\usepackage[noend]{algpseudocode}
\usepackage{enumerate,hyperref,verbatim}

\newcommand{\G}{\ensuremath{\mathcal G}}
\newcommand{\morph}[2]{\ensuremath{h_{#1 \to #2}}}
\newcommand{\Morph}[2]{\ensuremath{H_{#1 \to #2}}}
\newcommand{\invmorph}[2]{\ensuremath{h^{-1}_{{#1} \to {#2}}}}
\newcommand{\Invmorph}[2]{\ensuremath{H^{-1}_{{#1} \to {#2}}}}
\newcommand{\blockc}[1][a]{\ensuremath{bl_{#1}}}
\newcommand{\Blockc}[1][a]{\ensuremath{Bl_{#1}}}
\newcommand{\invblockc}[1][a]{\ensuremath{bl^{-1}_{#1}}}
\newcommand{\Invblockc}[1][a]{\ensuremath{Bl^{-1}_{#1}}}

\newcommand{\presentletters}{\textnormal{\textsl{Letters}}}

\DeclareMathOperator{\per}{per}

\newcommand{\algpair}{\algofont{PairCompNCr}}
\newcommand{\algpairc}{\algofont{PairComp}}
\newcommand{\opprepend}[1][a,X]{\ensuremath{\algofont{Prepend}_{#1}}}
\newcommand{\opappend}[1][a,X]{\ensuremath{\algofont{Append}_{#1}}}

\newcommand{\algsolveeq}{\algofont{WordEqSat}}
\newcommand{\algsolveeqlin}{\algofont{LinWordEqSat}}

\newcommand{\algprefsuff}{\algofont{CutPrefSuff}}
\newcommand{\algprefsuffi}{\algofont{CutPrefSuffImp}}

\newcommand{\algblocks}{\algofont{BlockCompNCr}}
\newcommand{\algblocksc}{\algofont{BlockComp}}
\newcommand{\algblocksi}{\algofont{BlockCompImp}}
\newcommand{\algpop}{\algofont{Pop}}

\newcommand{\algdiophantine}{\algofont{VerifyDiophantine}}
\newcommand{\wordtodiophantine}{\algofont{WordtoDioph}}

\newcommand{\algofont}[1]{\textnormal{\textsc \selectfont\sffamily  #1}}

\newcommand{\Makanin}{\algofont{MakSAT}}
\newcommand{\algPlandowskiRytter}{\algofont{PlaRytSAT}}
\newcommand{\algPlandowskiSTOC}{\algofont{PlaSat2EXP}}
\newcommand{\algPlandowskiFOCS}{\algofont{PlaSat}}
\newcommand{\algPlandowskiSTOCi}{\algofont{PlaSatImp}}
\newcommand{\algPlandowskiSTOCdwa}{\algofont{PlaSolve}}

\newcommand{\sol}[1]{\ensuremath{S(#1)}}
\newcommand{\solution}{\ensuremath{S}}
\newcommand{\obaeq}{\sol U and \sol V}
\newcommand{\letters}{\ensuremath{\Gamma}}
\newcommand{\variables}{\ensuremath{\mathcal X }}

\usepackage{color}
\newcommand{\NPclass}{{\sf NP}}

\newcommand{\PSPACE}{{\sf PSPACE}}
\newcommand{\NPSPACE}{{\sf NPSPACE}}

\newcommand{\NEXPTIMEclass}{{\sf NEXPTIME}}
\newcommand{\DEXPTIMEclass}{{\sf DEXPTIME}}
\newcommand{\poly}{{\sf {poly}}}

\newtheorem{theorem}{Theorem}
\newtheorem{lemma}{Lemma}
\newtheorem{corollary}{Corollary}
\theoremstyle{remark}
\newtheorem{remark}{Remark}
\newtheorem{example}{Example}
\theoremstyle{definition}
\newtheorem{definition}{Definition}

\providecommand{\Ocomp}{\mathcal{O}}

\newcommand{\twodots}{\mathinner{\ldotp\ldotp}}

\newtheorem{clm}{Claim}

\usepackage{color}
\definecolor{myYellow}{rgb}{0.9,0.9,0}

\title{Recompression: a simple and powerful technique for word equations
}

\author[A. Je\.z]{Artur Je\.z\\
\\
Institute of Computer Science\\
University of Wroc\l{}aw\\
Wroc\l{}aw, Poland
}

\thanks{This work was partially supported by
NCN grant number 2011/01/D/ST6/07164, 2011--2015.}


\begin{document}
\begin{abstract}
In this paper we present an application of a simple technique of local recompression,
previously developed by the author in the context algorithms for compressed strings~\cite{fullyNFA,FCPM,grammar},
to word equations.
The technique is based on local modification of variables (replacing $X$ by $aX$ or $Xa$)
and iterative replacement of pairs of letters occurring in the equation by a~`fresh' letter,
which can be seen as a bottom-up \emph{compression} of the solution
of the given word equation, to be more specific, building an SLP (Straight-Line Programme)
for the solution of the word equation.

Using this technique we give a new, independent and self-contained
proofs of many known results for word equations.
To be more specific, the presented (nondeterministic) algorithm
runs in $\Ocomp(n \log n)$ space and in time polynomial in $n$ and $\log N$,
where $n$ is the size of the input equation and $N$ the size of the length-minimal solution of the word equation.
Furthermore, for a $\Ocomp(1)$ variables the bound on the space consumption is in fact linear,
i.e.\ $\Ocomp(m)$ where $m$ is the size of the space used by the input.
This yields that for each $k$ the set of satisfiable word equations with $k$ variables is context-sensitive.
The presented algorithm can be easily generalised 
to a generator of all solutions of the given word equation
(without increasing the space usage).
Furthermore, a further analysis of the algorithm yields an independent proof of doubly exponential
upper bound on the size of the length-minimal solution.
The presented algorithm does not use
exponential bound on the exponent of periodicity.
Conversely, the analysis of the algorithm yields
an independent proof of the 
exponential bound on exponent of periodicity.
\end{abstract}

\maketitle

\section{Introduction}
\subsection*{Word equations}
Since the dawn of the computer science, the problem of \emph{word equations}
was one of the most intriguing on the intersection between algebra and formal languages:
given words $U$ and $V$, consisting of letters (from \letters) and variables
(from \variables) we are to check the \emph{satisfiability},
i.e.\ decide, whether there is a substitution for variables,
which turns this formal equation into an equality of strings of letters.
It is useful to think of a solution \solution{} as a homomorphism
$\solution : \letters \cup \variables \mapsto \letters^*$,
which is an identity on \letters.
In the more general problem of \emph{solving} the equation,
we are to give representation of (all or some) solutions of the equation.

The problem of satisfiability of word equations
was first fully solved by Makanin~\cite{makanin}.
The proposed algorithm \Makanin{} transforms equations and large part
of Makanin's work consists of proving that this procedure in fact terminates.
While terminating, \Makanin{} complexity is very high.
Over the years the algorithm was gradually improved:
by Jaffar and independently Schulz to {\sf 4-NEXPTIME}~\cite{Jaffar90,Schulz90}
Ko\'scielski and Pacholski to {\sf 3-NEXPTIME}~\cite{Koscielski},
by Diekert to {\sf 2-EXPSPACE} (unpublished) and by Guti{\'e}rrez
to {\sf EXPSPACE}~\cite{Gutierrez98}.
It is worth mentioning that for 20 years no essentially different algorithm
than \Makanin{} was proposed.
On the other hand, as for today only a simple \NPclass{} lower bound is known
and it is widely believed that this problem is in \NPclass.

One of the key factors in the proof of termination, as well in later
estimations of the complexity of the algorithm,
 was the estimation the upper bound on
\emph{exponent of periodicity} of the solution.
Roughly speaking, the exponent of periodicity of a word $w$
is the largest $p$ such that $w = w_1 u ^p w_2$ for some $u \neq \epsilon$.
The original proof of Makanin gave a doubly exponential bound on the exponent of periodicity
of any length-minimal solution of word equations.
Later it was shown by Ko\'scielski and Pacholski that exponent of periodicity is
at most exponential~\cite{Koscielski}, this bound is tight.

A major independent step in the field was done by
Plandowski and Rytter~\cite{PlandowskiICALP},
who for the first time applied the notion of the \emph{compression}
to the solutions of the word equations:
they have shown that each length-minimal 
solution of the word equation is highly compressible,
in the sense that using LZ77 compression
(a popular practical standard for compression)
we can represent each length-minimal solution of word equations (of size $N$)
using an $\Ocomp(\log N)$-size encoding.
This implies that also LZ77-encoding of values of variables in such a solution
has size $\Ocomp(\log N)$.
Thus, to solve the word equation it is enough to guess the LZ77-encoding
of \sol X for each variable $X$ and verify that \sol U = \sol V under this substitution.
The latter can be done using known (though recent at that time)
polynomial methods for testing the equivalence of two SLPs~\cite{PlandowskiSLPequivalence}.
This yielded a new algorithm for word equations satisfiability,
which works in (nondeterministic) polynomial time in terms of $\log N$ and $n$.
Unfortunately, at that time the only bound on $N$ followed from the original
Makanin's algorithm, and it was triply exponential.
This gave a {\sf 2-NEXPTIME} algorithm, which was worse than
{\sf EXPSPACE}~\cite{Gutierrez98} published in the same year (though a little later).

Later, Plandowski gave a doubly-exponential upper bound on the size
of the minimal solution~\cite{PlandowskiSTOC},
which immediately yielded a \NEXPTIMEclass{} algorithm \algPlandowskiSTOC{}
for the problem.
This upper bound was obtained by a clever and careful
analysis of the minimal solution using so-called $\mathcal D$-factorisations,
suggested by Mignosi.

Soon after, another algorithm \algPlandowskiFOCS,
with a \PSPACE\footnote{The presented algorithm has running time proportional to $N$,
however, it can be extended so that it has the same running-time bounds as the earlier \algPlandowskiSTOC~\cite{PlandowskiPersonal}.}
upper-bound was given by Plandowski~\cite{PlandowskiFOCS}.
This algorithm starts with a trivial equation $e = e$ and has a set
of operations that can be performed on the equation;
so it can be seen as a rewriting system.
The set of rewriting rules is quite simple and thus also the algorithm
is easy to understand,
moreover it is obvious that the rewriting rules are sound (i.e.\ preserve satisfiability).
However the proof of completeness of this rewriting system
(i.e.\ that it properly generates all satisfiable equations) is involved.
It was based on usage of exponential expressions, which can be seen as a very
simple compression, and on indexed factorisations of words,
which extend the already mentioned $\mathcal D$-factorisations.

In some sense one can think that this result was obtained in stages,
as \algPlandowskiRytter, fuelled with theoretical results on
$\mathcal D$-factorisations, yielded \algPlandowskiSTOC{} and this in in turn
was upgraded to \algPlandowskiFOCS, by exploiting better the interplay between
the compression and factorisations.

All mentioned algorithms have a little drawback: while they check satisfiability
and can be modified to return \emph{some} solution of the word equation,
they do not \emph{solve} it in the sense that they do not provide a representation
of all solutions.
This was fully resolved by Plandowski~\cite{PlandowskiSTOC2},
who gave an algorithm \algPlandowskiSTOCdwa, which runs in \PSPACE
\footnote{\algPlandowskiSTOCdwa{} is implemented in \PSPACE,
but the generated representation can be exponential and thus only \DEXPTIMEclass{} running time was claimed in the original paper.}
and generates a compact representation of all (finite) solutions of a word equation.
This algorithm uses an improved version of \algPlandowskiFOCS,
the~\algPlandowskiSTOCi, as subprocedure.
The representation of the solutions is a directed multigraph,
whose nodes are labelled with expressions and
edges define substitutions for constants and variables.
Such representation reduces many properties of word equations to
reachability in graphs (which were exponentially larger),
for instance the problem of finiteness of set of solutions is shown to be in \PSPACE.

Some research was also done in the restricted variants of word equations,
most notably, there are polynomial-time algorithms for equations
with only two variables~\cite{twovarold,twovarnew}.
The variant with only one variable has almost-linear running time~\cite{onevarold};
the special case of only one variable with $\Ocomp(1)$ occurrences in the equation
has an optimal linear-time algorithm~\cite{onevarnew}, which works in a very simple computational model.

\subsection*{Our contribution}
In this paper, we present an application of a simple technique
of \emph{local recompression}
developed by the author and successfully applied to problems related with compressed data~\cite{fullyNFA,FCPM,grammar}.

\subsubsection{Recompression}
The idea of the technique is easily explained in terms of solutions
of the equations (i.e.\ words) rather than the equations themselves:
consider a solution $\sol U = \sol V$ of the equations $U = V$.
In one phase we first list all pairs of different letters $ab$ that occur as substrings
in \sol U and \sol V. For a fixed pair $ab$ of this kind
we greedily replace all occurrences of $ab$ in \sol U and \sol V by a new letter $c$.
(A slightly more complicated action is performed for pairs $aa$,
for now we ignore this case to streamline the presentation of the main idea).
There are possible conflicts between such replacements for different pairs
(consider string $aba$, in which we try to replace both pairs $ab$ and $ba$),
we resolve them by introducing some arbitrary order on types of pairs
and performing the replacement for one type of pair at a time, according to the order
(so in the example, we can first compress $ab$, obtaining $ca$ and then $ba$, which has no effect).
When all such pairs are replaced, we obtain another equal strings
$\solution'(U')$ and $\solution'(V')$ (note that the equation $U = V$
may have changed, and the new one is $U'  =V'$).
Then we iterate the process.
In each phase the strings are shortened by a constant factor,
and so after $\Ocomp(\log N)$ rounds we obtain a pair of trivial
(i.e.\ consisting of a single letter) strings.
Now, the original equation is solvable if and only if the obtained letters are the same.

The presented method has many variants, for instance, the pairs that occur seldom are not compressed,
pairs that do not overlap are compressed simultaneously etc. However, the respective variants
are always based on the general idea and the modifications are introduced to reach some specific goal.

The most problematic part of this idea is that it performs the operation on the solutions,
which can be large.
If we were to simply guess the solution and then perform the compressions,
this would have running time polynomial in $N$, which is not acceptable.
We circumvent the problem, by performing the compression
directly on the equation (the \emph{recompression}):
the pairs $ab$ occurring in the solution are identified using only
the equation and the compression of the solution is done implicitly,
by compressing the constants in the equations.
However, not all pairs of letters can be compressed in this way,
as some of them occur on the `crossing' between a variable and a constant:
consider for instance $\sol X = ab$, a string of symbols $Xc$
and a compression of a pair $bc$.
This is resolved by \emph{local decompression} part of the method:
when trying to compress the pair $bc$ in the example above we first replace
$X$ by $Xb$ (implicitly changing \sol X from $ab$ to $a$),
obtaining the string of symbols $Xbc$, in which the pair $bc$ can be easily compressed.

By simple calculations it can be shown that this method:
\begin{itemize}
	\item transforms solvable equations to solvable equations
	(for proper nondeterministic choices);
	\item transforms unsolvable equations to unsolvable equations
	(for all nondeterministic choices);
	\item does not introduce new variables;
	\item in each phase shortens each string (of letters) by a constant factor;
	\item in one phase introduces only a linear number of new letters to the equation.
\end{itemize}
In this way, correctness easily follows and both the $\Ocomp(\log N \poly(n))$ time 
and \PSPACE{} bounds hold.

\begin{example}
Consider an equation $aXca = abYa$ with a solution $\sol X = baba$ and $\sol Y = abac$.
In the first phase,
the algorithm wants to compress the pairs $ab$, $ca$, $ac$, $ba$ in this order.
To compress $ab$, it replaces $X$ with $bX$, thus changing the substitution into $\sol X= aba$.
After compression we obtain equation $a'Xca = a'Ya$.
Notice, that this implicitly changed solution into $\sol X = a'a$ and $\sol Y = a'ac$
To compress $ca$ (into $c'$), we replace $Y$ by $Yc$, thus implicitly changing the
substitution into $\sol Y = a'a$.
Then, we obtain the equation $a'Xc' = a'Yc'$ with a solution $\sol X = a'a$
and $\sol Y = a'a$.
The remaining pairs no longer occur in the equations,
and so we can proceed to the next phase.
\end{example}

The main features of the presented technique is that, at the same time:
	it is easy to state and apply,
	its proof of correctness is simple and straightforward,
	only basic properties of word equations and strings  
	are used in the design, application and analysis. 
The last property seems to be the most surprising,
as in order to apply the technique,
no understanding of the word equations and its solutions is actually needed.
This is completely different than the approaches based on Makanin's
algorithm~\cite{makanin,Jaffar90,Schulz90,Koscielski,Gutierrez98}
and Plandowski's constructions~\cite{PlandowskiSTOC,PlandowskiFOCS,PlandowskiSTOC2};
however, the \algPlandowskiRytter~\cite{PlandowskiICALP} shared this treat.

\subsubsection*{Results}
Using the technique of local recompression we give a (nondeterministic)
algorithm for testing satisfiability of word equations that works in time $\Ocomp(\log N \poly(n))$ and in $\Ocomp(n \log n)$ (bit) space.
Furthermore, a more detailed analysis yields that for $\Ocomp(1)$ variables the space consumption can be lowered to $\Ocomp(m)$,
where $m$ is the space (counted in bits) used by the input,
thus showing that for each fixed $k$ the set of satisfiable word equations with $k$ variables is context-sensitive.

The presented algorithm and its analysis are stand-alone,
as they do not assume any (non-trivial) properties of the solutions
of word equations.
To the contrary, it supplies an easy proof of doubly-exponential
upper bound of Plandowski~\cite{PlandowskiSTOC} on lengths of length-minimal solutions
as well as giving a new proof of exponential bound on the exponent of periodicity
(though slightly weaker than the one presented by Ko\'scielski
and Pacholski~\cite{Koscielski}).

The presented method can be easily modified, so that it can be used as a subprocedure
in an algorithm generating a representation of all solutions,
similarly as \algPlandowskiSTOCi{} in \algPlandowskiSTOCdwa.
The representation provided by our algorithm is similar to representation
provided by \algPlandowskiSTOCdwa,
i.e.\ a directed multigraph with edges representing substitutions.
Then the algorithm for testing satisfiability is used
to find out whether there is an edge between two given nodes and what
is the substitution labelling it.
The whole modification to our algorithm consists of replacing non-deterministic guesses
of lengths of strings by guessing the arithmetical relation that these lengths satisfy.

\subsubsection*{Presentation}
We start off with presenting a recompression-based algorithm for word equations,
in Section~\ref{sec:main}.
Firstly, we shall describe only its basic properties, which are needed to show
that it works in \PSPACE{} and has $\Ocomp(\log N \poly(n))$ (nondeterministic) running time.
More involved definitions as well as results are given in the following sections.
To be more precise, in Section~\ref{sec:blocks} we analyse in more detail
the structure of maximal repetitions of one letter in solutions of word equations.
This allows reduction of space consumption to $\Ocomp(n \log n)$ and is essentially used in following sections.
Using these results and a special encoding of letters we show that for $\Ocomp(1)$ variables 
we can lower the space consumption of the algorithm to linear one,
hence showing that the word equations with $k$ variables (for a fixed $k$)
are context-sensitive; this is presented in Section~\ref{sec:linear space}.
Then in Section~\ref{sec:solutions}, we recall the classification of solutions, given by Plandowski~\cite{PlandowskiSTOC2}, and related notions.
Using this classification we generalise the main notions and algorithm
to a generator of all solutions, see Section~\ref{sec:generator}.
Lastly, in Section~\ref{sec:theoretical bounds}, we show that a more detailed analysis of the algorithm also yields alternative (simple) proofs of
exponential bound on the exponent of periodicity and double exponential bound on the
size of the length minimal solutions

\subsubsection*{Comparison with previous approaches to word equations}
The presented method and the obtained algorithm is independent from
all previously known algorithms for word equations,
i.e.\ from original \Makanin{} and its variants,
from \algPlandowskiRytter{} (and its variant \algPlandowskiSTOC),
from \PSPACE{} algorithm \algPlandowskiFOCS{} as well as
its modification \algPlandowskiSTOCi.
In fact, the only algorithm, with which it can be somehow compared,
is the LZ77-based \algPlandowskiRytter~\cite{PlandowskiICALP}.
The key difference was that Plandowski and Rytter 
showed that a length-minimal solution has a short LZ77-representation
and then explicitly guessed and verified it.
Furthermore, the guessing was in some sense done in top-down fashion.
Thus their solution, in some sense, was `global'
(as it guessed the whole solution in one go and did it top-down)
and based on solutions' properties (in particular a bound on the size of the length-minimal solution
is needed to bound the running time of \algPlandowskiRytter). 
The novelty and importance of the here proposed method is that it does not use
properties of the solutions and that it is very `local',
in the sense that it does not try to build the solution in one go,
instead it modifies the equations and variables locally.
In particular, in this way we are working with an SLP-encoding of the solution,
which is easier in handling than the LZ77-representation.

Lastly the presented algorithm
uses only a very limited variant of exponent of periodicity,
when the strings in question consist only of repetitions of a single letter.
In such a case an exponential bound is easy to obtain.
This makes the presented algorithm somehow similar to \algPlandowskiRytter,
which does not use at all the bound on exponent of periodicity.

We believe that the presented algorithm is simpler from the previously applied.
This is of course a personal feeling, but it is backed up by a smaller
memory consumption.
This is also backed up by a follow-up work employing this approach as well:
in another work of the author, it was shown that
the recompression approach in the case of equations with only one variable
(and arbitrary many occurrence of it) yields a linear-time algorithm~\cite{onevarlinear},
which is also some argument in favour of this method.
Secondly, the recompression approach to word equations generalises to terms,
which allowed showing that context-unification (which is a natural problem between word equations and second-order unification)
is decidable in \PSPACE; so far this is the only algorithm for word equations that was generalised to context unification.

\subsubsection*{Related techniques}
While the presented method of recompression is relatively new,
some of its ideas and inspirations go quite back.
This technique was developed in order to deal with
fully compressed membership problem for NFA and the previous work
on this topic by Mathissen and Lohrey~\cite{LohreySLP}
already implemented the idea of replacing strings with fresh letters
as well as modifications of the instance so that this is possible
and treated maximal blocks of a single letter in a proper way.
However, the replacement was not iterated, and the newly introduced
blocks could not be further compressed.

The idea of replacing short strings by a fresh letter and iterating
this procedure was used by Mehlhorn et.~al~\cite{MehlhornSU97},
in their work on data structure for equality tests for dynamic strings
(cf.~also an improved implementation of a similar data structure by Brodal~et~al.~\cite{AlstrupBrodalRauhe}).
They viewed this process as `hashing'.

A similar technique, based on replacement of pairs and blocks of the same letter
was proposed also by Sakamoto~\cite{SLPaproxSakamoto} in the context of
constructing a smallest SLP generating a given word.
His algorithm was inspired by the \algofont{RePair} algorithm~\cite{RePair},
which is a practical grammar-based compressor.
It possessed the important features of the method: iterated replacement of pairs and blocks,
phases (i.e.\ ignoring letters recently introduced).
However, the analysis that stressed the modification of the variables (nonterminals)
was not introduced and it was done in a more crude way.
Additionally, Sakamoto introduced a special (and involved) pairing technique,
which greatly increases the conceptual complexity of his work.

\subsubsection*{Citing conventions}
As this paper aims at being stand alone,
many lemmata known from the literature, are supplied with proofs
(though sometimes different than the original ones).
Thus, in order to distinguish these two types of results,
whenever a theorem/lemma has a citation, it means that it was shown before,
perhaps in a slightly different variant.
Otherwise, the theorem/lemma is new.

\section{Main notions and techniques: local compression}
\label{sec:notions}
Let us formalise the main notions.
By $\letters$ we denote the set of letters occurring in the equation $U = V$
or are used for representation of compressed strings
(we do not use $\Sigma$ for this purpose as it is often used for summations).
The set $\variables$ denotes a set of variables.
The equation is written as $U = V$, where $U, V \in (\letters \cup \variables)^*$.
By $|U|$, $|V|$ we denote the length of $U$ and $V$,
$n$ denotes the length the input equation,
$n_v$ denotes the number of occurrences of variables in the input equation.

A \emph{substitution} is a morphism
$\solution : \variables \cup \letters \to \letters^*$,
such that $\sol a = a$ for every $a \in \letters$.
Each substitution is naturally extended to $(\variables \cup \letters)^*$.
The name represents the intuitive meaning that substitution
simply replaces variables by (some) strings.
A \emph{solution} of an equation $U = V$ is
a substitution \solution, such that $\sol U = \sol V$.
We exclude solutions (and substitutions) that substitute $\epsilon$ for $X$ that is present in the equation.
This is not restricting, as the general word equations reduce easily to this case:
given a word equation it is enough to guess for each variable $X$ whether $\sol X = \epsilon$ or not
and remove from the equation the variables for which we guessed that they have $\epsilon$ as a solution.
On the other hand, by convention, we assume that $\sol X = \epsilon$ for every variable $X$ that is not present in the equation
(note that this somehow corresponds to removing the variable from the equation: when we remove $X$ from the equation,
we `assume' that $\sol X = \epsilon$, while when $\sol X = \epsilon$ we can in fact remove $X$ from the equation,
without affecting the satisfiability).

Clearly, some solutions are `smaller' than other and we are naturally
interested in the `smallest':
We say that a solution \solution{} is \emph{length-minimal},
if for every solution $\solution'$ it holds that $|\sol U| \leq  |\solution'(U)|$.

\subsubsection*{Operations}
In essence, the presented technique is based on performing two operations
on \sol U  and \sol V, consider the first one:
\begin{description}
	\item[pair compression of $ab$]
	For two different letters $ab$ occurring in \sol U
	replace each of $ab$ in \obaeq{} by a fresh letter $c$.
\end{description}	

The compression of pair $aa$ is ambiguous (consider pair $aa$ and a string $aaa$)
and thus problematic, we need a better notion.
For a letter $a \in \letters$ we say that $a^\ell$ is a $a$'s
\emph{maximal block} of length $\ell$ for \solution,
if $a^\ell$ occurs in \sol U (or \sol V) and this occurrence
cannot be extended by $a$ nor to the left, neither to the right.
We refer to \emph{$a$'s $\ell$-block} for shortness.
Now, we can introduce the second operation performed on the solutions:
\begin{description}
	\item[block compression for $a$]
	For a letter  $a$ occurring in \sol U and each $\ell > 1$ 
	replace all maximal blocks $a^\ell$s in \obaeq{} by a fresh letter $a_\ell$.
\end{description}

The lengths of the maximal blocks can be upper bounded using the well-known
exponential bound on exponent of periodicity:
\begin{lemma}[Exponent of periodicity bound~\cite{Koscielski}]
\label{lem:periodicity bound original}
If solution \solution{} is length-minimal and $w^\ell$ for $w \neq \epsilon$
is a substring of \sol U,
then $\ell \leq 2^{cn} $ for some constant $0 < c < 2$.
\end{lemma}
We shall use exponent of periodicity bound only
to estimate the lengths of the maximal blocks
(i.e.\ restrict $w$ to single letters in the above definition),
and in such a case the proof becomes substantially easier
than the general one, see Section~\ref{sec:per bound}.
Furthermore, an alternative approach, which does not need the exponent of
periodicity at all, is also possible, see Section~\ref{sec:blocks}.

\subsubsection*{Fresh letters}
As our algorithm runs in \PSPACE, it may introduce a large number
of `fresh letters', and so if we insist that each of them is in fact
different, this becomes problematic. However, it is enough to assume
that a `fresh letter' does not occur in the equation: after all,
even if it occurred in some other iteration, this is completely irrelevant.

\begin{remark}
\algsolveeq{} introduces new letters to the instance,
replacing pairs of letters or maximal blocks of one letter.
We insist that these new symbols are called and treated as letters.
On the other hand, we can think of them as non-terminals of a
context-free grammar (to be more specific, of so-called SLP):
if $c$ replaced $ab$, then this corresponds to a production $c \to ab$,
similarly, $a_\ell \to a^\ell$.
In this way we can think that \algsolveeq{} builds a context-free grammar
(an SLP) generating \sol U as a unique word in the language.
\end{remark}

\subsubsection*{Types of pairs and blocks}
Both pair compression and block compression
(however they are implemented) shorten \sol U (and \sol V),
which gives the main foundation for this technique.
On the other hand, sometimes it is hard
to perform these operations: for instance, if we are to compress
a pair $ab$ and $aX$ occurs in $U$, moreover, \sol X begins with $b$,
then the compression is problematic, as we need to somehow
modify $\sol X$.
The following definition allows distinguishing between pairs (blocks)
that are easy to compress and those that are not.
\begin{definition}[cf.~\cite{fullyNFA,FCPM}]
Given an equation $U = V$ and a substitution \solution{}
and a substring $u \in \letters^+$ of \sol U (or \sol V)
we say that this occurrence of $u$ is
\begin{itemize}
	\item \emph{explicit}, if it comes from substring $u$ of $U$ (or $V$, respectively)
	\item \emph{implicit}, it it comes from \sol X for some variable $X$
	\item \emph{crossing} otherwise.
\end{itemize}
A string $u$ is \emph{crossing} (with respect to a solution \solution)
if it has a crossing occurrence and \emph{non-crossing} (with respect to a solution \solution) otherwise.

We say that a pair of $ab$ is a
\emph{crossing pair} (with respect to a solution \solution), if $ab$ has a crossing occurrence.
Otherwise, a pair is \emph{non-crossing}.
Unless explicitly stated, we consider crossing/non-crossing pairs $ab$ in which $a \neq b$.
Similarly, a letter $a \in \letters$ has a \emph{crossing block},
if there is a maximal block of $a$ which has a crossing occurrence.
This is equivalent to a (simpler) condition that $aa$ is a crossing pair.
\end{definition}

Compression of noncrossing pairs is easy, so is block compression when $a$ has no crossing block.
In other cases, the compression seems difficult.

\subsubsection*{Visible lengths of blocks}
We say that $a^\ell$ is visible in \solution{}
(or $\ell$ is a \emph{visible length} of $a$ block in \solution),
if there is an occurrence of the $a$'s $\ell$-block that is explicit or crossing
or it is a prefix or suffix of some \sol X;
we say that $\ell$ is a \emph{visible length} for $a$ if there is a visible maximal block $a^\ell$.

The following lemma shows that if a pair occurs in the length-minimal solution then it has a crossing or an explicit occurrence;
similarly, all lengths of maximal blocks are visible.
This means that in order to know what are the pairs and blocks occurring in the length minimal solution,
it is enough to know for each variable $X$, what is the first and last letter of \sol X
and what is the length of the $a$-prefix and $b$-suffix of \sol X.

\begin{lemma}[{cf.~\cite[Lemma~6]{PlandowskiICALP}}]
\label{lem:maximal block is crossing}
\label{lem:over a cut}
Let \solution{} be a length-minimal solution of $U = V$.
\begin{itemize}
	\item If $ab$ is a substring of \sol U, where $a \neq b$,
	then $ab$ is an explicit pair or a crossing pair.
	\item If $a^k$ is a maximal block in \sol U
	then $a$ has an explicit occurrence in $U$ or $V$
	and there is a visible occurrence of $a^k$.
\end{itemize}
\end{lemma}
\begin{proof}
Suppose that $ab$, where $a \neq b$ has only implicit occurrences.
Consider $\solution'$: $\solution'(X)$ is \sol X with all $ab$s removed,
i.e.\ replaced with $\epsilon$.
Since all occurrences of $ab$ in \sol U and \sol V are implicit,
$\solution'(U)$ ($\solution'(V)$) is obtained from \sol U (\sol V, respectively),
by removing all pairs $ab$. Hence $\solution'(U) = \solution'(V)$, i.e.\ $\solution'$
is a solution and it is clearly shorter than $\solution$, contradiction.

Similar argument shows that if $a$ occurs in \sol U then it has an explicit occurrence in $U$ or $V$.

To streamline the rest of the presentation and analysis,
in the remainder of the proof assume that both $U$ and $V$ begin
and end with a letter and not a variable;
this is easy to achieve by prepending $\$$ and appending $\$'$
to both sides of the equation.
Alternatively, the cases with variables beginning or ending $U$ or $V$ can be handled
in the same way, as the general case.

Consider a maximal $a$ block $a^k$, for $k > 0$ in \sol U
and the letter preceding (succeeding) it, say $b$ and $c$, respectively;
by the assumption that $U$ and $V$ begin and end with a letter, such $b$ and $c$
always exist.
Consider the occurrences of $ba^k c$ in \sol U and \sol V.
Since $b \neq a \neq c$, these occurrences cannot have overlapping $a$'s
(though, if $b=c$, these letters can overlap for different occurrences).
We want to show that one of these occurrences is crossing or explicit.
In such a case the corresponding $a^k$ proves that $k$ is a visible length,
which ends the proof.

So suppose that none of these occurrences is crossing nor explicit.
Consider $\solution'$: define $\solution'(X)$ as \sol X with each $ba^kc$ replaced with $bc$.
This operation is well defined, as the $a^k$ blocks are non-overlapping.
As in the case of $ab$ pairs it can be shown that $\solution'$ is a solution (since all $ba^k c$ are implicit),
which contradicts the assumption that $\solution$ is length-minimal.
\qedhere
\end{proof}

\subsection*{Compression of noncrossing pairs and blocks}
Intuitively, when $ab$ is non-crossing, each of its occurrence
in \sol U is either explicit or implicit.
Thus, to perform the pair compression of $ab$ on \sol U it is enough
to separately replace each explicit pair $ab$ in $U$
and change each $ab$ in \sol X for each variable $X$.
The latter is of course done implicitly (as \sol X is not written down anywhere).
The appropriate algorithm is given below.

\begin{algorithm}[H]
  \caption{$\algpair(a,b)$ Pair compression for a non-crossing pair \label{alg:pc}}
  \begin{algorithmic}[1]
  	\State let $c \in \letters$ be an unused letter
  	\State replace each explicit $ab$ in $U$ and $V$ by $c$
 \end{algorithmic}
\end{algorithm}

Similarly when none block of $a$ has a crossing occurrence,
the $a$'s blocks compression consists simply of replacing explicit $a$ blocks.

\begin{algorithm}[H]
  \caption{$\algblocks(a)$ Block compression for a letter $a$ with no crossing block
  \label{alg:ac}}
  \begin{algorithmic}[1]
  \For{each explicit $a$ occurring in $U$ or $V$}
		\For{each $\ell$ that is a visible length of an $a$ block in $U$ or $V$}
				\State let $a_\ell \in \letters$ be an unused letter
				\State replace every explicit $a$'s maximal $\ell$-block occurring
					in $U$ or $V$ by $a_\ell$
  \EndFor
  \EndFor
  \end{algorithmic}
\end{algorithm}

In order to show the correctness of those two procedures,
we need to first introduce some terminology and notation.

\subsubsection*{Soundness and completeness}
We say that a nondeterministic procedure \emph{is sound},
when given a unsatisfiable word equation $U = V$
it cannot transform it to a satisfiable one, regardless of the nondeterministic choices;
such a procedure \emph{is complete},
if given a satisfiable equation $U = V$ 
for some nondeterministic choices it returns a satisfiable equation $U' = V'$.
Observe, that a composition of sound (complete) procedures is sound (complete, respectively)

A procedure that is complete
\emph{implements pair compression of $ab$} for \solution,
if given an equation $U = V$ with a solution \solution,
for some nondeterministic choices it returns equation $U' = V'$
with a solution $\solution'$, such that $\solution'(U')$ is obtained from \sol U
by replacing each $ab$ by $c$;
similarly we say that a procedure implements blocks compression of $a$ for \solution.

Observe that a very general class of operations are sound:
\begin{lemma}
\label{lem:preserving unsatisfiability}
The following operations are sound:
\begin{enumerate}
	\item replacing occurrences of a variable $X$ with $wXv$ for arbitrary $w, v \in \Gamma^*$;
	\item replacing all occurrences of a word $w \in \Gamma^+$ (in $U$ and $V$) with a fresh letter $c$;
	\item replacing occurrences of a variable $X$ with a word $w$.
\end{enumerate}
\end{lemma}
\begin{proof}
In the first  case, if $\solution'$ is a solution of $U' = V'$ then \solution{} defined as
$\sol X = w \solution'(X) v$ and $\sol Y = \solution'(Y)$ otherwise is a solution of $U = V$.

In the second case, if $\solution'$ is a solution of $U' = V'$ then \solution{}
obtained from $\solution'$ by replacing each $c$ with $w$ is a solution of $U = V$.

Lastly, in the third case, if $\solution'$ is a solution of $U' = V'$ then we can obtain \solution{} from $\solution'$
by defining the substitution $\sol X = w$ and $\sol Y = \solution'(Y)$ in other cases.
\qedhere
\end{proof}

\subsubsection*{Properties of \algpair{} and \algblocks}
Now we are ready to show properties of $\algpair(a,b)$ and $\algblocks(a)$.
\begin{lemma}
\label{lem:paircomp blockcomp}
$\algpair(a,b)$ preserves is sound,
when $ab$ is a non-crossing pair in an equation $U = V$ (with respect to some solution \solution)
then it is complete and implements the pair compression of $ab$ for \solution.

Similarly, $\algblocks(a)$ is sound
and when $a$ has no crossing blocks in $U = V$ (with respect to some solution \solution)
it is complete and implements the block compression of $a$ for \solution.
\end{lemma}
\begin{proof}
From Lemma~\ref{lem:preserving unsatisfiability} it follows that 
both $\algpair(a,b)$ and $\algblocks(a)$ are sound.

Suppose that $U = V$ has a solution \solution{} such that $ab$ is a noncrossing pair with respect to \solution.
Define $\solution'$: $\solution'(X)$ is equal to \sol X with each $ab$ replaced with $c$
(where $c$ is a new letter).
Consider $\sol U$ and $\solution'(U')$. Then $\solution'(U')$ is obtained from \sol U
by replacing each $ab$: the explicit occurrences of $ab$ are replaced by $\algpair(a,b)$,
the implicit ones are replaced by the definition of $\solution'$ and by the assumption
there are no crossing occurrences.
The same applies to \sol V and $\solution'(V')$.
Hence $\solution'(U') = \solution'(V')$ and concludes the proof in this case.

The proof for the block compression follows in the same way.
\qedhere
\end{proof}

\subsection*{Crossing pairs and blocks compression}
The algorithms presented in the previous section cannot be directly applied to crossing pairs
or to compression of $a$'s blocks that have crossing occurrences.
To circumvent the problem, we modify the instance:
if a pair $ab$ is crossing because there is a variable $X$ such that $\sol X = bw$
for some word $w$ and $a$ is to the left of $X$,
it is is enough to change \solution, so that $\sol X = w$;
similar action is applied to variables $Y$ ending with $a$ and with $b$ to the right.

This idea can be employed much more efficiently: consider a partition of \letters{} into $\letters_\ell$ and $\letters_r$.
The `left-popping' from each variable a letter from $\letters_r$ and `right-popping' a letter from $\letters_\ell$
guarantees that each pair $ab \in \letters_\ell\letters_r$ is non-crossing.
Since pairs from $\letters_\ell\letters_r$ do not overlap, after the popping they can be compressed in parallel.
As shown later, for appropriate choice of $\letters_\ell$ and $\letters_r$
a constant fraction of pairs from \sol U is of the form $\letters_\ell\letters_r$, see Claim~\ref{clm:words are compressed}.

\begin{algorithm}[H]
  \caption{$\algpop(\letters_\ell,\letters_r)$ \label{alg:leftpop}}
  \begin{algorithmic}[1]
	\For{$X \in \variables$} \label{pop main loop}
		\State let $b$ be the first letter of \sol X  \label{guess first letter}\Comment{Guess}
		\If{$b \in \letters_r$} 
			\State replace each $X$ in $U$ and $V$ by  $bX$ \label{leftpop}
			\Comment{Implicitly change $\sol X = bw$ to $\sol X = w$}
			\If{$\sol X = \epsilon$} \Comment{Guess}
				\State remove $X$ from $U$ and $V$
			\EndIf
		\EndIf
		\State let $a$ be the \ldots  \Comment{Perform a symmetric action for the last letter}
	\EndFor
  \end{algorithmic}
\end{algorithm}

\begin{lemma}
\label{lem:pop preserves solutions}
The $\algpop(\letters_\ell,\letters_r)$ is sound and complete.

Furthermore, if \solution{} is a solution
of $U = V$ then for some nondeterministic choices the obtained $U' = V'$ has a solution $\solution'$
such that $\solution'(U') = \solution(U)$ and for pair $ab$ from $\letters_\ell\letters_r$
is non-crossing (with regards to $\solution'$).
\end{lemma}
\begin{proof}
From Lemma~\ref{lem:preserving unsatisfiability} it follows that 
$\algpop(\letters_\ell,\letters_r)$ is sound.

Conversely, suppose that $U = V$ has a solution \solution.
Let $\algpop(\letters_\ell,\letters_r)$ always guess according to \solution,
i.e.\ in line~\ref{guess first letter} it guesses $b$ that is indeed the first letter of \sol X,
and similarly $a$ that is the last letter of \sol X, finally it removes $X$, when $\sol X = \epsilon$.
Suppose that $b \in \letters_r$ and $a \in \letters_\ell$.
Consider $\solution'(X)$ defined as $b\solution'(X)a = \sol X$
(when $\sol X = a$ then $\solution'(X) =\epsilon$).
It is easy to observe that $\sol U = \solution'(U')$, similarly $\sol V = \solution'(V')$,
hence $\solution'$ is a solution of $U' = V'$.
Note that we are interested only in non-empty solutions: if $\sol X = \epsilon$ at any point
then we simply remove it from the equation, in which case the solution is turned into a non-empty one.

The cases in which $b \notin \letters_r$ or $a \notin \letters_\ell$ are done in the same way
(for instance, when $b \notin \letters_r$ and $a \in \letters_\ell$ then $\solution'(X)a = \sol X$).

It is left to show that in $U' = V'$ each pair $ab \in \letters_\ell\letters_r$ is noncrossing with respect to such defined $\solution'$.
Assume for the sake of contradiction that $ab$ is crossing with respect to $\solution'$ in $U' = V'$.
There are three cases
\begin{description}
	\item[$a$ is to the left of some variable $X$ and the first letter of \sol X is $b$]
	Since $a \in \letters_\ell$, then \algpop{} did not popped a letter $a$ from $X$ in line~\ref{leftpop}.
	Hence the first letter of \sol X and $\solution'(X)$ are the same.
	However, as in line~\ref{leftpop} the letter was not popped from $X$ and we consider the case in which
	\algpop{} guessed correctly the first letter, we conclude that the first letter of \sol X is not in $\letters_r$,
	while the first letter of $\solution'(X)$ is, contradiction.
	
	\item[$b$ is to the right of some variable $X$ and the last letter of \sol X is $a$]
	This case is symmetric to the previous one.
	
	\item[$XY$ occurs in the equation, \sol X ends with $a$ and \sol Y begins with $b$] 
	The analysis is similar to the one in the first case.
\end{description}
This ends the case inspection. Hence $ab$ after the loop in line~\ref{pop main loop}
is noncrossing with respect to $\solution'$.
Note that for appropriate choices, all pairs $ab$ in $\letters_\ell \letters_r$ 
become noncrossing.
\qedhere
\end{proof}

Now the presented subprocedures can be merged into one procedure
that turns crossing pairs into noncrossing ones and then compresses them,
effectively compressing crossing pairs.

\begin{algorithm}[H]
  \caption{$\algpairc(\letters_\ell,\letters_r)$ Turning crossing pairs from $\letters_\ell\letters_r$ into non-crossing ones and compressing them \label{alg:paircompc}}
  \begin{algorithmic}[1]
  			\State run $\algpop(\letters_\ell,\letters_r)$
			\For{$ab \in \letters_\ell\letters_r$} \label{loop lefta}
				\State run $\algpair(a,b)$ \label{crossing pair compression}
			\EndFor
  \end{algorithmic}
\end{algorithm}

\begin{lemma}
\label{lem: crossing pairs preserve}
\label{lem:crossing non crossing}
$\algpairc(\letters_\ell,\letters_r)$ is sound and complete.
To be more precise, for any solution \solution{} it implements the pair compression of each pair $ab \in \letters_\ell\letters_r$.
\end{lemma}
\begin{proof}
All subprocedures are sound, and so also $\algpairc(\letters_\ell,\letters_r)$ is.

Concerning completeness and the implementation of the pair compression:
By Lemma~\ref{lem:pop preserves solutions}, for appropriate choices after $\algpop(\letters_\ell,\letters_r)$
the obtained equation $U' = V'$ has a solution $\solution'$ such that $\sol U = \solution'(U')$ and each $ab \in \letters_\ell\letters_r$
is noncrossing with regards to $\solution'$
Then, by Lemma~\ref{lem:paircomp blockcomp}
each of $\algpair(a,b)$ implements the pair compression, when $ab$ is noncrossing.
As occurrences of different pairs $ab$ and $a'b'$ from $\letters_\ell\letters_r$ do not overlap,
a composition of $\algpair(a,b)$ for each $ab \in \letters_\ell\letters_r$
implements the pair compression for all $ab \in \letters_\ell\letters_r$.
This concludes the proof.
\qedhere
\end{proof}

The problems with crossing blocks can be solved in a similar fashion:
$a$ has a crossing block, if $aa$ is a crossing pair.
So we `left-pop' $a$ from $X$ until the first letter of \sol X is different than $a$,
we do the same with the ending letter $b$.
This can be alternatively seen as removing the whole $a$-prefix ($b$-suffix, respectively) from $X$:
suppose that $\sol X = a^\ell w b^r$, where $w$ does not start with $a$ nor end with $b$.
Then we replace each $X$ by $a^\ell X b^r$ implicitly changing the solution to $\solution'(X) = w$, see Algorithm~\ref{alg:prefix}.

\begin{algorithm}[H]
  \caption{\algprefsuff{} Cutting prefixes and suffixes \label{alg:prefix}}
  \begin{algorithmic}[1]
  \For{$X \in \variables$}
	\State let $a$, $b$ be the first and last letter of \sol X
	\State guess $\ell_X \geq 1$, $r_X \geq 0$ \label{guess ell}\Comment{$\sol X = a^{\ell_X} w b^{r_X}$, where $w$ does not begin with $a$ nor end with $b$}
	\State \Comment{If $\sol X = a_X^{\ell_X}$ then $r_X=0$}
	\State replace each $X$ in $U$ and $V$ by  $a^{\ell_X} X b^{r_X}$
		\Comment{$a_X^{\ell_X}$, $b_X^{r_X}$ is stored in a compressed form},
	\State		\Comment{implicitly change $\sol X = a_X^{\ell_X} w b_X^{r_X}$ to $\sol X = w$}
	\If{$\sol X = \epsilon$} \Comment{Guess}
		\State remove $X$ from $U$ and $V$
	\EndIf
  \EndFor
  \end{algorithmic}
\end{algorithm}

\begin{lemma}
\label{lem:cutpref cutsuff}
\algprefsuff{} is sound.
It is complete, to be more precise: For a solution \solution{} of $U = V$
let for each $X$ the $a_X$ be the first letter of $\sol X$ and $a_X^{\ell_X}$ the $a_X$ suffix of \sol X
while $b_X$ the last letter and $b_X^{r_X}$ the $b_X$ suffix.
Then when \algprefsuff{} pops $a_X^{\ell_X}$ to the left and $b_X^{r_X}$ to the right,
the returned equation $U' = V'$ has a solution $\solution'$ such that $\sol U = \solution'(U')$
and $U' = V'$ has no crossing blocks with respect to $\solution'$.
\end{lemma}
\begin{proof}
From Lemma~\ref{lem:preserving unsatisfiability} we obtain that \algprefsuff{} is sound.

We present the proof in the case when $\sol X \neq a^{\ell_X}$ for each variable,
the argument in the other case is similar.

Suppose that $U = V$ has a solution \solution.
Then let \algprefsuff{} guess according to \solution, i.e.\ let $\ell_X \geq 1$ and $r_X \geq 1$
be guessed so that $\sol X = a^{\ell_X} w_X b^{r_X}$, where $w_X$ does not begin with $a$ nor end with $b$
Define $\solution'(X) = w_X$. It is easy to see that $\sol U = \solution'(U')$ and $\sol V = \solution'(V')$,
in particular, $\solution'$ is a solution of $U' = V'$.
Furthermore, observe that as the first letter of $w_X$ is not $a$ and the last is not $b$,
there are no crossing blocks in $U' = V'$ with respect to $\solution'$.
\qedhere
\end{proof}

The \algprefsuff{} allows defining a procedure \algblocksc{}
that compresses maximal blocks of all letters, regardless of whether they have crossing blocks or not.

\begin{algorithm}[H]
	\caption{\algblocksc{} Compressing blocks of $a$}
	\label{alg:blocksc}
	\begin{algorithmic}[1]
		\State run \algprefsuff \Comment{Removes crossing blocks of $a$} \label{cut pref}
		\For{each letter $a \in \letters$} \label{loop of compressions}
			\State $\algblocks(a)$ \label{block compression local}
		\EndFor
	\end{algorithmic}
\end{algorithm}

\begin{lemma}
\label{lem: consistent no crossing block}
$\algblocksc$ is sound. It is complete, 
to be more precise, let $a_X$ be the first and $b_X$ the last letter of \sol X and
$\ell_X$ the length of the $a_X$-prefix and $r_X$ of $b_X$ suffix of \sol X
($r_X$ is undefined if \sol X is a block of letters).
Then for non-deterministic choices for which the \algprefsuff{} pops $a_X^{\ell_X}$ to the left and $b_X^{r_X}$ to the right
the \algblocksc{} implements the blocks compression.
\end{lemma}
\begin{proof}
The proof is similar to the proof of Lemma~\ref{lem:crossing non crossing}.
As \algblocksc{} is a composition of sound operations, it is also sound.

So suppose that $U = V$ has a solution \solution.
By Lemma~\ref{lem:cutpref cutsuff} after popping the $a_X$ prefix and $b_X$ suffix from each variable,
by \algprefsuff{} in line~\ref{cut pref}, the obtained (intermediate) equation
$U' = V'$ has a solution $\solution'$ such that $\sol U = \solution'(U')$ and 
$\sol V = \solution'(V')$ and there are no crossing blocks with respect to
$\solution'$ in $U' = V'$.
Then, by Lemma~\ref{lem:paircomp blockcomp}, each of the $\algblocks(a)$
is sound and implements the $a$ blocks compression.
As blocks of different letters are disjoint,
this means that the loop in line~\ref{loop of compressions} implements
the blocks compression for each letter $a\in \letters$.
\qedhere
\end{proof}

\section{Main algorithm, its time and space consumption}
\label{sec:main}
Now, the algorithm for testing satisfiability of word equations
can be conveniently stated.

\begin{algorithm}[H]
	\caption{\algsolveeq{} Checking the satisfiability of a word equation}
	\label{alg:main}
	\begin{algorithmic}[1]
	\While{$|U| > 1$ or $|V| > 1$} \label{alg:mainloop}
  		\State $\algblocksc$ \label{block compression} 		\label{presentletters 1}
		\State $\presentletters \gets $ the set of letters present in $U$ or $V$
		\For{$i\gets 1 \twodots 2$} \label{two iterations} \Comment{One iteration to shorten the solution, one to shorten the equation}
			\State guess partition of $\presentletters$ into $\presentletters_1$ and $\presentletters_2$ 
			\State $\algpairc(\presentletters_1,\presentletters_2)$ \label{pair compression crossing 1}
		\EndFor
	\EndWhile
	\State Solve the problem naively \label{naive solve}
		\Comment{With sides of length $1$, the problem is trivial}
 \end{algorithmic}
\end{algorithm}

We refer to one iteration of the main loop in \algsolveeq{} as one \emph{phase}.
Observe that 
one phase of \algsolveeq{} is executed in (nondeterministic) $\poly(|U| + |V|)$ time.

The somehow counter-intuitive repetition in line~\ref{two iterations} has very
simple explanation: one of the guessed partition guarantees that the solution's
size is reduced by a constant factor, the other guarantees the same for the equation.

The properties of \algsolveeq{} are summarised in the following theorem.
\begin{theorem}
\label{thm:main}
\algsolveeq{} nondeterministically verifies the satisfiability of word equations.
It can verify an existence of a length-minimal solution of length $N$
in $\Ocomp(\poly(n) \log N)$ time and $\Ocomp(n^2)$ space;
furthermore, the stored equation has length $\Ocomp(n)$.
\end{theorem}

The analysis of the space consumption is done in Lemma~\ref{lem: space consumption},
of time consumption in Lemma~\ref{lem:logM iterations}
while the correctness is shown in Lemma~\ref{lem:correctness of main}.
Furthermore, it is shown in the following section that the space consumption
can be bounded by $\Ocomp(n \log n)$.

\begin{lemma}
\label{lem: space consumption}
For appropriate nondeterministic choices,
the equations stored by (successful) computation of \algsolveeq{} are of length $\Ocomp(n)$,
the additional computation performed by \algsolveeq{} use $\Ocomp(n^2)$ space.

Furthermore, for appropriate nondeterministic choices,
the number of phases is at most $\Ocomp(\log n + n_v^{cn_v})$.
\end{lemma}
\begin{proof}
For the purpose of this proof let a \emph{symbol} be either
$a \in \letters$, or $a^\ell$, where $a \in \letters$
and $\ell = \Ocomp(2^{cn})$, for constant $c$ from Lemma~\ref{lem:periodicity bound original}.
Let us first calculate, how many symbols are introduced into the equation in one round.
By ``introduce'' we do not mean letters that merely replaced pairs or blocks during the compression,
but rather letters that were popped into the equation from the variables.

\algblocksc{} is run once and it runs (also once) \algprefsuff, which introduces two symbols per variable occurrence;
\algpairc{} is run $2$ times, each time it runs \algpop{} which introduces at most two symbols per variable occurrence.
Hence, in one round, at most $6 n_v$ letters are introduced into the equation.

On the other hand, the main task of the whole algorithm is compression:
it can be shown that for appropriate choices large fraction of letters in $U = V$ are compressed.

\begin{clm}
\label{clm:words are compressed}
Let $U = V$ has a solution \solution.
Consider a phase of \algsolveeq{} in which \algblocksc{} implements the blocks compression for \solution,
obtaining $U' = V'$ with a corresponding $\solution'$
the first invocations of \algpairc{} implements the pair compression obtaining
$U' = V'$ with $\solution'$ (obtaining $U'' = V''$ with $\solution''$)
and the second implements the pair compression for $U'' = V''$ with $\solution''$ (obtaining $U''' = V'''$ with $\solution'''$).
Then there are partitions $\presentletters_1$, $\presentletters_2$ and $\presentletters_1'$, $\presentletters_2'$
such that
\begin{itemize}
	\item $1/6$ of letters in \sol U (rounding down) is compressed in $\solution'''(U''')$;
	\item at least $(|U| + |V| - 3 n_v - 4)/6$ of letters in $U$ or $V$ are compressed in $U'''$ or $V'''$.
\end{itemize}
\end{clm}

This can be used to show (inductively) that the length of $U = V$ is at most $79 n$:
clearly this bound holds for the input instance, which is length $n$.
For the inductive step consider that
there are at most $6n_v$ symbols introduced into $U'$ and $V'$ (some of them might be compressed later).
On the other hand, by Claim~\ref{clm:words are compressed},
the number of original letters of $U$ and $V$ decreased by at least $(|U| + |V| - 3 n_v - 4)/12$.
Hence,
\begin{align}
\label{eq:equation length}
|U'| + |V'| &\leq |U| + |V| - (|U| + |V| - 3 n_v - 4)/12 + 6 n_v\\
	&\leq
\notag
\frac {11} {12} (|U| + |V|) + \frac {7} {12} n + 6 n\\
\notag
	&\leq
\frac {11} {12} \cdot 79 n + \frac {1} {12} \cdot 79 n\\
\notag
	&\leq
79n \enspace .
\end{align}
Note, that this is the number of symbols, and not letters.
However, each symbol representing $a^\ell$ is compressed
into a single letter before the end of the phase, so the given bound holds
for the number of letters as well.

Concerning the space consumption, there are three types of symbols in the equation:
\begin{itemize}
	\item individual letters
	\item blocks of letters popped from variables
	\item variables.
\end{itemize}
Individual letters clearly take at most $\log n $ bits each,
so $\Ocomp(n \log n)$ bits in total.

By Lemma~\ref{lem:periodicity bound original} we know that for the length-minimal solution,
the blocks of letters $a^\ell$
popped from a variable have length at most exponential in the length of the equation.
Since we are interested only in the satisfiability of the equation,
we may assume that the considered solution \solution{} is indeed length-minimal
and so these lengths can be encoded using $\Ocomp(n)$ bits,
which gives $\Ocomp(n^2)$ space consumption for such symbols in total
(at any moment we have at most $2n_v \leq 2n$ such letters).

Lastly, the space consumption of variables:
the number of variables is at most $n_v \leq n$,
(as \algsolveeq{} does not introduce new variable in to the equation)
and so they also fit in $\Ocomp(n \log n_v)$ bits.

Concerning the number of phases of \algsolveeq, 
observe that calculation similar to the one in~\eqref{eq:equation length}
shows that if the equation $U = V$ has length larger than
$120 n_v$, its length drops by a constant factor in a phase. Hence, after at most $\Ocomp(\log n)$
phases the length of the equation is reduced to $\Ocomp(n_v)$.
We can imagine that we restart \algsolveeq{} for this instance.
Since the length of the equation will not exceed $c n_v$ for some constant $c$
and the accepting computation clearly does not have loops,
we obtain that the number of phases is at most $\Ocomp(\log n + (cn_v)^{cn_v}) = \Ocomp(\log n + n_v^{c'n_v})$
for some larger constant $c'$.

It remains to give the proof of Claim~\ref{clm:words are compressed}.

\begin{proof}[proof of Claim~\ref{clm:words are compressed}]
We first show the first property.
Divide \sol U into three-letters segments (ignore the last, partial segment).
Consider a random partition of \presentletters{} into $\presentletters_1$ and $\presentletters_2$,
each letters goes into the part of the partition with probability $1/2$.
Take any segment occurring in \sol U, let it be $abc$.
We show that with probability at least $1/2$ at least one letter in this segment is compressed.

If any of those letters is equal to its neighbouring letters (perhaps outside this three-letter segment),
then it is compressed by Lemma~\ref{lem: consistent no crossing block}.
So suppose that none of these letters is the same as its neighbouring letters,
in particular, they are not compressed by \algblocksc.
There is a compression inside $abc$ if 
$ab \in \presentletters_1\presentletters_2$ or $bc \in \presentletters_1\presentletters_2$.
Each of those events has probability $1/4$ and they are disjoint, hence the compression occurs
with probability $1/2$.
So regardless of the case, with probability $1/2$ at least one of letters in $abc$ is compressed.
There are $\lfloor |\sol U|/3 \rfloor$ three-letter segments.
The expected number of segments in which at least one letter is compressed is thus at least $\lfloor |\sol U|/6 \rfloor$,
so for some partition at least $\lfloor |\sol U|/6 \rfloor$ letters are compressed.

Concerning the second property, observe, that the analysis above applies in the same way,
consider any explicit word $w'$ between two variables in $U$ or $V$ (or the explicit word beginning or ending $U$ or $V$).
Then the analysis is the same, except that the number of segments of $w'$ is at least $\lfloor |w'| / 3 \rfloor \geq |w'|/3 - 2/3$.
Let now $w_1$, $w_2$, \ldots, $w_k$ be all such words in $U = V$.
Then $\sum_{i=1}^k |w_i| \geq |U| + |V| - n_v$ (as at most $n_v$ symbols in the  equation are variables)
and $k \leq n_v + 2$ (as at most $n_v$ variables and the `$=$' sign are the ends of words).
So in total there are at least
\begin{align*}
\sum_{i=1}^k \left( \frac{|w_i|}{3} - \frac{2}{3}\right) &= \frac{\sum_{i=1}^k |w_i|}{3}  - \frac{2k}{3} \\
	&\geq
\frac{1}{3}\left((|U| + |V| - n_v) - 2(n_v + 2)\right)\\
	&=
\frac{|U| + |V| - 3 n_v - 4}{3}
\end{align*}
The same expected-value argument yields that at least $(|U| + |V| - 3 n_v - 4)/6$ letters are compressed,
note that the appropriate partition is guessed as the second partition
of $\presentletters$ to $\presentletters_1$ and $\presentletters_2$.
This shows the claim.\qedhere
\end{proof}
With the end of proof of Claim~\ref{clm:words are compressed}, the lemma follows.  \qedhere
\end{proof}

\begin{lemma}
\label{lem:logM iterations}
Let $N$ be the size of the length-minimal solution.
Then for appropriate nondeterministic choices \algsolveeq{} accepts
after $\Ocomp(\log N )$ phases.
\end{lemma}
\begin{proof}
The proof follows from the first item in Claim~\ref{clm:words are compressed}.
\end{proof}

\begin{lemma}
\label{lem:correctness of main}
\algsolveeq{} nondeterministically verifies the satisfiability of a word equation.
\end{lemma}
\begin{proof}
Firstly, observe that if $|U| = |V| = 1$ then the satisfiability of word
equation is trivial to verify, which is done in last line of \algsolveeq.

As \algsolveeq{} is a composition of sound and complete subprocedures,
it also is sound and complete. So if the equation is unsatisfiable, `YES' is never returned,
while if `YES' is returned, the original equation is satisfiable.
Finally, Lemma~\ref{lem:logM iterations} shows that for a satisfiable solution, i.e.\ a one that has
a length-minimal solution of length $N$ for some $N$,
\algsolveeq{} accepts the equation after $\Ocomp(\log N)$ phases (for appropriate nondeterministic choices).
Lastly, since the computation fits in polynomial memory and the accepting computation
should not loop, after an exponential number of steps (kept in a counter) we can reject.
\qedhere
\end{proof}

\section{Maximal blocks}
\label{sec:blocks}
The quadratic memory consumption of \algsolveeq{} is due to \algblocksc.
Since we aim at $\Ocomp(n \log n)$ memory consumption (counted in bits), we need to improve it.
To this end we analyse more carefully the structure and possible lengths of maximal blocks.
This analysis allows a different approach to blocks compression:
instead of guessing the explicit values of $a$-prefixes and $b$-suffixes of variables,
we parametrise those values and check which sets of values of those parameters are allowed.
To be more precise, the lengths of maximal blocks are expressed in terms of lengths of $a$-prefixes and $b$-suffixes
while blocks of the same lengths are identified (using non-deterministic guesses) and replaced by the same letter.
The verification of feasibility of the guesses boils down to checking the satisfiability of a system of linear Diophantine equations.
In particular, the actual lengths of the blocks are not important, it is the satisfiability of the system that matters
(in this way we can omit the space consuming guesses of the exact lengths);
due to special form of the Diophantine equations, their satisfiability can be checked in linear space.

The contents of this section is a simple case of the general approach (of decompositions according to some primitive words)
presented in the work of Ko\'scielski and Pacholski~\cite{Koscielski}.

\subsection*{Arithmetic expressions}
We shall now define what is a general form of lengths of maximal blocks in $\sol U = \sol V$.
Those lengths are parametrised by the lengths of $a$-prefixes and $b$-suffixes of \sol X,
and so are not simply numbers, but rather expressions involving both numbers and some parameters.

Consider arithmetic expressions using natural constants and variables,
such that all expressions are linear in these variables.
These expressions are obtained from a word equation $U = V$ in a way described in the following subsection.
We say that a set of $e_1$, $e_2$, \ldots, $e_m$
is a \emph{small set of linear Diophantine expressions}
(for a word equation $U = V$ with $n_v$ occurrences of variables)
if 
\begin{itemize}
	\item the coefficients and constants in each expression are positive natural numbers;
	\item each variable in the expressions is either $x_X$ or $y_X$, where $X$ is a variable from $U = V$;
	\item if $X$ occurs $k$ times in $U = V$ then the sum of coefficients of $x_X$ ($y_X$) is at most $k$;
	\item the sum of values of constants in $\{e_i\}_{i=1}^m$ is at most $|U| + |V| - n_v$.
\end{itemize}
We say that a system of linear Diophantine equations and inequalities (all inequalities are of the form $x \geq 1$) 
is a \emph{small linear Diophantine system},
if sides of its equalities form a small set of linear
Diophantine expressions and each $e_i$ in this set is used at most two times in this system.
As a simple consequence small system of linear Diophantine
equations and inequalities has the following properties:
\begin{itemize}
	\item each variable in the system is either $x_X$ or $y_X$, where $X$ is a variable from $U = V$;
	\item if $X$ occurs $k$ times in $U = V$ then $x_X$ ($y_X$) has sum of values of its coefficients at most $2k$;
	\item the sum of values of constants is at most $2(|U| + |V|)$;
	($2(|U| + |V|-n_v)$ comes from the equalities while $2n_v$
	from the right-hand sides of inequalities).
\end{itemize}

The size of the small linear Diophantine system is proportional to the size of representation of $U = V$
and furthermore its satisfiability can be (non-deterministically) checked in the same space limits.

\begin{lemma}
\label{eq:small system memory}
If the equation $U = V$ is represented using $m$ bits, then the corresponding small linear Diophantine system
can be encoded using $\Ocomp(m)$ bits, moreover, it can be (nondeterministically) verified in the same space,
whether it has a natural solution.
\end{lemma}
\begin{proof}
We encode the equalities in unary: i.e.\ each constant $c$ is represented as $\underbrace{1+1+\cdots+1}_{c \text{ times}}$,
while each $cx$ is represented as $\underbrace{x+x+\cdots+x}_{c \text{ times}}$.
The variables $x_X$, $y_X$ are encoded in the same way as $X$ in $U = V$,
with additional bit to distinguish them.
The assumptions on the small system guarantee that
\begin{itemize}
	\item the total space used by constants is $2(|U| + |V|-n_v)$, which is at most $2 m$;
	\item the space used by variables in equalities is at most $8$ times as much as the space used by variables in $U = V$:
	a denotation of a variable $x_X$ ($y_X$) is at most twice as long as the variable $X$ and it occurs at most
	$2$ times more in the small linear Diophantine system as $X$ in $U = V$;
	\item all inequalities use a variable and one bit to denote $1$, so it can be shown
	(as in the item above) that the space consumption is at most $6$ times as much as the space used by variables in $U = V$.
\end{itemize}
The additional space needed to denote `$+$' and `$=$' may increase the space usage only by a constant
(note that we might need to change the denotation of other symbols a bit).
So indeed the used space is just constantly larger than the space used by $U = V$.

The idea of the verification is that instead of guessing the whole solutions,
we guess only the last bits (i.e.\ the parity of integer variables), verify this guess, simplify the equation and proceed:
for each variable $x$ we guess, whether it is even or odd
and appropriately replace it with $2x$ or $2x+1$.
Then we verify the guess by checking
whether both sides of each equality have the same parity.
If so, we divide each side of the inequality by $2$ (rounding down)
and proceed in the same fashion. Note, that in this way, the coefficients at each variable remain
the same over the whole procedure.

There is a little comment concerning the inequalities: all of them were initially $x \geq 1$,
and during the algorithms they can be also of the form $x \geq 0$.
When rounding down, we need to take care that rounding is done in an appropriate way,
for instance $2x \geq 1$ is in fact $2x \geq 2$ and so after dividing and rounding we should end up
with $x \geq 1$ again.
Observe that this boils down to replacing $x \geq 1$ by $x \geq 0$ if and only if $x$ is replaced by $2x+1$,
otherwise, the inequality remains as it were.

\begin{algorithm}[H]
  \caption{\algdiophantine{} Checks the satisfiability of a small linear Diophantine system \label{alg:small_satisfiability}}
  \begin{algorithmic}[1]
	\While{there is a non-zero constant \emph{or} an inequality $x \geq 1$}
  		\For{each variable $x$} \label{line change system}
  			\State guess $b_x \in \{0,1\}$
  			\State replace each $x$ with $2x+b_x$
  		\EndFor
  		\If{there is an equation with different parity of constants on the sides} \label{line system verification}
  			\State \Return Unsatisfiable
  		\EndIf
  		\State divide each equation by $2$, rounding down \label{line system rounding}
  		\State divide each inequality by $2$, round appropriately
  	\EndWhile
  	\State \Return Satisfiable \Comment{Has a trivial solution $(0, \ldots , 0)$}
 \end{algorithmic}
\end{algorithm}

Suppose that the small linear system has a solution $(q_1, q_2, \ldots , q_r)$.
We show that for some nondeterministic guesses, the obtained system has a solution
$(\lfloor q_1 /2 \rfloor, \lfloor q_2/2 \rfloor, \ldots ,\lfloor q_r/2 \rfloor)$.
Let the algorithm guesses the parity of $x$ according to $(q_1, q_2, \ldots , q_r)$.
Then after the loop in line~\ref{line change system} the obtained system has a solution
$(\lfloor q_1 /2 \rfloor, \lfloor q_2/2 \rfloor, \ldots ,\lfloor q_r/2 \rfloor)$.
Since each coefficient by the variable is even, the constants at side of each equation should be of the same parity,
and so the algorithm does not terminate in line~\ref{line system verification}.
Line~\ref{line system rounding} halves each equation. Observe, that an inequality $2x\geq 0$ is equivalent to $x \geq 0$
and inequalities $2x \geq 1$ in fact meant that $2x \geq 2$ and so they are also simply halved.

On the other hand, if $(q_1, q_2, \ldots , q_r)$ is a solution after the changes,
then $(2q_1+b_1, 2q_2+b_2, \ldots , 2q_r+b_r)$ was the solution of the system at the beginning
of the iteration.

Concerning the space usage of \algdiophantine,
the inequalities are simply stored as a bit for each variable (bit set to $i$ means that $x \geq i$).
When the start systems has size $\Ocomp(m)$, the intermediate ones have size $\Ocomp(m)$ as well (with a larger constant, though):
observe that the only new constants
(which are also stored in unary) are the $1$s from $2x+1$.
Suppose that initially the sum of constants was $c$ and the sum of coefficients at variables $m$.
We show by induction that the sum of constants during the algorithm is at most $\max(c,m)$.
This clearly holds in the beginning, let us investigate the changes in one round.
The sum of all $b_x$s introduced is at most $m$ so afterwards the sum of constants is at most $\max(c,m) + m$.
Each constant is halved in this round (rounding down), so their sum is at most
$(\max(c,m) + m)/2 \leq \max(c,m)$ at the end of the round, as claimed.
\qedhere
\end{proof}

\subsection*{Constructing a small Diophantine system from a word equation}
Lemma~\ref{lem:maximal block is crossing} suggests that the crucial
to consider, when dealing with maximal blocks, are the lengths of the
$a$ prefixes and $b$-suffixes of \sol X for different variables $X$.
We now show that this intuition can be formally stated and later show how to use this formulation
in a more efficient implementation of \algblocksc.

In order to perform the blocks compression in \algblocksc, we first guess the first ($a_X$) and last ($b_X$)  letter of \sol X for each $X$,
then the length of the $a_X$-prefix ($\ell_X$) and $b_X$-suffix ($r_X$) of \sol X, pop the $a_X$-prefix and $b_X$ suffix from $X$
and finally compress the maximal explicit blocks.
We now defer the guess of $\ell_X$ and $r_X$ for as long as possible, in fact, we shall not guess them at all.
Intuitively, we treat the lengths of the $a_X$-prefix and $b_X$-suffix of \sol X as \emph{parameters}
(or variables ranging over positive natural numbers) and to stress this we denote them by $x_X$ and $y_X$.
We can pop prefixes in this way, by replacing $X$ with $a_X^{x_X}Xb_X^{y_X}$ and even calculate the
lengths $e_1$, $e_2$, \ldots of explicit maximal blocks $E_1$, $E_2$, \ldots in \sol U and \sol V:
these are arithmetic expressions using constants and parameters $\{x_X,y_X \}_{X \in \variables}$.
In order to compress such maximal blocks, we guess which of them are equal,
and write the corresponding linear Diophantine equations.
If this system is feasible then we replace the maximal blocks:
$E_i$ and $E_j$ are replaced with the same letter if and only if $e_i$ and $e_j$ are declared to be equal.

This approach still requires that we know what is the first and last letter of each variable.
A \emph{prefix-suffix structure} for an equation $U = V$ tells for each $X$ occurring in the equation
what is its first (by convention: $a_X$) and last (by convention: $b_X$) letter
and whether $X$ is a block of one letter, i.e.\ whether $\sol X \in a_X^+$.

Given a prefix-suffix structure by $x_X$ and $y_X$ we denote the parameters (or variables)
that denote the lengths of the $a_X$-prefix and $b_X$-suffix of $X$
(if $X$ is a block of letters then by convention the $y_X$ is not used).
Given a prefix-suffix structure we can identify the visible maximal blocks and describe their lengths
(in terms of $\{x_X,y_X\}_{X \in \variables}$),
they are simply arithmetic expressions in $\{x_X,y_X\}_{X \in \variables}$.
To distinguish such blocks from `real' blocks, we call them \emph{parametrised visible maximal blocks}.
To distinguish them from maximal blocks, we denote the former by $\mathcal E_1, \ldots, \mathcal E_m$
while the latter by $E_1, \ldots, E_m$.

\begin{lemma}
\label{lem: small sides}
Consider a prefix-suffix structure for $U = V$.
Let $\mathcal E_1, \mathcal E_2, \ldots$ be parametrised visible maximal blocks of $U = V$ for this structure
and $e_1, e_2, \ldots$ be their lengths expressed in terms of $\{x_X,y_X \}_{X \in \variables}$ and constants.
Then $e_1, e_2, \ldots$ are a small set of linear Diophantine expressions in $\{x_X,y_X \}_{X \in \variables}$.
\end{lemma}
\begin{proof}
Consider a parametrised visible $\mathcal E_i$:
\begin{itemize}
	\item $\mathcal E_i$ may begin with either an explicit $a$
	or a maximal $a$-suffix of some $\solution(X)$
	(which may be whole $\solution(X)$);
	\item `in the middle' it may contain either explicit $a$s or
	$\solution(X) \in a^+$;
	\item it ends with an explicit $a$
	or a maximal $a$-prefix of some $\solution(X)$ (which may be whole $\solution(X)$).
\end{itemize}
Thus, whenever $\mathcal E_i$ is visible, $e_i$ is a linear combination of $x_X$, $y_X$ (where $X \in \variables$)
and natural numbers.

We show that the terms $e_1$, \ldots, $e_k$ are a small set
of linear Diophantine expressions.
For the purpose of the proof, denote by $n_X$ the number of times variable $X$
is used in the equation $U = V$.
We bound the number of times $x_X$ and $y_X$ occur in expressions
$e_1$, \ldots, $e_k$ and the size of additive constants used in $e_1$, \ldots, $e_k$:
\begin{itemize}
	\item each $x_X$ ($y_X$) occurs at most $n_X$ times, as for a fixed occurrence
	of variable $X$ there is at most one parametrised maximal block $\mathcal E_i$ that spans over the
	prefix (suffix, respectively) of this occurrence;
	\item the total size of used constants is $|U| + |V| - n_v$:
	for a fixed explicit occurrence of a letter $a$,
	there are is exactly one parametrised maximal block $\mathcal  E_i$ that spans over it.\qedhere
\end{itemize}
\end{proof}

Now let us explain the relation between solutions and prefix-suffix structures.
A solution \solution{} and prefix-suffix structure are \emph{coherent} if
\sol X indeed begins and ends with $a_X$ and $b_X$
and \sol X is a block of letters if and only if the prefix-suffix structure says so.
Furthermore, there is also a relation between lengths of maximal visible blocks of \solution{}
and parametrised visible blocks:
intuitively, when $\ell_X$ and $r_X$ are the lengths of the $a_X$-prefix and $b_X$-suffix of \sol X
then the length of the $i$-th visible maximal block is $e_i$ with $\{\ell_X,r_X \}_{X \in \variables}$
substituted for $\{x_X,y_X \}_{X \in \variables}$.
To make this more formal and shorter,
in the following, for an arithmetic expression $e$ in variables $\{x_X,y_X \}_{X \in \variables}$,
we use $e[\{\ell_X,r_X \}_{X \in \variables}]$ to denote the value of $e$ when $\{\ell_X,r_X \}_{X \in \variables}$
is substituted for $\{x_X,y_X \}_{X \in \variables}$.

\begin{lemma}
\label{lem: coherent solution}
Given a coherent prefix-suffix structure and a substitution \solution,
let $\mathcal E_1, \mathcal E_2, \ldots, \mathcal E_k$ be the parametrised visible maximal blocks for this structure,
and $e_1, e_2, \ldots, e_k$ their lengths,
while $E_1$, $E_2$, \ldots, $E_{k'}$ be the lengths of the visible maximal blocks for \solution.
Then $k' = k$ and for each $i$ the $\mathcal E_i$ and $E_i$ are blocks of the same letter
and $|E_i| = e_i[\{\ell_X, r_X\}_{X \in \variables}]$,
where $\ell_X$ and $r_X$ are the lengths of the $a_X$ prefix and $b_X$ suffix of \sol X.
\end{lemma}
\begin{proof}
Since the first and last letters of \sol X are the same as in the prefix-suffix structure
and \sol X is a block of letters if and only if prefix-suffix structure says so,
the $\mathcal E_i$ and $E_i$ consists of the corresponding explicit letters and prefixes/suffies of variables.
In particular, the number of parametrised blocks and blocks is the same,
for each $i$ the $\mathcal E_i$ and $E_i$ are blocks of the same letter and lastly, the length of $E_i$
corresponds to the length of $\mathcal E_i$ in which the values of parameters $x_X$ and $y_X$
are replaced by the actual lengths of prefixes and suffixes of $\sol X$;
this shows that $|E_i| = e_i[\{\ell_X, r_X\}_{X \in \variables}]$.
\qedhere
\end{proof}

So far we do not know, which parametrised blocks represent blocks of the same length.
To identify such blocks, we write a (small) linear Diophantine system that bounds
the $e_1$, \ldots, $e_k$ together:
we guess the partition of $e_1, e_2, \ldots, e_k$,
elements of one partition should correspond to parametrised blocks of the same length
(which in particular means that we assume that if $e_i$ and $e_j$ are in one part
then $E_i$ and $E_j$ are parametrised blocks of the same letter).
Then for each part $\{e_{i_1}, e_{i_2}, \ldots, e_{i_m}\}$ of the partition we write equations equalising the lengths:
\begin{equation}
\label{eq:equalising lengths}
e_{i_1} = e_{i_2}, \: e_{i_2} = e_{i_3}, \: \ldots, \: e_{i_{m-1}} = e_{i_m} \enspace .
\end{equation}
We also add the inequalities $x_X \geq 1$ (and $y_X \geq 1$) for every variable used in the equalities
(intuitively, since we claim that $\sol X$ begins or ends with a block of letters of length
$x_X$ or $y_X$, we want those blocks to consist of at least one letter).
If for some variable $X$ the \sol X is a block of letters, we use only $x_X$ in the equations, $y_X$ is not used
in the constructed system.
Thus we have obtained a linear Diophantine system in $x_X$ and $y_X$.
This is formalised in \wordtodiophantine.

\begin{algorithm}[h]
  \caption{$\wordtodiophantine$ Creates a system of equations for a prefix-suffix structure \label{alg:wtd}}
  \begin{algorithmic}[1]
	\Require prefix-suffix structure
	\For{$X \in \variables$}
		\If{$X$ represents a block of letters}	\Comment{According to the prefix-suffix structure}
			\State let $a_X$ be the first letter of $X$ \Comment{According to the prefix-suffix structure}
			\State introduce parameter $x_X$ \Comment{$\sol X = a_X^{x_X}$} \par
			\State add inequality $x_X \geq 1$ to $D$ \Comment{\sol X is non-trivial}
		\Else
			\State let $a_X$ and $b_X$ be the first and last letter of $X$
			\Comment{According to the prefix-suffix structure}
			\State introduce parameters $x_X$ and $y_X$ \Comment{Lengths of the of $a_X$-prefix and $b_X$-suffix of \sol X}
			\State add inequalities $x_X \geq 1$ and $y_X \geq 1$ to $D$ \Comment{The leading and ending blocks are non-trivial}
		\EndIf
	\EndFor
 	\State let $\{\mathcal E_1, \ldots, \mathcal E_k\}$ be the parametrised visible maximal blocks (read from left to right)
 	\For{each $\mathcal E_i$}
 		\State let $e_i \gets |\mathcal E_i|$
 		\Comment{Arithmetic expression in $\{x_X,y_X\}_{X \in \variables}$}
 	\EndFor                                                        
  	\State partition $\{\mathcal E_1, \ldots, \mathcal E_k\}$, 
		each part has only $a$-blocks for some $a$
  		\Comment{Guess} \label{make partition}
  	\For{each part $\{\mathcal E_{i_1}, \ldots, \mathcal E_{i_{k_p}}\}$}
  		\For{each $\mathcal E_{i_j} \in \{\mathcal E_{i_1}, \ldots, \mathcal E_{i_{k_p}}\}$}
				\State add an equation $e_{i_j} = e_{i_{j+1}}$ to $D$
				\Comment{Ignore the meaningless last equation}
			\EndFor
  	\EndFor
	\State \Return the partition, arithmetic expressions $e_1,\ldots,e_k$ and $D$.
	\end{algorithmic}
\end{algorithm}

\begin{lemma}
\label{lem:system is small}
The system of linear Diophantine equations and inequalities returned by $\wordtodiophantine$ is small.
\end{lemma}
\begin{proof}
The system of Diophantine linear equations is small if its sides form a small set of linear expressions
and each such an expression is used at most twice.
The sides are of this form by Lemma~\ref{lem: small sides} and each expression is used at most twice by the construction.
\qedhere
\end{proof}

It remains to link the constructed system to some solution of the word equation:
we say that \solution{} and a system $D$ constructed by $\wordtodiophantine$ are \emph{coherent}
(or simply, that $D$ is \solution-coherent),
if \solution{} is coherent with the prefix-suffix structure used by \wordtodiophantine{} to generate $D$
and the partition of parametrised visible maximal blocks
$\{\mathcal E_1, \ldots , \mathcal E_k\}$ in line~\ref{make partition} is done as in $\sol U = \sol V$,
i.e.\ $\mathcal E_i$ and $\mathcal E_j$ go into the same part if and only if the corresponding maximal blocks $E_i$ and $E_j$ of $\sol U = \sol V$
are equal.

\begin{lemma}
\label{lem: S compatible}
For a solution \solution{} of a word equation $U = V$ there is a unique \solution-coherent system $D$.
When $\ell_X$ and $r_X$ are the lengths of the $a_X$-prefix and $b_X$-suffix of \sol X
the $\{\ell_X, r_X\}_{X \in \variables}$ is a solution of $D$.
\end{lemma}
\begin{proof}
Concerning the existence and uniqueness:
in $\wordtodiophantine$ we simply make all the nondeterministic choices according to \solution.

To see that $\{\ell_X, r_X\}_{X \in \variables}$ is a solution of $D$:
observe that an equation $e_i = e_j$ is added only when $|E_i| = |E_j|$.
From Lemma~\ref{lem: coherent solution} we know that
$e_i[\{\ell_X, r_X\}_{X \in \variables}] = |E_i|$ and $e_j[\{\ell_X, r_X\}_{X \in \variables}] = |E_j|$,
hence $\{\ell_X, r_X\}_{X \in \variables}$ satisfies this equation and as $e_i = e_j$ was chosen arbitrarily,
we obtain that it satisfies $D$.
Note that all the inequalities are trivially satisfied.
\qedhere
\end{proof}

\subsection*{Improving \algblocksc}
We make the next step in the outlined strategy: after guessing a small system of Diophantine equations,
we verify its satisfiability and use it to perform the block compression.
To be more precise:
the \algblocksi{} firstly guesses the prefix-suffix structure, then uses $\wordtodiophantine$ to generate
a system of linear equations out of $U = V$, then it verifies its satisfiability.
Then it pops the prefixes and suffixes out of each variable, however, it does not guess the exact lengths,
but rather uses the prefix-suffix structure,
i.e.\ it pops $a^{x_X}$ to the left of $X$ and $b^{y_X}$ to the right
(of course, no popping to the right is done when $X$ is removed after the initial pop,
i.e.\ the prefix-suffix structure declares that \sol X is a block of letters).
Then we replace blocks of the same letter whose lengths are equalised in the system of the linear Diophantine equations by a fresh letter.

Note that in this way in the word equation (temporarily) we have symbols $a^x$, where $x$ is a variable with a value in natural numbers.
We are not going to give any semantics for that, as this is not needed,
but still we would like to consider maximal blocks of letters: the $a^x$ can be a part of a maximal $a$-block,
moreover, we assume that $x>0$, i.e.\ if the letter to the left of $a^x$ is $b \neq a$ and the same letter is to the right,
those $b$s are in different blocks.
We use the name \emph{parametrised explicit maximal blocks} with an obvious meaning.

As a first step, we begin with describing the improved version of \algprefsuff, the \algprefsuffi.
\begin{algorithm}[H]
  \caption{\algprefsuffi{} Cutting prefixes and suffixes, parametrised version \label{alg:prefixi}}
  \begin{algorithmic}[1]
	\Require prefix-suffix structure
  \For{$X \in \variables$}
		\State let $a_X$, $b_X$ be the first and last letter of \sol X \Comment{Given by the prefix-suffix structure}
		\If{$X$ is a block of letters} \Comment{According to the prefix-suffix structure}
			\State replace each $X$ in $U$ and $V$ by  $a^{x_X}$
		\Else
			\State replace each $X$ in $U$ and $V$ by  $a^{x_X} X b^{y_X}$
					\Comment{$x_X$, $y_X$ are variables}
			\If{$\sol X = \epsilon$} \Comment{Guess}
				\State remove $X$ from $U$ and $V$
			\EndIf
		\EndIf
  \EndFor
  \end{algorithmic}
\end{algorithm}

We are not going to state the exact properties of \algprefsuffi,
we shall give them collectively for \algblocksi, the improved version of \algblocksc.
For now we only note that during \algprefsuffi{} the visible parametrised blocks are changed into explicit ones.

\begin{lemma}
\label{lem: pref suff improved}
Let $\mathcal E_1$, $\mathcal E_2$, \ldots, $\mathcal E_m$ be the parametrised visible maximal blocks for a prefix-suffix structure.
Then after \algprefsuffi{} these are exactly the parametrised explicit maximal blocks.
\end{lemma}
\begin{proof}
The proof is obvious: whenever a prefix (suffix) of an occurrence of $X$ took part in some parametrised visible maximal block,
we popped this prefix (suffix) from $X$ and so now it is part of a corresponding parametrised explicit maximal block.
\qedhere
\end{proof}

Now we are ready to describe the improved version of \algblocksc{} as well as its properties.

\begin{algorithm}[h]
  \caption{$\algblocksi$  \label{alg:aci}}
  \begin{algorithmic}[1]
	\State guess the prefix-suffix structure
	\State run \wordtodiophantine
  \State run \algdiophantine{} on $D$ \Comment{Check if the guessed choices can be fulfilled}
	\State run \algprefsuffi \Comment{There are no crossing blocks} \par

  	\State let $\mathcal E = \{ E_1, \ldots , E_k\}$ be the explicit maximal blocks \label{explicit maximal blocks}
		\par
		\Comment{Those are exactly the parametrised visible maximal blocks from \wordtodiophantine}
	 	\For{each $\mathcal E_i = \{E_{i_1}, \ldots, E_{i_{k_p}}\}$ returned by \wordtodiophantine}
  		\State let $a_{e_{i_1}} \in \letters$ be an unused letter
  		\For{each $E_{i_j} \in \mathcal E_i$}
				\State replace every $E_{i_j}$ by $a_{e_{i_1}}$ 
			\EndFor
  	\EndFor
	\end{algorithmic}
\end{algorithm}

\begin{lemma}[cf.~Lemma~\ref{lem: consistent no crossing block}]
\label{lem:improved block compression}
\algblocksi{} is sound.
For a solution \solution{} of $U = V$ and the nondeterministic choices that lead to a creation of
an \solution-coherent system by \wordtodiophantine{} the \algblocksi{} implements the blocks compression;
to be more precise, the obtained word equation $U' = V'$ is identical (up to renaming the letters)
to the equation obtained by \algblocksc{} when it implements the block compression for \solution.
In particular, \algblocksi{} is complete.

\algblocksi{} uses a constant time more memory than the equation $U = V$,
in particular, the additional memory usage of \algsolveeq{} when using \algblocksi{} is linear.
\end{lemma}
\begin{proof}
Suppose that $\algblocksi$ applied on $U = V$ created a linear Diophantine system $D$
that has a solution $\{\ell_X, r_X\}_{X \in \variables}$.
Then we can think of $\algblocksi$ as if it replaced each $X$ with $a^{\ell_X}Xb^{r_X}$
and then replaced some blocks of the same letter and the same length with fresh letters.
Thus by Lemma~\ref{lem:preserving unsatisfiability} it is sound.

Concerning completeness and the implementation of the block compression,
we use the fact that \algblocksc{} has both those properties (Lemma~\ref{lem: consistent no crossing block}).
Suppose that $U = V$ has a solution \solution{} and consider the satisfiable instance $U' = V'$ obtained by \algblocksc{} out of $U = V$
that has a solution $\solution'$ such that $\solution'(U')$ is obtained from \sol U by compressing blocks of letters
(by Lemma~\ref{lem: consistent no crossing block} we know that for some non-deterministic choices
indeed \algblocksc{} returns such an equation).
We show that the run of \algblocksi{} in which \wordtodiophantine{} returns the \solution-coherent
system $D$ returns $U' = V'$ (up to renaming letters), which will end the proof.
In the following, let $\ell_X$ and $r_X$ be the length of the prefix and suffix popped from $X$ by \algprefsuff,
by Lemma~\ref{lem: consistent no crossing block} we know that we can restrict ourselves to the case
when $\ell_X$ is the length of the $a_X$ prefix of \sol X and $r_X$ of the $b_X$ suffix of \sol X.

Concerning the corresponding \algblocksi,
consider the non-deterministic choices for which the \wordtodiophantine{}
returns a small Diophantine system that is \solution-coherent:
by Lemma~\ref{lem: S compatible} such a system exists and it is satisfiable.
Let $\mathcal E_1$, $\mathcal E_2$, \ldots, $\mathcal E_k$ be the consecutive parametrised visible maximal blocks in $U = V$
and $E_1$, $E_2$, \ldots, $E_k$ be the visible maximal blocks in $U  = V$ for \solution.
By Lemma~\ref{lem: S compatible} the $\mathcal E_i$ and $E_i$ are blocks of the same letter and $k = k'$
Consider, what happens with the former blocks when we apply \algprefsuffi:
they become the parametrised explicit maximal blocks, see Lemma~\ref{lem: pref suff improved}.
Similarly, the $E_1'$, $E_2'$, \ldots, $E_k'$ become explicit blocks when \algprefsuff{} is applied on them,
as we pop the $a_X$-prefix and $b_X$ suffix from each variable.
Now, $E_i$ and $E_j$ are replaced with the same letter by \algblocksi{}
if and only if $e_i$ and $e_j$ are equalised in $D$ (note that not necessarily $e_i = e_j$
is in $D$, but it contains equation that imply this, i.e.\ a sequence of equations $e_i =e_{i_1}$, $e_{i_1} = e_{i_2}$, \ldots, $e_{i_m} = e_j$).
By definition of the \solution-coherent system this happens if and only if $|E_i| = |E_j|$.
Hence $\mathcal E_i$ and $\mathcal E_j$ are replaced with the same letter by \algblocksi{} if and only if $E_i$ and $E_j$ are by \algblocksc.
Which ends the proof for the second claim.

Concerning the memory consumption, observe that by Lemma~\ref{eq:small system memory},
the linear Diophantine system $D$ can be encoded using only a constant more bits than the word equation
and the same space can be used to verify the satisfiability of the system.
All other operations can be easily implemented in the same memory bounds.
\qedhere
\end{proof}

\subsection*{Similar solutions}
Thanks to Lemma~\ref{lem: S compatible} we know that each solution \solution{} has a corresponding system of Diophantine equations
(the \solution-coherent one) and that the lengths $\{ \ell_X, r_X\}_{X \in \variables}$ of the $a$-prefixes and $b$-suffixes of \solution{}
are a solution of the \solution-coherent system.
Still, there are two questions: on one hand for a given system $D$ we know nothing about letters in \sol X
that are not in the $a_X$-prefix nor in the $b_X$-suffix of \sol X.
Moreover, other solutions of $D$ should also induce a solution of a word equation.
In this section we investigate the relations between all such induced solutions of the word equation.
Intuitively, different solutions $\solution'$ of $U = V$ induced in this way
differ from \solution{} in lengths of maximal blocks in $\solution'(U)$.

We say that two words $w$ and $w'$ are \emph{similar},
if $w = E_1E_2 \dots E_k$ and $w' = E_1'E_2' \dots E_k'$,
where for each $i$ the $E_i$ and $E_i'$ are non-empty blocks of the same letter, i.e.\ for some $a$ we have $E_i,E_i' \in a^+$,
and they are maximal blocks in $w$ and $w'$, respectively, i.e.\ $E_{i-1}$ and $E_{i+1}$ as well as $E_{i-1}'$ and $E_{i+1}'$
are blocks of some other letters.
Two substitutions \solution{} and $\solution'$ are \emph{similar}, if for every variable $X$ the \sol X and $\solution'(X)$ are similar.
Note that from the definition it follows that if \solution{} and $\solution'$ are similar
than they have the same coherent prefix-suffix structure.

If \solution{} and $\solution'$ are similar then also $\sol U$ and $\solution'(U)$ are.

\begin{lemma}
	\label{lem: crossing for similar}
Let \solution{} and $\solution'$ be similar solutions of a word equation $U = V$.
Then \sol U and $\solution'(U)$ are similar.

Consider representation of \sol U and $\solution'(U)$ as concatenation of maximal blocks
$E_1$, $E_2$, \ldots, $E_k$, and $E_1'$, $E_2'$, \ldots, $E_{k'}'$ respectively.
Then for each $i$ the $E_i$ is a crossing (visible) block if and only if $E_i'$ is.
\end{lemma}
\begin{proof}
Concerning the first claim:
since \solution{} and $\solution'$ are similar,
for each variable the \sol X and $\solution'(X)$ can be represented as $F_1\ldots F_m$ and $F_1'\ldots F_m'$,
where each $F$ and $F'$ are maximal blocks of letters and $F_i$ and $F_i'$ are blocks of the same letter.
Now each $E_i$ and $E_i'$ consist of corresponding explicit letters as well as corresponding blocks,
in particular, $E_i$ includes some $F_j$ from \sol X if and only if $E_i'$ includes some $F_j'$ from $\solution'(X)$.

Concerning the second claim, we show it for the visible case, the proof is the same for the crossing blocks.
By symmetry it is enough to show that when $E_i$ is visible (crossing) then also $E_i'$ is.
We use the same observation, as before: note that $E_i$ is visible, when it contains an explicit letter
or a leading (or ending) block $F_j$ of letters from some \sol X.
But then the same happens for $E_i'$ and $F_j'$.
\qedhere
\end{proof}

Now, given a solution \solution{} of a word equation $U = V$ and its \solution-coherent system $D$ of Diophantine equations
we shall define a class of solutions of $U = V$, all such solutions will be similar.
Each such a solution $\solution'$ is uniquely defined by one solution $\{\ell_X',r_X'\}_{X \in \variables}$ of $D$,
to stress it we  denote it by $\solution[\{\ell_X',r_X'\}_{X \in \variables}]$.
When $\{\ell_X,r_X\}_{X \in \variables}$ are the lengths of $a_X$-prefixes and $b_X$-suffixes of \sol X (for each $X$),
the construction shall guarantee that $\solution[\{\ell_X,r_X\}_{X \in \variables}] = \solution$.

Consider a variable $X$ and its representation as maximal blocks $F_1F_2 \ldots F_k$ of \sol X.
Since $\solution[\{\ell_X',r_X'\}_{X \in \variables}]$ is to be similar with \solution,
$\solution[\{\ell_X',r_X'\}_{X \in \variables}](X)$ is defined as $F_1'F_2' \ldots F_k'$, where $F_i$ and $F_i'$
are blocks of the same letter. It is left to define the lengths of $F_1'$, $F_2'$, \ldots, $F_k'$
with respect to $\{\ell_X',r_X'\}_{X \in \variables}$.
Let $e_1$, $e_2$, \ldots, $e_i$ be the length of the parametrised visible maximal blocks of the prefix-suffix structure
that is coherent with \solution.

Consider any solution $\{\ell_X', r_X'\}_{X \in \variables}$ of $D$ and blocks $F_i$ and $F_i'$ in \sol X and
$\solution[\{\ell_X',r_X'\}_{X \in \variables}](X)$, respectively.
There are three cases:
\begin{enumerate}[(L 1)]
	\item  \label{L 1} $F_i$ is a prefix of suffix of \sol X (and so also $F_i'$ is for $\solution'(X)$).
	Then the length of $F_i$ is $\ell_X$ when it is a prefix (or $r_X$ when it is a suffix)
	and we set the length of $F_i'$ to $\ell_X'$ (or $r_X$, respectively).
	\item  \label{L 2} $F_i$ is not the prefix nor the suffix but it has visible length 
	(so $F_i'$ is also not a prefix nor a suffix and has visible length, by Lemma~\ref{lem: crossing for similar});
	Let $E_{i'}$ be a visible block (in \sol U or \sol V) such that $|E_{i'}| = |F_i|$,
	by definition of a visible length such a block exists.
	By Lemma~\ref{lem: S compatible} we know that $|E_{i'}| = e_{i'}[\{\ell_X,r_X\}_{X \in \variables}]$
	and so we set $|F_i'|$ to $e_i[\{\ell_X',r_X'\}_{X \in \variables}]$.
	\item  \label{L 3} $F_i$ is not the prefix nor the suffix and does not have a visible length (so the same applies to $E_i'$).
	In this case we simply give $F_i'$ the same length as $F_i$.
\end{enumerate}
It remains to check the validity of the construction.

\begin{lemma}
\label{lem:diophantine solution word solution}
Given a solution \solution{} of a word equation $U = V$ and the \solution-coherent Diophantine system $D$,
for each solution $\{\ell_X',r_X'\}_{X \in \variables}$ the corresponding $\solution[\{\ell_X',r_X'\}_{X \in \variables}]$
is a solution of $U = V$, which is similar to \solution.

Furthermore, for any variable $X$ we can give an arithmetic expression $e_X$ in variables $\{x_X,y_X\}_{X \in \variables}$
such that $|\solution'(X)| = e_X [\{\ell_X',r_X'\}_{X \in \variables}]$ and $e_X$ depends on $x_X$ and $y_X$
(if the latter exists).
\end{lemma}
\begin{proof}
Let $E_1,\ldots, E_k$ be a representation of \sol U as a concatenation of maximal blocks
and $E_1',\ldots, E_k'$ a representation of $\solution'(U)$.
Since by Lemma~\ref{lem: crossing for similar} the $\solution'(U)$ and \sol U are similar,
to show that $\solution'$ is a solution of $U = V$ it is enough to show that $|E_i| = |E_j|$ then also $|E_i'| = |E_j'|$,
and the rest follows by a simple induction.

Consider a maximal block $E_i$.
There are three cases:
\begin{description}
	\item[visible]
	It is visible. Then by Lemma~\ref{lem: crossing for similar} also $E_i'$ is visible.
	Furthermore, by Lemma~\ref{lem: coherent solution} the length of $E_i'$ is $e_i[\{\ell_X',r_X'\}_{X \in \variables}]$.
	\item[invisible with visible length]
	It is invisible but has a visible length, so also $E_i'$ is invisible, by Lemma~\ref{lem: crossing for similar}.
	Then $E_i'$ is a block in some $\solution'(X)$ that is not a prefix not a suffix of $\solution'(X)$.
	Then by~\Lref{2} its length is $e_{i'}[\{\ell_X',r_X'\}_{X \in \variables}]$,
	where $E_{i'}$ is a visible block such that $|E_{i'}| = |E_i|$.
	In particular, in the previous case it was shown that $|E_{i'}'| = e_{i'}[\{\ell_X',r_X'\}_{X \in \variables}]$
	and so $|E_{i'}'| = |E_i'|$.
	\item[invisible length]
	It has an invisible length, in particular, it is invisible.
	Then by Lemma~\ref{lem: crossing for similar} also $E_i'$ is invisible and so by~\Lref{3} it has length	$|E_i|$
\end{description}

Now, consider some $E_i'$ and $E_j'$ that are blocks of the same letter and such that $|E_i| = |E_j|$.
There are two possibilities: $|E_i|$ is a visible length or it is an invisible length.
If it is an invisible length, then it was shown already that $|E_i'| = |E_i|$ and $|E_j'| = |E_j|$
and hence $|E_i'| = |E_j'|$ as claimed.
If it is a visible length, then it was already shown that $|E_i'| = e_{i'}[\{\ell_X',r_X'\}_{X \in \variables}]$
and $|E_j'| = e_{j'}[\{\ell_X',r_X'\}_{X \in \variables}]$,
where $e_{i'}$ and $e_{j'}$ are such that $|E_i| = |E_{i'}|$ and $|E_j| = |E_{j'}|$
(note that it might be that $i = i'$ or that $i \neq i'$ and similarly for $j$ and $j'$).
Then $|E_{i'}| = |E_{j'}|$ and so the equality $e_{i'} = e_{j'}$ follows from the system $D$
(i.e.\ there is a sequence of equations $e_{i'} = e_{i_1}, e_{i_1} = e_{i_2}, \ldots, e_{i_p} = e_{j'}$).
As $\{\ell_X',r_X'\}_{X \in \variables}$ is a solution of $D$,
we conclude that $e_{i'}[\{\ell_X',r_X'\}_{X \in \variables}] = e_{j'}[\{\ell_X',r_X'\}_{X \in \variables}]$,
and so $|E_i'| = |E_j'|$.

It is left to show the second claim, concerning the existence of an arithmetic expression for $e_X$.
This is obvious by the definition of $\solution'$:
let $\sol X = F_1F_2\ldots F_m$, where each $F_i$ is a maximal block.
Then $F_1$ has length $\ell_X$, $F_m$ length $r_X$ (so we add $x_X$ and $y_X$ to $e_X$),
when $F_i$ has invisible length then $F_i'$ has length $|F_i|$ (so we add a constant $|F_i|$ to $e_X$)
and if it has a visible length, then the length is expressed as some $e_{i}[\{\ell_X',r_X'\}_{X \in \variables}]$
(so we add $e_i$ to $e_X$).
In the end, $e_X$ is the sum of all such arithmetic expressions for $|F_1|$, $|F_2|$, \ldots, $|F_m|$.
\qedhere
\end{proof}

\section{Linear space for $\Ocomp(1)$ variables}
\label{sec:linear space}
\subsection*{Idea}
As already shown, the length of the word equation kept by \algsolveeq{} is linear, see Lemma~\ref{lem: space consumption}
and the additional space consumption of \algsolveeq{} is proportional to the storage size of the current equation,
see Lemma~\ref{lem:improved block compression}.
However, the letters in this equation can be all different, even if the input equation
is over two letters.
Hence the upper bound on the space usage that we can give is (nondeterministic) $\Ocomp(n \log n)$ bits.
We would like to improve the space consumption to linear;
to be more precise, we would like the space consumption to be $\Ocomp(m)$ bits,
where the input equation used $m$ bits in a natural encoding.%
\footnote{The proofs given in this section work assuming that each occurrence of a letter (variable) in the input
is always given using the same bit representation, however, it is \emph{not} assumed that all letters
and/or variables have the representations of the same length,
in particular the presented method works also when the input equation is compressed using Huffman coding.
}
We fail in a general case, such a bound is shown only for $\Ocomp(1)$ variables
(although it holds for arbitrary many occurrences of these variables in the equation, i.e.\ $n_v$ is not bounded
and the alphabet size is arbitrary).

The main obstacle is the encoding of letters introduced by \algsolveeq.
We show that when we look at the computation of \algsolveeq{} that do not remove
the variables from the equation, the space consumption can be limited to $\Ocomp(m)$,
where $m$ is the storage size (calculated in bits) of the equation at the beginning of the stage.
Then for $k = \Ocomp(1)$ variables we can consider $k$ stages of \algsolveeq,
a stage ends when a variable is removed from the equation.
In this way the space consumption will be estimated by $c^k m$ bits, which is linear for a constant $k$.

\subsubsection*{Encoding of letters}
Consider string of explicit letters between two consecutive variables
$X$ and $Y$ in $U = V$, together with the variables.
During \algsolveeq{} the $XwY$ will be changed to $Xw^{(1)}Y$, $Xw^{(2)}Y$, \ldots.
Observe, that each $w^{(i)}$ can be partitioned into $3$ substrings $x^{(i)}v^{(i)}y^{(i)}$,
where the letters in $v^{(i)}$ represent solely the letters from $w$, while
each letter in $x^{(i)}$ ($y^{(i)}$) represent also some letter popped at some point from $X$
($Y$, respectively).
It is easy to encode $v^{(i)}$ using only a constant time more bits than $w$:
we represent letters as trees and when merging $a$ and $b$ into $c$, the tree
of $c$ has the tree of $a$ as a left subtree and a tree of $b$ as a right subtree;
using any usual encoding the size of such representation is only constant times larger
than the original text $w$.

On the other hand, the letters in $x^{(i)}$ and $y^{(i)}$ depend solely on $XwY$, so we simply
encode them as $(XwY)1$, $(XwY)2$, \ldots, $(XwY)(|x^{(i)}|+|y^{(i)}|)$,
where `$(XwY)$' is encoded exactly as it was in the input equation while the following numbers are encoded in binary.
Note that the same code `$(XwY)1$' is (usually) used in each phase, but it denotes different letters in the respective phases.

\subsubsection*{Compressing all pairs}
In this way different occurrences of the same letter $a$ may get different codes:
in such case we collect the codes for $a$ and add the information that they all represent the same letter.

However, this approach raises a new concern: it might be that the length $|x^{(i)}|+|y^{(i)}|$ is non-constant:
\algsolveeq{} only guarantees that the length of the whole $|U| + |V|$ is $\Ocomp(n)$,
but some fragments (i.e.\ explicit words between variables) may become large.
However, for $\Ocomp(1)$ variables this can be solved easily, as we can enforce that in one phase
\emph{each} pair of consecutive letters is compressed: firstly, a simple preprocessing (to be precise, $\algpop(\letters,\letters)$)
ensures that there are only $\Ocomp(k)$ crossing pairs, where $k$ is a number of variables.
Then non-crossing pairs are compressed separately (not causing any increase of size of the kept word equations)
and each of the crossing pair $ab$ is compressed using $\algpairc(\{a\},\{b\})$.

\begin{algorithm}[H]
	\caption{\algsolveeqlin{} Checking the satisfiability of a word equation in linear space for $\Ocomp(1)$ variables}
	\label{alg:main linear}
	\begin{algorithmic}[1]
	\While{$|U| > 1$ or $|V| > 1$}
  		\State \algblocksi{} \label{block compression improved} \Comment{Block compression}
  		\State $\algpop(\letters,\letters)$ \Comment{The number of crossing pairs is $\Ocomp(k)$} \label{pop all}
  		\State $P \gets $ list of non-crossing pairs \label{guess noncrossing pairs} \Comment{Guess}
  		\State $P' \gets $ list of crossing pairs \label{guess crossing pairs} \Comment{Guess, at most $2k$ pairs}
  		\For{$ab \in P$}
  			\State run $\algpair(a,b)$
  		\EndFor
  		 \For{$ab \in P'$} \Comment{$P' \leq 2 k$}
  			\State $\algpairc(\{a\},\{b\})$ 
  		\EndFor
	\EndWhile
	\State Solve the problem naively
		\Comment{With sides of length $1$, the problem is trivial}
 \end{algorithmic}
\end{algorithm}

The properties of \algsolveeqlin{} are summarised in the below theorem.

\begin{theorem}
\label{thm:constant number of variables}
\algsolveeqlin{} is sound and complete.
For $k$ variables, it runs in (nondeterministic) space of $\Ocomp(m k^{ck})$ bits, for some constant $c$,
where $m$ is the space consumption (measured in bits) of the input word equation.
\end{theorem}

For the input equation $U = V$ define consecutive \emph{stages}:
a stage ends immediately when one variable is removed
from the kept equation. Then the next stage starts instantly afterwards.
In this way there are at most $k+1$ stages.

We begin with showing the correctness of \algsolveeqlin, the proof is a slight modification
of the proof of correctness of \algsolveeq, see Lemma~\ref{lem:correctness of main}.

\begin{lemma}
\label{lem:linear correctness}
\algsolveeqlin{} is sound and complete.
The kept equation has length $\Ocomp(k n)$ in one stage,
where equation at the beginning of the stage has length $n$.
\end{lemma}
\begin{proof}
All subprocedures in \algsolveeqlin{} are known to be sound and complete,
note that a proper guess of noncrossing pairs in line~\ref{guess noncrossing pairs} is needed,
as $\algpair(a,b)$ is complete only for a noncrossing pair $ab$.
Observe that if $a'b'$ is another noncrossing pair to be compressed,
then after $\algpair(a,b)$, when $\solution'$ is a solution of $U' = V'$ which implements the pair compression for $ab$,
the pair $a'b'$ is noncrossing with respect to $\solution'$,
as none of the first/last letter of any \sol X can be $b'$/$a'$
So also \algsolveeqlin{} is sound and complete.

Concerning the space consumption: since we try to compress each crossing pair,
a stronger version of Claim~\ref{clm:words are compressed} can be shown:

\begin{clm}[cf.~Claim~\ref{clm:words are compressed}]
\label{clm:words are better compressed}
Let $U = V$ has a solution \solution.
For appropriate choices, the equation $U' = V'$ obtained at the end one stage of \algsolveeqlin{}
has a solution $\solution'$ such that
\begin{itemize}
	\item for each pair of two consecutive letters in $U$ (or $V$), one of these letters is compressed in $U'$ (or $V'$, respectively);
	\item for each pair of two consecutive letters in \sol U, one of these letters is
	compressed in $\solution'(U')$.
\end{itemize}
\end{clm}
\begin{proof}
Consider any two consecutive letters $ab$. If $a=b$ then they are compressed by \algblocksi.
If they are not and one of the letters is compressed in \algblocksi{} then we are done.
Otherwise, $ab$ will be either in $P$ or in $P'$ and we try to compress it.
We fail only if one of those two letters was already compressed.
\qedhere.
\end{proof}

To show the bound on the length of the kept equation,
we first estimate that the number of crossing pairs is indeed $\Ocomp(k)$.
Observe that after $\algpop(\letters,\letters)$ in line~\ref{pop all} each occurrence of a variable $X$
is preceded (succeeded) by the same letter, say $a_X$ ($b_X$, respectively).
When $b$ ($a$) is the first (last, respectively) letter of \sol X, the $X$ brings only two crossing pairs $a_Xb$ and $b_Xa$.
As there are $k$ different variables, there are at most $2k$ different crossing pairs.

Using a similar argument as in Lemma~\ref{lem: space consumption},
it can be shown that the length of the kept equation is $\Ocomp(nk)$,
as $\algpop$ is run $k+1$ times and \algblocksi{} once in one stage and each such run introduces at most $\Ocomp(n)$ letters.
\qedhere
\end{proof}

\subsection*{Occurrences of letters}
We distinguish two types of occurrences of explicit letters in $U = V$ in one stage:
\emph{inner} and \emph{outer} occurrences; note, that the same letter $a$ may have at the same time both
an inner and an outer occurrence.
Each explicit letter at the beginning of the stage is inner, each letter popped from a variable is outer.
When we compress two (or more) inner letters, the result is an inner letter; otherwise the letter that replaced some string is outer.
Observe that this implies that each substring $w$ between two variables, say  $X$ and $Y$, can be partitioned
into $w = x v y$, where $x,y$ consist solely of outer letters and $v$ consists solely of inner letters (each of $x$, $v$, $y$ may be empty).
As already noted, the same letter may be encoded in several different ways, this is not a problem, we separately keep
a list of different representations of the same letter. Note that this increases the space consumption by a constant.

\subsubsection*{Inner letters}
The inner letters are encoded as follows: when compressing two (or more) inner letters represented as $a$ and $b$
we represent them as $(a,b)$, where `$($', `$,$' and `$)$' are some appropriately coded symbols;
we can think of this as a flattened tree.
Note, that
in this way when a string of input symbols is compressed into string $w'$, then $w'$ uses only constant time more bits than $w$.

\begin{lemma}
\label{lem: inner consumption}
The space used for encoding of the inner letters is $\Ocomp(m)$, where $m$ is the space (in bits)
used for the encoding of the equation at the beginning of the stage.
\end{lemma}
The proof is obvious from the above definition.

\subsubsection*{Outer letters}
The outer letters are encoded in a different way:
note that if $XwY$ has two different occurrences in $U = V$ then in both of them the outer letters (and inner ones)
will be equal in one stage, and so can be encoded using the same symbols.

\begin{lemma}
\label{lem: context determines letters}
Let $XwY$ has two different occurrences in $U = V$ at the beginning of the stage.
Then in this stage both those occurrences are represented using the same strings.
\end{lemma}
\begin{proof}
Observe that the letters popped from a variable depend only on the variable and not on the surrounding letters.
Then the string $w$ between those two variables is transformed exactly in the same way in both occurrences.
\qedhere
\end{proof}

We want to encode the outer letters
occurring in the string representing $XwY$ as $(XwY)\#\text{(\textit{letter number})}$,
where `$(XwY)$' is encoded as in was in the equation at the beginning of the stage
and the following `letter number' is encoded in binary.
Lemma~\ref{lem: context determines letters} guarantees that such representations used for different occurrences of $XwY$
are the same, however, we still do not know, how many different such letters are needed.
The following lemma shows that $|x|$ and $|y|$ are linear in $k$,
which guarantees that numbers used to denote `letter number' are also linear in $k$.

\begin{lemma}
\label{lem: k outer letters}
In one stage, at the beginning of the phase, the maximal substring of outer letters has length $\Ocomp(k)$.
Furthermore, the space used for the encoding of outer letters in a stage is $\Ocomp(k m)$, where $m$ is the size of
the representation of the equation at the beginning of the stage.
\end{lemma}
\begin{proof}
As there are $k+1$ application of \algpop, the length of such block increases by at most $2k+2$
(it may be expanded from both ends if $v$ is empty).
On the other hand, by Claim~\ref{clm:words are better compressed} each substring of length $4$ is replaced by a substring of length $3$ or less
in one phase of \algsolveeqlin.
This applies to the substrings of outer letters, and similarly as in~\eqref{eq:equation length}
it can be shown that these substrings have length $\Ocomp(k)$.

As only $\Ocomp(k)$ number of different letters per 
$XwY$ is encoded as outer letters, and each occurrence of a letter encoded as $(XwY)i$ can be charged to an occurrence of $XwY$
at the beginning of the stage, so the space consumption can be bounded as a $\Ocomp(k)$ times the consumption at the beginning of the stage.
\qedhere
\end{proof}

Now the proof of Theorem~\ref{thm:constant number of variables} follows easily.

\begin{proof}[proof of Theorem~\ref{thm:constant number of variables}]
Since the number of different variables is $k$, there are at most $k$ stages.
Note that during one stage the space consumption increases at most $c k$ times, where $c$ does not depend on $k$, nor $n$,
see Lemma~\ref{lem: inner consumption} and \ref{lem: k outer letters}.
Thus, the total space consumption is at most $(kc)^k$ times greater than the one of the input equation.

The correctness follows from Lemma~\ref{lem:linear correctness}.
\qedhere
\end{proof}

\section{Solutions other than length minimal}
\label{sec:solutions}
In the next section we give an algorithm generating a (finite) representation of all solutions
of a~word equation. However, so far we have considered mainly the length minimal solutions,
and clearly there are other ones.
In this section we recall the classification of solutions, taken from work of Plandowski~\cite{PlandowskiSTOC2}.
The main result of this classification is the identification of the \emph{minimal solutions},
which have all properties of the length-minimal solutions that we use, except the exponential bound
on the exponent of periodicity;
however \algblocksi{} eliminates the need for this bound, which suggest that this bound is not essential,
at least for checking the validity of a word equation.
The solutions (substitutions) are classified not by their length,
instead we consider whether one solution can be obtained from another using homomorphisms.
If so then the former solution is clearly `more complicated' than the latter.

We first extend the notion of the solution, so that it can include letters that do not occur in the equation:
By $\letters'$ we denote the letters that can occur in the solution,
even though they do not occur in the equation;
formally $\letters'$ is an arbitrary set such that
$\letters' \cap \letters = \emptyset$
(and of course $\letters' \cap \variables = \emptyset$).
Then \emph{substitution} is a morphism $\solution : \variables \cup \letters \mapsto (\letters \cup \letters')^+$
that satisfies the previous assumption that $\sol a = a$ for every $a \in \letters$;
a notion of the solution generalises to this setting.
We call $\letters'$ \emph{free letters} of the solution.

We use the name \emph{operator} to denote functions
transforming substitutions. A special class of operators is particularly
important for us: given a morphism
$\phi : \letters \cup \letters' \mapsto (\letters \cup \letters')^+$
by $\Phi$ (so capitalised $\phi$) we denote a corresponding morphism that acts on
substitutions, changing \sol X by $\phi$, to be precise
$\Phi[\solution](X) = \phi(\sol X)$ and $\Phi[\solution](a) = a$ for $a \in \letters \cup \letters'$.
For composition of operators we use the usual symbol $\circ$,
however, when indexed composition is used, we denote it by $\prod$,
for lack of a better symbol.

\begin{definition}[cf.~\cite{PlandowskiSTOC2}]
A solution $\solution: \variables \cup \letters \mapsto (\letters \cup \letters')^+$
of an equation $U = V$ is a \emph{unifier} (with \emph{free letters} $\letters'$),
when \sol U contains at least one letter from $\letters'$.
$\solution'$ is an \emph{instance}
of a unifier solution \solution, if $\solution' = \Phi[\solution]$ for some non-erasing
non-permutating%
\footnote{A morphism $\phi$ is \emph{non-erasing} if $\phi(a) \neq \epsilon$ for every letter $a$
and it is \emph{non-permutating} if $\phi$ is not a permutation on its domain.}
morphism $\phi: (\Gamma \cup \Gamma') \mapsto (\Gamma \cup \Gamma')^+$
that is constant on \letters.
A solution \solution{} is \emph{minimal},
if it is not a unifier solution, nor an instance of a unifier solution;
it is a \emph{minimal unifier} if it is a unifier solution and it is not an instance of another unifier solution.
\end{definition}

The assumption that the instance of a unifier solution
is obtained by a non-erasing morphism is technical,
but it ensures easier and cleaner classification of minimal solutions.
We forbid the homomorphism to be a permutation, as we do not want
that a solution is its own instance.
It is easy to observe that as $\letters' \cap \letters = \emptyset$,
every instance $\solution'$ of a unifier solution \solution{}
is a solution (perhaps a unifier one).
Note that in general a satisfiable word equation may have no minimal solutions or no minimal unifier solutions.

\begin{example}
Consider an equation $aXb = Y$. Then each $\sol X = w$ and $\sol Y = awb$ is a solution.
Then $\sol X = w \in \Gamma$ is length-minimal;
when $w$ contains a free letter, then \solution{} is a unifier solution,
when additionally $w \in \letters'$ then this is a minimal unifier solution.
There are no minimal solutions.

Consider an equation $aX = Xa$, then each $\sol X = a^n$ is a minimal solution,
$\sol X = a$ is a length-minimal one; there are no unifier solutions.

Consider an equation $aXYX^3 = XYaY^2$. Since $\sol{aXY}$ and $\sol{XYa}$
have always the same length, this is equivalent to a system of
equations $aXY=XYa$ and $X^3=Y^2$. The former has solutions
$X = a^n, Y = a^m$ and the latter ensures that $3n = 2m$.
All such solutions are minimal and $\sol X = a^2$, $\sol Y = a^3$ is length-minimal.
There are no other solutions, in particular, no unifier solutions.
\end{example}

\subsection*{Typical operators}
While in the definition of minimal solutions the operator $\Phi$ corresponding to a~morphism $\phi$
is arbitrary, in the proofs we usually see morphisms that are related to pair compression and blocks compression.
By \morph c {ab} denote the morphism which replaces $c$ by $ab$ and is constant on all other letters,
the \invmorph c {ab} is the corresponding inverse morphism (note, that when $a \neq b $ the inverse is well-defined);
by \blockc{} denote the morphism which, for each $\ell$, replaces $a_\ell$ by $a^\ell$.
Since a block of $a$ can have various partition into subblocks of $a$s,
\invblockc{} is not well defined.
For the purpose of this paper, we specify its action as follows:
\invblockc{} replaces each $a$'s maximal $\ell$-block by a letter $a_\ell$.
The \Morph c {ab} and \Blockc{} denote the corresponding operators,
\Invmorph c {ab} the inverse operator and \Invblockc{} the inverse mapping.

\subsection*{Properties of minimal solutions}
As already noted, the minimal solutions inherit most of the crucial properties
of length-minimal solutions. In particular, a variant of Lemma~\ref{lem:over a cut}
holds for them.
\begin{lemma}[{cf.~\cite[Lemma~6]{PlandowskiICALP}}, cf.~Lemma~\ref{lem:over a cut}]
\label{lem:over a cut improved}
Let \solution{} be a minimal solution of $U = V$.
\begin{itemize}
	\item If $ab$ is a substring of \sol U, where $a \neq b$,
	then $ab$ is an explicit pair or a crossing pair.
	\item If $a^k$ is a maximal block in \sol U
	then $a$ has an explicit occurrence in $U$ or $V$
	and there is a visible occurrence of $a^k$.\
	\item If $a^k$ is a maximal block in \sol U and \solution{} has no crossing $a$ blocks
	then $a$ has an explicit occurrence in $U$ or $V$.
\end{itemize}
\end{lemma}
\begin{proof}
The first claim, which regards a pair $ab$, is shown using the following fact
(we do not assume that \solution{} is minimal,
as we reuse Claim~\ref{clm:each pair is crossing} later on in this more general setting):
\begin{clm}
\label{clm:each pair is crossing}
If $ab$, where $a \neq b$, is not an explicit nor a crossing pair for a solution \solution{} for $U = V$,
then $\solution ' = \Invmorph c {ab} [\solution]$ for a free letter $c \in \letters'$
is a unifier solution of $U = V$.
In particular, $\solution = \Morph c {ab} [\solution']$ is an instance of $\solution'$
and so it is not minimal.
\end{clm}
\begin{proof}
Consider $\solution ' = \Invmorph c {ab} [\solution]$.
Since $ab$ is not an explicit nor a crossing pair,
each occurrence of $ab$ in \sol U (and \sol V) comes from \sol X for some variable $X$.
Thus $\solution'(U)$ is obtained from $\sol U$ be replacing each $ab$ by $c$.
The same applies to \sol V and $\solution'(V)$ as well,
consequently $\solution'$ is a solution of $U = V$. Formally:
$$
\solution ' (U) = (\Invmorph c {ab} [\solution])(U) = \invmorph c {ab} (\sol U )
= \invmorph c {ab} (\sol V ) = (\Invmorph c {ab} [\solution])(V) = \solution ' (V).
$$
Since $c$ is a free letter, $\solution'$ is a unifier solution.
Furthermore, as $c$ does not occur in $\sol X$ for any $X$,
then $(\Morph c {ab}  \circ \Invmorph c {ab})[\solution] = \solution$:
indeed, the $\Invmorph c {ab}$ replaces each $ab$ by $c$ in every $\solution'(X)$, while
$\Morph c {ab}$ replaces each $c$ by $ab$ in every \sol X.
Hence, $\solution = \Morph c {ab} [\solution']$ and as $\morph c {ab}$
is non-erasing, non-permutating and constant on $\letters$,
we conclude that \solution{} is an instance of $\solution'$,
which contradicts the assumption that \solution{} is minimal.
\qedhere
\end{proof}

Now the first claim of the lemma follows by a contraposition of Claim~\ref{clm:each pair is crossing}.

Consider now the second claim, which regards the maximal blocks of $a$.
Observe that if $a$ occurs in \sol U and it does not occur in $U$, nor $V$,
then it is a letter from $\letters'$ and so, by definition, \solution{} is a unifier
solution and thus cannot be a minimal solution, hence $a$ occurs in $U$ or in $V$.

To streamline the presentation and analysis,
in the remainder of the proof assume that both $U$ and $V$ begin
and end with a letter and not a variable;
this is easy to achieve by prepending $\$$ and appending $\$'$
to both sides of the equation.
Alternatively, the cases with variables beginning or ending $U$ or $V$ can be handled
in the same way, as the general case.

Consider a maximal $a$ block $a^\ell$, for $\ell > 0$ in \sol U
and the letter preceding (succeeding) it, say $b$ and $c$, respectively;
by the assumption that $U$ and $V$ begin and end with a letter, such $b$ and $c$
always exist.
Consider the occurrences of $ba^\ell c$ in \sol U and \sol V.
Since $b \neq a \neq c$, these occurrences cannot have overlapping $a$'s
(though, if $b=c$, these letters can overlap for different occurrences).
Suppose that none of these occurrences is crossing or explicit.
Then for each of such occurrences there is a variable $Y$ such that $ba^\ell c$
is wholly contained within some occurrence of \sol Y.
Change the solution $\solution$ into $\solution'$,
by replacing each $ba^\ell c$ in each $\sol Y$ by $bvc$ for a free letter $v$;
since $a^\ell$ in various occurrences of $ba^\ell c$ do not overlap,
such replacement is well-defined.
Then $\solution'$ is still a solution,
in fact, a unifier solution.
Furthermore, \solution{} is its instance, contradiction.
Hence, there is an explicit or a crossing occurrence (with respect to \solution)
of $ba^\ell c$.
Then this occurrence restricted to $a^\ell$ satisfies the claim of the lemma.

Consider now the last, third claim. Suppose that $a^\ell$ occurs in \sol U and there is not explicit occurrence of maximal $a^\ell$
in $U = V$. By the case assumption there is also no crossing occurrence, so all occurrences of maximal $a^\ell$ are in fact implicit.
Construct a new solution $\solution'$ obtained by replacing each maximal $a^\ell$ by a free letter $x$.
Note that $\solution'$ is a unifier solution of $U = V$ and \solution{} is its instance,
so \solution{} was not minimal, contradiction.
\qedhere
\end{proof}

\subsection*{Minimal unifier solutions}
It is already known from the work of Plandowski~\cite[Lemma~1]{PlandowskiSTOC2}
that finding minimal unifier solutions reduces to finding minimal solutions.
This is a consequence of the following lemma,
which is strengthening of Lemma~\ref{lem:over a cut improved}
for minimal unifier solutions.
\begin{lemma}[cf.~{\cite[Lemma~1]{PlandowskiSTOC2}}]
\label{lem:unifier reduces}
If $S$ is a minimal unifier solution with a free letter $v \in \letters'$,
then for some variables $X$ and $Y$ it holds that $v$ is the first letter of \sol X and $v$ is the last letter of \sol Y.
\end{lemma}
\begin{proof}
The proof is similar to the proof of Lemma~\ref{lem:over a cut improved}.
Suppose that $v$ is not a last letter for any \sol X.
Consider any occurrence of $v$ in \sol U, and let $a$ be some letter
directly to the right of one of $v$'s occurrences;
such a letter exists as $v$ has no occurrence in the equation and is not a last letter in any \sol X.
The pair $va$ is non-crossing for \solution{} (by the assumption) and so by Claim~\ref{clm:each pair is crossing}
we obtain that $\solution = \Morph b {va} [\solution']$ for some fresh letter $b$
and a unifier solution $\solution'$.
To conclude that $\solution$ is not minimal, it is left to show that
$\morph b {va}$ is non-erasing (obvious), non-permutating
(also true, as $\morph b {va}(b) = va \notin \letters \cup \letters'$)
and constant on $\letters$ (true, as $b \notin \letters$).

Symmetric argument can be given, when $v$ is not a first of some \sol X.
\qedhere
\end{proof}
Intuitively Lemma~\ref{lem:unifier reduces} yields
that search for minimal unifier solutions
reduces to looking for minimal solutions:
it is enough to `left-pop' a letter from each variable, in this way the free letters
are introduced into the equation and become standard letters.
With appropriate nondeterministic guesses,
the unifier solutions of $U = V$ will correspond to non-unifier solutions of $U' = V'$.
The precise statement needs some additional definitions, which are introduced in the next section; 
thus the formal statement is deferred to the following section, see Lemma~\ref{lem:unifier to minimal}.

\section{Representation of all solutions}
\label{sec:generator}
The first \PSPACE{} algorithm for verifying the satisfiability of word equations,
\algPlandowskiFOCS, was extended by its author to \algPlandowskiSTOCdwa,
which returns a finite, graph-like,
representation of all finite solutions of a given word equation~\cite{PlandowskiSTOC2}.
This extension is done in two stages: firstly the original \algPlandowskiFOCS{} is
modified into another algorithm
\algPlandowskiSTOCi, which also only verifies satisfiability of word equations;
then \algPlandowskiSTOCdwa{} uses \algPlandowskiSTOCi{} as a subprocedure
in generation of a graph representation of \emph{all} finite solutions of a~word equation.
The modification into \algPlandowskiSTOCi{} is nontrivial,
and its correctness required a separate, involved proof.
In this section we show that \algsolveeq{} that uses \algblocksi{} instead of \algblocksc{}
also can be used to generate a (similar) representation of all solutions
of a given word equation.

\subsubsection*{Representing all solutions}
We want to use \algsolveeq, which still only verifies satisfiability,
as a~subprocedure for an algorithm generating a (finite) description of all finite solutions.
The approach is simple and in fact is similar to
the earlier approach used by Plandowski~\cite{PlandowskiSTOC2}:
the representation is modelled by a graph,
with nodes labelled with equations $U = V$ that are considered
by \algsolveeq{} and edges representing transformation performed by \algsolveeq.
To be precise, if an equation $U = V$ is transformed into $U' = V'$ by \algsolveeq{}
(for some nondeterministic choices)
we put an edge between nodes labelled by these two equations
and label it with an operator that transforms solutions of $U' = V'$
into solutions of $U = V$;
furthermore, each solution of $U = V$ can be represented in this way
(perhaps by transformation of a solution of some other equation $U'' = V''$,
which is obtained from $U = V$ for different non-deterministic choices);
note that we do not guarantee that there is a unique way to represent \solution{} in such a way.
Also, nodes with trivial equations (i.e.\ $|U|= |V| =1$) have only one, easy to define, minimal solution
(or no solution at all).
Concerning the space consumption, since \algsolveeq{} runs in \PSPACE,
such generation of labelled vertices and edges can also be performed in \PSPACE.

\begin{theorem}[cf.~\cite{PlandowskiSTOC2}]
\label{thm:generation}
The graph representation of all minimal solutions and minimal unifier solutions
of an equation $U = V$ can be constructed in \PSPACE.
The size of the constructed graph is at most exponential.
\end{theorem}

As already noted, representing all minimal unifier solution can be reduced to
representing all minimal solutions, which will be formally stated in Lemma~\ref{lem:unifier to minimal}.
Thus, in the following, we focus on the representation of minimal solutions of a given word equation.
We begin with the description of operators that are used to transform the solutions.

\subsection*{Transforming solutions and inverse operators}
So far we only know that \algsolveeq{} is sound and complete,
however, we do not really know what happens with particular solutions:
we do not know how to obtain the solution of the original equation $U = V$
from the transformed equation $U' = V'$, even worse
it might be that many of them are somehow lost in the translation.
To describe the correspondence of solutions,
we strengthen the notions of soundness and completeness
(so that they resemble more the notions of implementing the pair compression and block compression).

Given a (nondeterministic) procedure transforming the equation $U = V$
we say that this procedure \emph{transforms the minimal solutions},
if based on the nondeterministic choices and the input equation we can define a family of operators $\mathcal H$
such that
\begin{itemize}
	\item for any minimal solution \solution{} of $U = V$ there are some nondeterministic choices
	that lead to an equation $U' = V'$
	such that $\solution = H[\solution']$ for some minimal solution $\solution'$
	of the equation $U' = V'$ and some operator $H \in \mathcal H$;
	\item for every equation $U' = V'$ that can be obtained from $U = V$ and any its solution $\solution'$
	and for every operator $H \in \mathcal H$ the $H[\solution']$ is a solution of $U = V$.
\end{itemize}
Note that both $U' = V'$ and $\mathcal H$ depend on the nondeterministic choices, so it might
be that for different choices we can transform $U = V$ to $U' = V'$ (with $\mathcal H'$)
and to $U'' = V''$ (with a family $\mathcal H''$).

We also say that the equation $U=V$ with its solution \solution{}
are \emph{transformed into} $U'=V'$ with $\solution'$ and
that $\mathcal H$ is the \emph{corresponding family of inverse operators}.
In many cases, $\mathcal H$ consists of a single operator $H$,
in such case we call it the \emph{corresponding inverse operator},
furthermore, in some cases $H$ does not depend on $U = V$, nor on the nondeterministic choices.

In some cases for an equation $U = V$ and its solution \solution{} we explicitly tell, for which nondeterministic choices it is transformed
to some other equation $U' = V'$ with a solution $\solution'$ (intuitively: for the choices that implement the pair compression for \solution{}
or the block compression).

Our main goal is to show that subprocedures of \algsolveeq{} transform the minimal solutions
and to give the appropriate family of operators.

\subsection*{Inverse operators occurring in \algsolveeq}
We describe the family of inverse operators corresponding to various subprocedures of \algsolveeq.

\subsubsection*{Pair compression}
Define an operator \opprepend[w,X] for a string $w \in (\letters \cup \letters')^*$, which prepends $w$ to substitution for $X$
and leaves other variables untouched, formally:
$$
{\opprepend[w,X]}[\solution](X) = w \sol X \quad \text{ and } \quad {\opprepend[w,X]}[\solution](Y) = \sol Y, \text{ for } Y \neq X.
$$
Define \opappend[w,X] similarly, by appending $w$ to \sol X:
$$
{\opappend[w,X]}[\solution](X) = \sol X w \quad \text{ and } \quad {\opappend[w,X]}[\solution](Y) = \sol Y, \text{ for } Y \neq X.
$$

\begin{lemma}[cf.~Lemma~\ref{lem:pop preserves solutions}]
\label{lem: pop transforms}
$\algpop(\letters_\ell,\letters_r)$ transforms the minimal solutions.

Suppose it left-popped $b_X \in \letters_r \cup \{ \epsilon \}$
and right-popped $a_X \in \letters_\ell \cup \{ \epsilon \}$ from $X$,
then the corresponding inverse operator is:
$$
H = \displaystyle  \prod_{X \in \variables}  \opappend[a_X,X] \circ \opprepend[b_X,X] 
 \enspace .
$$
\end{lemma}
\begin{proof}
Fix the nondeterministic choices and let $H$ be as defined in the lemma statement.
Observe that from the proof of Lemma~\ref{lem:pop preserves solutions} it follows that if \solution{} is a solution of $U = V$
(note that we do not need to assume here that \solution{} is minimal)
then $U' = V'$ has a solution $\solution'$ such that $\sol U -= \solution'(U')$.
Since for each variable $X$ we replaced $X$ with $a_XXb_X$ (or $a_Xb_X$, when $X$ was removed from the equation),
this means that $\solution = H[\solution']$.

On the other hand, when $\solution'$ is a solution of $U' = V'$ then $\solution = H[\solution']$ satisfies
$\sol U = \solution'(U')$:
we replaced $X$ we replaced $X$ with $a_XXb_X$ (or $a_Xb_X$, when $X$ was removed from the equation) in $U = V$
and \sol X is obtained exactly by prepending $a_X$ and appending $b_X$ to $\solution'(X)$.
Similarly $\sol V = \solution'(V')$, which makes \solution{} a solution of $U = V$.

It is left to show that if $U = V$ with minimal solution \solution{}
is transformed to $U' = V'$ with $\solution'$
then also $\solution'$ is a minimal solution.
Suppose that $\solution'$ is not minimal, i.e.\ it is either a unifier solution or an instance of a unifier solution.
\begin{description}
	\item[it is a unifier solution] 
	Then $\solution'$ has a free letter. As $H$ only prepends and appends letters, also $\solution = H[\solution']$ has a free letter,
	which makes it a unifier solution,contradicting its minimality.
	\item[it is an instance of a unifier solution]
	Then $\solution' = \Phi[\solution'']$ for some unifier solution $\solution''$ and non-erasing, non-permutating morphism $\phi$
	which is constant on \letters.
	Observe that $\solution = H[\Phi[\solution'']]$. We claim that $\solution = \Phi[H[\solution'']]$:
	indeed, this follows from the fact that $H$ only prepends and appends letters from \letters, which are not affected by $\phi$.
	Since $H[\solution'']$ has a free letter, this makes \solution{} an instance of a $H[\solution'']$, contradiction
	with the minimality of \solution.
\qedhere
\end{description}
\end{proof}

We now investigate the inverse operator associated with \algpairc:  
\begin{lemma}[cf.~Lemma~\ref{lem: crossing pairs preserve}]
\label{lem: crossing pairs transforms}
$\algpairc(\letters_\ell,\letters_r)$ transforms the minimal solutions.
To be more precise, for the nondeterministic choices that implement the pair compression for solution \solution{} of $U = V$
obtaining $U' = V'$ with a corresponding solution $\solution'$,
the $U = V$ with \solution{} is transformed to $U' = V'$ with $\solution'$.

Let $H$ be the inverse operator of the \algpop{} applied in \algpairc,
furthermore let \algpairc{} replaced pairs $ab \in \letters_\ell \letters_r$ with $c^{(ab)}$.
Then the corresponding inverse operator is:
\begin{equation}
\label{eq: proper operator 1}
H \circ \prod_{ab \in \letters_\ell\letters_r} \Morph {c^{(ab)}} {ab} \enspace .
\end{equation}
\end{lemma}
Note that as $c^{ab}$ is not in $\letters_\ell \cup \letters_r$ then the order of applying the $\Morph {c^{(ab)}} {ab}$ does not matter. 
\begin{proof}
Observe that by Lemma~\ref{lem: pop transforms} the application of \algpop{} transforms the minimal solutions;
furthermore, the inverse operator for it is well-described.
So it is left to consider the compression of pairs performed by \algpairc.
Such a compression is a composition of many \algpair{}
(and as the pairs are non-overlapping, the order of those compression does not matter),
so it is enough to show
that when $\algpair(a,b)$ transforms $U = V$ with a minimal \solution{} to $U' = V'$ with $\solution'$
then $\solution'$ is minimal and \Invmorph {c} {ab} is the corresponding inverse operator.

Suppose that \solution{} is a solution of $U = V$.
Observe that for appropriate non-deterministic choices (done in \algpop) the pair $ab$ is noncrossing,
see Lemma~\ref{lem:pop preserves solutions}.
Moreover, the compressions of a pair $a'b'$ cannot make $ab$ crossing,
as $a'b'$ do not overlap with $ab$, the letter that replaced $a'b'$ is not $a$ and not $b$
and lastly no letters are popped from the variables during the compression of non-crossing pairs.
Then by Lemma~\ref{lem:paircomp blockcomp},
the equation $U' = V'$ returned by $\algpair(a,b)$ has a solution $\solution'$ which implement the pair compression,
i.e.\ $\sol U = \morph {c} {ab} (\solution'(U'))$.
Since $\algpair(a,b)$ does not modify variables, this means that $\solution = \Morph {c} {ab} [\solution']$, as claimed.

Suppose that $U = V$ with a minimal solution \solution{} is transformed
into $U' = V'$ with a solution $\solution'$, which is not minimal.
There are two cases: either $\solution'$ is a unifier solution,
or it is an instance of the unifier solution;
we consider only the latter, the former is shown using a similar argument.
Then on one hand $\solution = \Morph c {ab} [\solution']$ and on the other
$\solution' = \Phi[\solution'']$ for morphism $\phi: \letters \cup \letters' \mapsto (\letters \cup \letters')^+$
which is non-erasing, non-permutating and constant on $\letters$.
Then $\solution = (\Morph c {ab} \circ \Phi) [\solution'']$
and it is left to show that $\Morph c {ab} \circ \Phi$ corresponds to some
non-erasing non-permutating morphism $\phi'$.
Define $\phi'(x) = (\morph c {ab} \circ \phi)(x)$.
It is easy to observe that $\phi'$ is non-erasing and constant on \letters:
indeed, both $\phi$ and $\morph c {ab}$ are non-erasing and constant on \letters,
so their composition is as well.
Lastly, it is non-permutating: observe that $c$ is not in the image of $\phi'$, as it is not in the image of $\morph c {ab}$,
so $\phi'$ cannot be a permutation.

It is left to show that when $\solution'$ is a solution of $U' =V'$ returned by $\algpair(a,b)$
then $\solution = \Morph c {ab} [\solution']$ is a solution of $U = V$.
Note that $\sol U = \morph {c} {ab} (\solution'(U'))$, as each explicit $c$ in $U'$ was obtained by replacing $ab$ by $\algpair(a,b)$,
while each implicit $c$ in $\solution'(U')$ was replaced by $ab$ by $\Morph {c} {ab}$.
In the same way $\sol V = \morph {c} {ab} (\solution'(V'))$, which shows that \solution{} is a solution of $U = V$.
\qedhere
\end{proof}

\subsubsection*{Block compression}
For block compression \algblocksi, the family of inverse operators is quite complicated.
Instead of introducing it and proving its properties in one go, we choose to make some intermediate steps,
which hopefully make a smoother presentation.
We begin with describing the family of inverse operators for \algprefsuff{} and then give the one for \algblocksc;
in each of those cases the corresponding family consists of a single operator, similarly as in the case of \algpop{} and \algpairc.
From Lemma~\ref{lem:improved block compression} we know that a run of \algblocksi{} represents several runs of \algblocksc,
and so in some sense it is a `parametrised' \algblocksc. This approach extends to inverse operators:
The family of inverse operators for \algblocksi{} represents
(in a parametrised way) several inverse operators for different runs of \algblocksc.
The family is defined using the solutions of the system $D$ created by \algblocksi,
with a single inverse operator corresponding to a solution of the Diophantine system $D$.

For \algprefsuff{} the analysis is exactly the same as for \algpop:
the inverse operator appends and prepends the letters that were popped by \algprefsuff.

\begin{lemma}
\label{lem:cutpref cutsuff transform}
\algprefsuff{} transforms the minimal solutions.
The corresponding inverse operator is
$$
\prod_{X \in \variables} \opprepend[a_X^{\ell_X},X] \circ \opappend[b_X^{r_X},X] \enspace ,
$$
where $a_X^{\ell_X}$ ($b_X^{r_X}$) is the prefix (suffix, respectively) removed from \sol X.
\end{lemma}
\begin{proof}
The proof is similar to the proof of Lemma~\ref{lem: pop transforms} and it is thus omitted.
\qedhere
\end{proof}

This allows stating the result for \algblocksc: intuitively it first replaces each $a_\ell$ with appropriate $a^\ell$ and then
appends and prepends the prefixes and suffixes popped by \algprefsuff, in a similar way as the inverse operator for \algpairc,
see Lemma~\ref{lem: crossing pairs transforms}.

Note that the inverse operator needs to supply the information, which letter is $a_\ell$ and the value of $\ell$
(since $a_\ell$ is just a naming convention that makes the read-up of the paper more accessible,
the algorithm does not know which letter `is' $a_\ell$ and what is the value of $\ell$).
Also, at the first glance it seems that the inverse operator could replace an arbitrary number of different letters $a_\ell$.
Still, as we are interested only in transforming the minimal solutions, this is not the case:
by Lemma~\ref{lem:over a cut improved} in minimal solutions if $a^\ell$ occurs in \sol U then $a^\ell$ is a visible maximal block.
Hence, the corresponding inverse operator does not need to introduce blocks of other form.

\begin{lemma}
\label{lem: consistent no crossing block transform}
$\algblocksc$ transforms the minimal solutions.
To be more precise, for the nondeterministic choices that implement the blocks compression for solution \solution{} of $U = V$
obtaining $U' = V'$ with a corresponding solution $\solution'$,
the $U = V$ with \solution{} is transformed to $U' = V'$ with $\solution'$.

Let $H$ be the inverse operator of the \algprefsuff, then the corresponding inverse operator for \algblocksc{} is
$$
H \circ \prod_{a \in \letters} \Blockc \enspace,
$$
where $\blockc$ replaces $a_\ell$ with $a^\ell$;
if $a^\ell$ replaces $a_\ell$ then $a^\ell$ is a a visible maximal block in $U = V$ for \solution.
Without loss of generality we may assume that $H$ appends $a_X^{\ell_X}$ and prepends $b_X^{r_X}$,
where $a_X$ and $b_X$ are the first and last letter of \sol X and $\ell_X$ and $r_X$ are the lengths of the $a_X$ prefix and $b_X$ suffix of \sol X.
\end{lemma}

At the first glance, the condition that the inverse operator for \algblocksc{} replaces $a_\ell$ by $a^\ell$ only when
$a^\ell$ is a visible block in $U = V$ for \solution{} seems bad,
as we promised that the inverse operators \emph{do not} depend on the particular solutions,
but rather on the equation and the non-deterministic choices (here: of \algblocksc).
However, after a second thought, there is nothing bad with this: observe that the lengths of the visible blocks
are implicitly defined in the nondeterministic choices of \algprefsuff,
as it pops prefixes of length $\ell_X$ and suffixes of length $r_X$ from variable $X$
and the lengths of the visible blocks linearly depend on $\{\ell_X,r_X \}_{X \in \variables}$, see Lemma~\ref{lem: coherent solution}.
This observation is not formalised: in any case Lemma~\ref{lem: consistent no crossing block transform}
is used only as an intermediate step in description of the inverse operators for \algblocksi.
And in case of \algblocksi{} the given inverse operator shall not depend on the solution \solution{}
(though this formulation of Lemma~\ref{lem: consistent no crossing block transform} is helpful in the proof).

\begin{proof}
By Lemma~\ref{lem:cutpref cutsuff transform} the \algprefsuff{} transforms minimal solutions.
Furthermore, for appropriate choices, there are no crossing blocks in the obtained $U' = V'$ with respect to $\solution'$,
see Lemma~\ref{lem:cutpref cutsuff}.

Observe that afterwards $\algblocksc$ is a composition of $\algblocks(a)$ for all letters $a$.
The rest of the proof is similar to the one in Lemma~\ref{lem: crossing pairs transforms}.
Furthermore, by Lemma~\ref{lem: consistent no crossing block} we know that
in the run of \algblocksc{} that implements the blocks compression
the \algprefsuff{} pops exactly the $a_X$-prefix and $b_X$-suffix from each variable $X$,
where \sol X begins with $a_X$ and ends with $b_X$,
so we can also use those choices when transforming the solution \solution.

Concerning the restriction of the replaced letters $a_\ell$: when \solution{} is a minimal solution,
then whenever $a^\ell$ is a maximal block, it also has a visible occurrence in $U = V$ for \solution, see Lemma~\ref{lem:over a cut improved}.
As the $\solution'$ of $U' = V'$ corresponds to the implementation of block compression, the maximal blocks in \sol U and \sol V
are obtained by replacing single letters in $\solution'(U')$ and $\solution'(V')$
by blocks.
Hence, suppose that the inverse operator morphs a letter $b_k$ to $b^k$ and $b^k$ is not a visible maximal block in $U = V$ for \solution.
Hence $b_k$ is not a maximal block in \sol U.
In particular, $b_k$ does not occur in $\solution'(U')$, so
we can remove this morphed pair from the inverse operator.
\qedhere
\end{proof}

Now we are ready to define the family of operators for \algblocksi.
Intuitively, it will encode several inverse operators associated with different
nondeterministic choices of \algblocksc{} in a compact way:
instead of making explicit guesses about the lengths of the prefixes and suffixes of $\sol X$,
it will parametrise them using variables $x_X$ and $y_X$.
On the other hand $\solution'(U')$ contains letters, which represent blocks of letters and lengths of those blocks
linearly depend on $x_X$ and $y_X$.
The coherence of all these lengths is guaranteed by an appropriate system of Diophantine equations
(exactly as in the case of \algblocksi).

Recall that for an arithmetic expression $e_i$ in variables $\{x_{X}, y_{X}\}_{X \in \variables}$
the $e_i[\{\ell_X, r_X\}_{X \in \variables}]$ denotes the value of $e_i$ when $\{\ell_X, r_X\}_{X \in \variables}$
are substituted for $\{x_{X}, y_{X}\}_{X \in \variables}$.

Suppose that \algblocksi{} constructs a linear Diophantine system $D$ in variables $\{x_{X}, y_{X}\}_{X \in \variables}$,
and popped a prefix $a_X^{x_X}$ of $a_X$ and suffix $b_X^{y_X}$ of $b_X$ from $X$.
(Note that $D$ depends on the nondeterministic choices of \algblocksi.)
Then we define a (finite or infinite) family of inverse operators:
let $\{e_i\}_{i=1}^m$ be the lengths of parametrised visible maximal blocks for the
coherent prefix-suffix structure (in variables $\{x_X, y_X\}_{X \in \variables}$).
By the definition the sides of equations in $D$ are those expression.
Then, for a solution $\{\ell_X, r_X\}_{X \in \variables}$ of $D$
the following operator is in a family of operators $\mathcal H_{D}$:
\begin{equation}
\label{eq: proper operator 2}
\left(\displaystyle \prod_{X \in \variables} \opprepend[a_X^{\ell_X},X] \circ \opappend[b_X^{r_X},X]\right) \circ \prod_{a \in \letters}\Blockc
\enspace,
\end{equation}
where $\blockc$ replaces the letter $a_{e_i}$ by $a^{e_i[\{\ell_X, r_X\}_{X \in \variables}]}$, and no other letters are replaced.
Note that the operator needs to explicitly point the letters that it treats as $a_{e_i}$ as well
as the arithmetic expressions $e_i$ (so that it can calculate $e_i[\{\ell_X, r_X\}_{X \in \variables}]$):
the index $e_i$ in $a_{e_i}$ is just a notation convention to make the read-up of the paper easier,
the actual letter does not carry any information about $a$ nor $e_i$.

\begin{lemma}
\label{lem:improved block compression transform}
Let \algblocksi{} return a satisfiable linear Diophantine system $D$.
Then \algblocksi{} transforms the solutions
and $\mathcal H_D$ is the corresponding family of inverse operators.
\end{lemma}
\begin{proof}
Consider an equation $U = V$ and its solution \solution.
We first want to show that for some choices of $\algblocksi$ the obtained
equation $U' = V'$ has a solution $\solution'$ such that $\solution = H[\solution']$ for
some $H \in \mathcal H_D$, where $\mathcal H_D$ is the corresponding family of inverse operator, defined in~\eqref{eq: proper operator 2}.
To this end we use the fact that \algblocksc{} transforms the minimal solution and that its corresponding inverse operator is known,
see Lemma~\ref{lem: consistent no crossing block transform}.

Consider a run of $\algblocksc$ which implements the blocks compression for \solution{} and $U = V$
(obtaining $U' = V'$ with a corresponding $\solution'$).
From Lemma~\ref{lem: consistent no crossing block transform} we know that for those very choices
\algblocksc{} transforms the $U = V$ with \solution{} to $U' = V'$ with $\solution'$.
Consider on the other hand the run of \algblocksi{} in which \wordtodiophantine{} returns a Diophantine system $D$
that is \solution-coherent.
Then by Lemma~\ref{lem:improved block compression} this run leads to the same instance $U' = V'$
(up to renaming of letters).
So it is left to show that the appropriate inverse operator is in $\mathcal H_D$.

The inverse operator $H$ for \algblocksc{} first replaces letters $a_\ell$ with $a^\ell$,
where $a^\ell$ is a visible maximal block in $U = V$ for \solution,
and then appends $a_X^{\ell_X}$ and prepends $b_X^{r_X}$ to each variable $X$,
where $\ell_X$ and $r_X$ are the lengths of the $a_X$-prefix and $b_X$-suffix of \sol X.
Observe that by Lemma~\ref{lem: S compatible} the $\{\ell_X,r_X \}_{X \in \variables}$ is a solution of $D$.
Note that $H$ is in $\mathcal H_D$: when $a^\ell$ is the maximal block $E_i$
then by Lemma~\ref{lem: coherent solution} the corresponding parametrised maximal block
has length $e_i$ such that $|E_i| = e_i[\{\ell_X,r_X \}]$
and the inverse operator corresponding to $\{\ell_X,r_X \}_{X \in \variables}$
replaces $a_{e_i}$ with $a^{e_i[\{\ell_X,r_X \}]} = a^{|E_i|} = a^\ell$
and then prepends $a_X^{\ell_X}$ and appends $b_X^{r_X}$ to substitution for $X$.

We now show that if an equation $U = V$ is transformed
by $\algblocksi$ into $U' = V'$ which has a solution $\solution'$
and $H \in \mathcal H_D$ then $\solution = H[\solution']$ is a solution of $U = V$.
Let us first recall, how $U' = V'$ is obtained from $U = V$ and how $H$ looks like.

By definition, \algblocksi{} first guesses the prefix-suffix structure for $U = V$ (i.e.\ what is the first and last letter of \sol X
and whether \sol X is a block of letters)
pops the $a_X^{\ell_X}$ prefix and $b_X^{r_X}$-suffix from $X$ for each $X \in \variables$, guesses system $D$
coherent with the prefix-suffix structure
(whose sides are lengths of parametrised explicit maximal blocks)
and replaces blocks whose lengths are equalised in the system by a single letter $a_{e_i}$
(where $e_i$ is one of the lengths of the equalised blocks).
Then $H$ corresponds to a solution $\{\ell_X,r_X \}_{X \in \variables}$ of $D$:
it first replaces $a_{e_i}$ by $a^{e_i[\{\ell_X,r_x \}_{X \in \variables}]}$
and then prepends $a_X^{\ell_X}$ and appends $b_X^{r_X}$ to \sol X for each $X \in \variables$.

To show that $\solution = H[\solution']$ is a solution of $U = V$
we show that \sol U is obtained from $\solution'(U')$ by replacing each $a_{e_i}$ with $a^{e_i[\{\ell_X,r_x \}_{X \in \variables}]}$;
the same will hold for \sol V and in this way \solution{} is shown to be a solution of $U = V$.
To this end we define an intermediate substitution $\solution_1$ and an equation $U_1 = V_1$:
the $\solution_1$ is obtained from $\solution'$ similarly as \solution,
but without the appending and prepending the letters, just by replacing $a_{e_i}$ with $a^{e_i[\{\ell_X,r_x \}_{X \in \variables}]}$.
Similarly, define $U_1$ and $V_1$ by replacing those letters in $U'$ and $V'$.
Then clearly $\solution_1(U_1)$ is $\solution'(U')$ with each $a_{e_i}$ replaced with $a^{e_i[\{\ell_X,r_x \}_{X \in \variables}]}$.

Since appending $a_X^{\ell_X}$ and prepending $b_X^{r_X}$ to $\solution_1(X)$ for each $X \in \variables$ turns $\solution_1$ to \solution,
by reversing the procedure we obtain that
popping $a_X^{\ell_X}$ to the left and $b_X^{r_X}$ to the right from each $X \in \variables$ turns \solution{} to $\solution_1$.
To finish the proof we show that when we pop $a_X^{\ell_X}$ to the left and $b_X^{r_X}$ to the right, we turn $U = V$ to $U_1 = V_1$.
To this end we just need to show that the consecutive maximal blocks of letters in $U = V$ after the popping are the same as in $U_1 = V_1$.
Since we pop exactly the prefixes and suffixes, the former are exactly the visible maximal blocks in $U = V$ for \solution,
which by Lemma~\ref{lem: coherent solution} have lengths $e_i[\{\ell_X,r_X\}]$.
On the other hand, the maximal blocks in $U_1 = V_1$ are obtained by substitutions for single letters in $U' = V'$
and are of lengths $e_i[\{\ell_X,r_X\}]$, which ends the proof.
\qedhere
\end{proof}

\subsection*{Representation of a single minimal solution}
We now show that each \emph{minimal} equation can be obtained by retracing the steps of some successful run of \algsolveeq.
It turns out that during this retracing we can restrict ourselves to equations that are short
and inverse operators that morph, prepend and append only letters actually present in the equation.
We define these notions formally, they are used again the definition of the graph representation
of all solutions.

\begin{definition}
We say that a word equation $U = V$ is \emph{proper} if $U, V \in (\letters \cup \variables)^*$,
in total $U$ and $V$ have at most $n_v$ occurrences of variables (where $n_v$ is hte number of occurrences of variables in the input equation)
and $|U| + |V| \leq cn$ (for an appropriately chosen in advance constant $c$).
An equation is \emph{trivial} if both its sides have length at most $1$.

A family $\mathcal H$ of inverse operators (corresponding to the transformation of $U = V$ to $U' = V'$)
is \emph{proper} if it is one of the families defined in~\eqref{eq: proper operator 1}
or in~\eqref{eq: proper operator 2}. 
Furthermore $H \in \mathcal H$ morphs only letters present in $U' = V'$ and appends/prepends only letters that occur in $U' = V'$
or images of such letters by the morphing in $H$.
\end{definition}
The intuition is as follows:
the first type of edges, labelled with~\eqref{eq: proper operator 1}, corresponds to \algpairc{} in \algsolveeq.
The second, labelled with~\eqref{eq: proper operator 2}, corresponds to \algblocksi.

Note that the assumption that proper families of operators may only append, prepend and morph letters that
are present in the equation is a restriction on the potential form of inverse operators
from families~\eqref{eq: proper operator 1} and~\eqref{eq: proper operator 2}
and it is shown in Lemma~\ref{lem: solution path} below that indeed such restricted families are enough to describe all minimal solutions.

On the other hand such a restriction makes it easier to describe such families:
proper families of inverse operators need to specify, which letters are replaced,
but in both cases they list at most $cn$ letters for replacement (as this is the size of the equation).
Furthermore, the family of inverse operators~\eqref{eq: proper operator 2} needs also to specify
the expression $e_i$ for each of the letters it intends to replace.
Since there are at most $cn$ such letters, there are also at most $cn$ such expressions.
So the whole description size is polynomial.

\begin{lemma}
\label{lem: solution path}
If $U_0 = V_0$ is of size $n$ and has a minimal solution $\solution_0$ then there exists
a sequence of proper equations $U_0 = V_0$, $U_1 = V_1$, \ldots, $U_m = V_m$ such that
\begin{itemize}
	\item $V_m = U_m$ is trivial and $m = \Ocomp(\log |\solution_0(U_0)|)$;
	\item for appropriate nondeterministic choices a subprocedure (\algblocksi{} or \algpairc) of \algsolveeq{} 
	transforms an equation $U_{i-1} = V_{i-1}$ with a minimal solution $\solution_{i-1}$ to $U_{i} = V_i$
	with a minimal solution $\solution_i$;
	\item the corresponding inverse operator is from a proper family.
\end{itemize}
\end{lemma}
\begin{proof}
Lemma~\ref{lem: crossing pairs transforms} and Lemma~\ref{lem:improved block compression transform}
guarantee that \algpairc{} and \algblocksi{} transform minimal solutions,
so there exist a sequence $U_0 = V_0$, $U_1 = V_1$, \ldots, $U_m = V_m$ together with minimal solutions $\solution_0$,
$\solution_1$,\ldots, $\solution_m$ such that
\begin{itemize}
	\item $U_{i-1} = V_{i-1}$ is transformed to $U_i = V_i$ by some subprocedure (\algblocksi{} or \algpairc) of \algsolveeq,
	where $\mathcal H_i$ is the corresponding family of inverse operators;
	\item $\solution_{i-1} = H_i[\solution_i]$ for some $H_i \in \mathcal H_i$;
	\item $U_m = V_m$ is trivial.
\end{itemize}
What is not known is whether:
\begin{enumerate}[(Q 1)]
	\item each $U_i = V_i$ is proper?  \label{Q 1}
	\item $m = \Ocomp(\log (|\solution_0(U_0)|))$? \label{Q 2}
	\item each $\mathcal H_i$ is proper?  \label{Q 3}
\end{enumerate}
As soon as we settle \Qref{1}--\Qref{3}, the proof is complete.

The \Qref{3} is easy: since $\solution_i$ is minimal, by definition it assigns only letters from $U_i = V_i$,
so there is no reason for $H_i$ to morph any other letters (as they are simply not in $\solution_i(U_i)$);
also, by Lemma~\ref{lem: pop transforms} the inverse operator for \algpop{} prepends and appends
letters that were popped from variables, i.e.\ either they are present in $U_{i+1} = V_{i+1}$
or some letters in $U_{i+1} = V_{i+1}$ replaced pairs that includes such a letter.
a similar argument holds for the inverse operator for \algblocksi,
see Lemma~\ref{lem:improved block compression transform}.
This shows that indeed each $\mathcal H_i$ is proper and so establishes~\Qref{3}.

Concerning~\Qref{1} note that we do not introduce any new occurrences of variables,
so we just need to bound the lengths of equations $U_0 = V_0$, $U_1 = V_1$, \ldots, $U_m = V_m$;
such a bound was already given in Lemma~\ref{lem: space consumption},
though it was not guaranteed there that the appropriate minimal solution is transformed.

Similarly, for \Qref{2} the length of the successful computation is given in Lemma~\ref{lem:logM iterations},
but again nothing is known about transforming minimal solutions.

Both proofs of Lemma~\ref{lem: space consumption} and Lemma~\ref{lem:logM iterations}
rely on Claim~\ref{clm:words are compressed},
the following stronger version of Claim~\ref{clm:words are compressed} gives all the needed properties:
\begin{clm}
\label{clm:words are compressed improved}
Let $U = V$ has a minimal solution \solution.
For appropriate choices, during one phase it is transformed to an equation $U' = V'$ with a minimal solution $\solution'$
such that
\begin{itemize}
	\item $1/6$ of letters in \sol U (rounding down) is compressed in $\solution'(U')$.
	\item at least $(|U| + |V| - 3 n_v - 4)/6$ of letters in $U$ or $V$ are compressed in $U'$ or $V'$;
\end{itemize}
\end{clm}
Now, the proof for \Qref{1} follows in the same way as in Lemma~\ref{lem: space consumption}:
we choose the sequence of equations $U_0 = V_0$, $U_1 = V_1$, \ldots, $U_m = V_m$
guaranteed to exist by Claim~\ref{clm:words are compressed improved} .
As in Lemma~\ref{lem: space consumption} it can be shown that each such an equation has size
at most $79 n$ and the intermediate equations have length at most $85n$.
Similarly, the same run guarantees that the size of the corresponding solution $\solution_i$
shrinks by a constant factor at the beginning of each phase (so every three equations).

It is left to show Claim~\ref{clm:words are compressed improved}.
To this end note that Claim~\ref{clm:words are compressed} shows that the two shortening properties for a solution \solution{}
hold for the non-deterministic choices that implement the block compression and pair compression for this solution
(for appropriate partition of letters).
But Lemma~\ref{lem:improved block compression transform}
shows that for a solution \solution{} the non-deterministic choices in \algblocksc{} that implement the block compression
also transform the minimal solution; similar claim holds for \algpairc{} by Lemma~\ref{lem: crossing pairs transforms}.
Which ends the proof.
\end{proof}

\subsection*{Representation of all minimal solutions}
As already said, the set of minimal solutions will be represented by a directed graph.
Intuitively, the graph \G{} represents all paths from Lemma~\ref{lem: solution path}.
\begin{definition}
Directed graph \G{} for a satisfiable input equation $U_0 = V_0$ has
\begin{itemize}
	\item nodes labelled with satisfiable proper equations;
	\item edges of \G{} are labelled with a family of proper operators;
	\item an edge from $U = V$ to $U' = V'$ labelled $\mathcal H$ if and only if
	for some nondeterministic choices \algsolveeq{} (that uses \algblocksi{} instead of \algblocksc)
	transforms $U = V$ into $U' = V'$ and $\mathcal H$ is the corresponding
	family of inverse operators;
	\item each node is reachable from node labelled with $U_0 = V_0$.
\end{itemize}
\end{definition}
We say that \solution{} for $U = V$ is \emph{obtained by a path} to $U' = V'$ with $\solution'$
if $\solution$ is obtained by applying a composition of inverse operators on the path from $U = V$ to $U' = V'$ applied to $\solution'$.

We need to show that on one hand \G{} can be constructed in \PSPACE,
and on the other, that it describes all minimal solutions of a word equation.
We begin with the latter.

\begin{lemma}
\label{lem: graph is OK}
Let \solution{} be a minimal solution of an equation $U = V$ (of size $n$) that is a node in \G.
Then there is a path in \G{} starting in $U = V$ and ending in a trivial equation $U' = V'$ with a minimal solution $\solution'$
such that \solution{} can be obtained from $\solution'$ by this path.

Moreover, each \solution{} obtained in this way is a solution of $U = V$.

If the equation $U = V$ is trivial (i.e.\ $|U|, |V| \leq 1$) then there is at most one minimal solution,
which is easy to describe.
\end{lemma}
\begin{proof}
Concerning the third claim, 
observe that describing all minimal solution of a satisfiable equation $U = V$
such that $|U| = |V| = 1$ is easy:
\begin{itemize}
	\item if a variable $X$ does not occur in the equation, then $\sol X = \epsilon$
	for each minimal solution $X$, so it is enough to consider variables that occur in $U = V$;
	\item if $U, V \in \letters$ then either there is one minimal solution ($\sol X = \epsilon$ for each $X \in \variables$), when $U = V$,
	or none solution, when $U \neq V$;
	\item if one of the sides is a letter and the other a variable, then there is only one solution;
	\item if both sides consist of variables, then there is no minimal solution
	(as each solution is an instance of a unifier solution that assigns $v$ to variables on both sides).
\end{itemize}
If $U = \epsilon$ or $V = \epsilon$, then a similar analysis shows that if $U = V = \epsilon$ then the unique minimal solution
assigns $\epsilon$ to each variable and otherwise there is no minimal solution.

For the second claim, observe that it summarizes the properties of subprocedures
of \algsolveeq, presented in Lemma~\ref{lem: crossing pairs transforms} and~\ref{lem:improved block compression transform},
which claim that \algpairc{} and \algblocksi{} transform minimal solutions and the definition of the graph $\mathcal G$.

For the first claim observe that this is just a reformulation of Lemma~\ref{lem: solution path}.
\qedhere
\end{proof}

\subsubsection*{Constructing the representation of all minimal solutions}
We show that it is possible to generate \G{} within \PSPACE.
In order to do so we should be able to decide in \PSPACE{} whether:
\begin{description}
	\item[node label check] a given proper equation $U = V$ labels a node in \G.
	\item[edge label check] given two satisfiable proper equations $U = V$ and $U' = V'$ and a proper
	family $\mathcal H$ of operators there is an edge from $U = V$ to $U' = V'$ labelled with
	$\mathcal H$.
\end{description}
When \PSPACE{} procedures for these two tasks are known, constructing $\mathcal G$
is easy:
\begin{itemize}
	\item we iterate over all proper equations (they have length at most $cn$), 
	for a fixed equation $U = V$ we check whether $U = V$ labels a node in $\mathcal G$.
	If so, we output it.
	\item we iterate over all pairs of proper equations $U = V$ and $U' = V'$
	and every proper family of operators $\mathcal H$ labelling this edge;
	note that both~\eqref{eq: proper operator 1} and~\eqref{eq: proper operator 2}
	can be described using $\Ocomp(n)$ symbols.
	For a fixed triple we verify, whether there is an edge from $U = V$ to $U' = V'$
	labelled with $\mathcal H$.
	If so, we output triple $(U =V, \mathcal H, U' = V')$.
\end{itemize}
Clearly this procedure uses only polynomial space and properly generates \G.

It is thus left to show that node and edge label checks
can be performed in \PSPACE.

\begin{lemma}
\label{lem:generation}
Node and label checks can be performed in \PSPACE.
\end{lemma}
\begin{proof}
Consider first node label check.
It is trivial to verify, whether $U = V$ has length at most $cn$ and at most $n_v$
variables' occurrences.
Using \algsolveeq{} we can verify in \PSPACE, whether $U = V$ is satisfiable.
Also, in \NPSPACE{} we can verify, whether \algsolveeq{}
transforms $U_0 = V_0$ to $U = V$:
we begin with $U_0 = V_0$ and transform it using \algsolveeq{} until
$U = V$ is obtained. As by Lemma~\ref{lem: space consumption}
\algsolveeq{} uses $\Ocomp(n \log n)$ space,
this is doable in \NPSPACE. As \NPSPACE=\PSPACE, we are done.

In a similar way we can show that edge label check can be performed in \NPSPACE.
Firstly using node label check we verify, whether both
$U = V$ and $U' = V'$ label nodes of \G.
The label $\mathcal H$ uniquely identifies the subprocedure of \algsolveeq{}
that should be applied to $U = V$ in order to obtain $U' = V'$;
we thus take $U = V$ and apply this subprocedure.
As \algsolveeq{} uses $\Ocomp(n \log n)$ space, by Lemma~\ref{lem: space consumption}
and Lemma~\ref{lem:improved block compression},
then this can be tested in \NPSPACE{}
and consequently edge label check can be executed in \PSPACE
\qedhere
\end{proof}

\subsubsection*{Minimal unifier solutions}
The presented procedures are enough to construct a finite representation of
all minimal solutions. It is left to formalise the transformation of
minimal unifier solutions to minimal solutions.
Since by Lemma~\ref{lem:unifier reduces} we know that \solution{} is a minimal
unifier solution with a free letter $a$ then this free letter is a first letter of some \sol X.
So when we left-pop the first letter of each variable, we introduce all free letters used by \solution{}
into the equation, making them usual letters and turning the minimal unifier solution \solution{}
into a minimal solution $\solution'$ of the new equation.

\begin{lemma}
\label{lem:unifier to minimal}
Let \solution{} be a minimal unifier solution of $U = V$.
Then for some nondeterministic choices $\algpop(\letters',\emptyset)$ returns an equation $U' = V'$
such that $\solution = \Phi[\solution']$ for some minimal solution $\solution'$ of $U' = V'$,
where
$$
\Phi = \prod_{X \in \variables} \opprepend[a_X,X]
 \enspace ,
$$
where $a_X$ is the symbol left-popped from a variable $X$,
i.e.\ $a_X \in \letters' \cup \epsilon$.
\end{lemma}
\begin{proof}
We know from Lemma~\ref{lem: pop transforms} that \algpop{} transforms minimal solutions
and the given operator $\Phi$ is the corresponding inverse operator.
In fact, the proof in the direction we use does not assume that \solution{}
is minimal, it can be an arbitrary solution.
Hence for a minimal unifier solution \solution{}
and appropriate guesses the $U = V$ with \solution{} is transformed by \algpop{} into $U' = V'$ with $\solution'$,
such that $\solution = \Phi[\solution']$.
To be more precise, the guesses are consistent with \solution,
in the sense that \algpop{} left-pops a letter $a_X$
if and only if the first letter of \sol X is $a_X$ and $a_X \in \letters'$.
So it is left to show that if \solution{} is a minimal unifier solution,
then $\solution'$ is a minimal solution.

We first show that $\solution'(U')$ has no free letters.
By Lemma~\ref{lem:unifier reduces}, if $a$ is a free letter in \solution,
then $a$ is the first letter of some \sol X. And we fixed the nondeterministic choices for which \algpop{} left-pops $a$ from $X$,
and so $a$ occurs as an explicit letter in $U' = V'$.
As we choose $a$ arbitrarily, all free letters of \sol U occur in $U' = V'$.
As $\solution'(U') = \sol U$, by Lemma~\ref{lem:pop preserves solutions},
the $\solution'$ has no free letters and so it is a non-unifier solution.

Suppose that $\solution'$ is not minimal, as the case that it is a unifier solution is already excluded,
this happens only in the case when \solution{} is an instance of some other unifier solution,
i.e.\ there is a unifier solution $\solution''$ such that $\solution' = \Phi'[\solution'']$ for some 
non-erasing, non-permutating morphism
$\phi': (\letters \cup \letters') \cup \letters'' \mapsto (\letters \cup \letters')^+$
which is constant on $\letters \cup \letters'$
(the set $\letters''$ is a new set of free letters, such that
$\letters'' \cap (\letters' \cup \letters) = \emptyset$).
We show that this contradicts the assumption that \solution{} is a minimal unifier solution.
Since $\Phi$ prepends letters from $\letters \cup \letters'$,
which are not affected by $\phi'$,
it can be concluded that $\solution = \Phi' [\Phi[\solution']]$:
\begin{align*}
\solution
	&=
\Phi[\solution'] &\text{by definition}\\
	&=
\Phi [\Phi'[\solution'']] &\text{by a contrario assumption on $\solution'$}\\
	&=
\Phi' [\Phi[\solution'']] &\text{as $\phi'$ does not affect letters added by } \Phi \enspace .
\end{align*}
It is left to show that $\Phi[\solution'']$ is a unifier solution:
but $\solution''$ contains letters from $\letters''$ and $\Phi$
only prepends letters to $\solution''(X)$, hence $\Phi[\solution'']$
is a unifier solution.
As we already know that $\phi'$ is non-erasing and non-permutating,
and constant on $\letters \cup \letters'$,
we conclude that \solution{} is an instance of $\Phi[\solution'']$,
which contradicts the assumption that it is minimal.
\qedhere
\end{proof}

Now we are ready to give the proof of the representation Theorem~\ref{thm:generation}
on generating the finite representation of minimal unifier solutions of a word equation.

\begin{proof}[proof of Theorem~\ref{thm:generation}]
Consider first a graph representation of minimal solutions of an equation.
Each node has at most polynomial description, so does the edges.
By Lemma~\ref{lem:generation} it can be checked in polynomial space,
whether a node is present in the graph and whether an edge (labelled) joins two nodes.
Since the description of a node (edge) has polynomial size,
there are at most exponentially many nodes (edges, respectively).

In order to generate a graph representation of all minimal and unifier-minimal solutions,
we use the approach presented in Lemma~\ref{lem:unifier to minimal}.
Given an equation $U = V$ we iterate over equations $U' = V'$
of length at most $n + n_v$ and using at most $n_v$ free letters from $\letters'$,
whether $\algpop(\letters',\emptyset)$ can transform $U = V$ into $U' = V'$.
If so, output the node labelled with $U' = V'$ and make a graph representation
of all its minimal solutions.
The label on the edge from $U = V$ to $U' = V'$
is the inverse operator returned by \algpop,
see Lemma~\ref{lem:unifier to minimal}.
Clearly, this procedure still runs in \PSPACE,
and so also the generated graph has exponential size.
\qedhere.
\end{proof}

\section{Other theoretical properties}
\label{sec:theoretical bounds}
In this section, we give (alternative) proofs of two known theoretical properties of word equations,
using the approach of recompression:
an exponential bound on the periodicity bound and the doubly-exponential bound on
the size of the length-minimal solution.

\subsection*{Exponential bound on exponent of periodicity}
\label{sec:per bound}
As already described in the introduction, exponential bound on exponent of periodicity,
shown by Ko\'scielski and Pacholski~\cite{Koscielski}, is one of the most
often used results on words equations.
Their proof follows by first considering
so-called $P$-presentations of a string; roughly, given a string $w$
and a primitive word $P$, a $P$-presentation is a canonical factorisation
of $w$ into powers of $P$ and other strings.
Then each power of $P$ is associated with
a number and treating such numbers as variables leads to a system of satisfiable
Diophantine equations. Solutions of this system induced solutions of the word equation.
In particular, length-minimal solution corresponded to minimal (in some sense)
solution of the Diophantine equation.
This is similar to results presented in Section~\ref{sec:blocks},
but considering $P$-presentations instead of letters makes the argument
much more involved.

Using known results on minimal solutions of Diophantine equations and some simple
calculus, an exponential upper-bound on the exponent of periodicity was shown.
The last step of this proof, i.e.\ the estimation of the minimal
solution, was relatively easy, while both the $P$-presentations and reduction
from $P$-presentations to a system of equations were involved.

We now show that using local recompression
one can obtain exponential upper bound on exponent of periodicity relatively easy.

\subsubsection*{Exponent of periodicity}
We begin with a bit more detailed definition of the exponent of periodicity.

\begin{definition}
For a word $w$ the \emph{exponent of periodicity} $\per(w)$ is the maximal $k$ such that $u^k$ is a substring of $w$, for some $u \in \letters^+$;
$\letters$-\emph{exponent of periodicity} $\per_\letters(w)$ restricts the choice of $u$ to $\letters$.
The notion of exponent of periodicity is naturally transferred from
strings to equations: For an equation $U = V$,
define the exponent of periodicity as
$$
\per(U = V) = \max_\solution \left[ \per(\sol U) \right]\enspace ,
$$
where the maximum is taken over all length-minimal solutions \solution{} of $U = V$;
define the $\letters$-\emph{exponent of periodicity} of $U = V$ in a similar way.
\end{definition}

We show that an exponential upper bound on $\letters$-exponent of periodicity is easy and natural
to obtain, one can think of it as a restriction of Ko\'scielski and Pacholski original
proof to its last part, i.e.\ to estimation of the minimal solution of a system
of Diophantine equations.
Then we show that the compression applied in \algsolveeq{} basically preserves
the exponent of periodicity, in particular it reduces the calculation of
upper bound on $\per(U = V)$ to calculation of upper bound on $\per_\letters(U = V)$.

\subsubsection*{Minimal solutions of linear Diophantine systems}
Consider a system of $m$ linear Diophantine equations in $r$ variables
$x_1$, \ldots, $x_r$, written as
\begin{subequations}
\label{eq:diophantine}
\begin{align}
\label{eq:diophantine eq}
\sum_{j=1}^r n_{i,j} x_j  &= n_i &\text{ for $i=1$, \ldots, $m$}
\intertext{together with inequalities guaranteeing that each $x_i$ is positive}
\label{eq:diophantine ineq}
x_j &\geq 1 &\text{ for $j=1$, \ldots, $r$}
\enspace .
\end{align}
\end{subequations}
In the following, we are interested only in \emph{natural} solutions, i.e.\ the ones in which 
each component is a natural number; observe that inequality~\eqref{eq:diophantine ineq}
guarantees that each of the component is greater than zero.
We introduce a partial ordering on such solutions:
$$
(q_1, \ldots ,q_r) \geq (q'_1, \ldots ,q'_r)
	\quad \text{ if and only if } \quad 
q_j \geq q'_j \text{ for each }j=1, \ldots, r.
$$
A solution $(q_1, \ldots ,q_r)$ is a \emph{minimal}
if it satisfies~\eqref{eq:diophantine} and there is no solution smaller than it.
(Note, that there may be incomparable minimal solutions.)

It is known, that each component of the minimal solution is at most exponential:
\begin{lemma}[cf.~{\cite[Corollary 4.4]{Koscielski}}]
\label{lem:diophantine upper bound}
For a system of linear Diophantine equations~\eqref{eq:diophantine}
let $w = r + \sum_{i=1}^m |n_i|$ and $c = \sum_{i=1}^m \sum_{j=1}^r |n_{i,j}|$.
If $(q_1, \ldots,q_r)$ is its minimal solution, then $q_j \leq (w+r)e^{c/e}$.
\end{lemma}
The proof is a slight extension of the original proof of Ko\'scielski and Pacholski,
which takes in to the account also the inequalities.
For completeness, we recall its proof, as given in~\cite{Koscielski}.
\begin{proof}[proof, cf.~\cite{Koscielski}]
The proof follows by estimation based on work of von zur Gathen and Sieveking~\cite{gathen}
and independently by Lambert~\cite{lambert}
\begin{clm}[von zur Gathen and Sieveking~\cite{gathen}; Lambert~\cite{lambert}]
\label{clm:gathen}
Consider a (vector) equations and inequalities $Ax = B$, $Cx \geq D$
with integer entries in $A$, $B$, $C$ and $D$.
Let $M$ be the upper bound on the absolute values
of the determinants of square submatrices of the matrix
$ \left( {\begin{array}{c}
 A\\
 C
\end{array} } \right)
$,
$r$ be the number of variables and $w$ the sum of absolute values
of elements in $B$ and $D$.
Then for each minimal natural solution $(q_1, \ldots , q_r)$ of~\eqref{eq:diophantine},
for each $1 \leq i \leq r$ we have $q_i \leq (w + r)M$.
\qed
\end{clm}
So it remains to estimate $M$ from Claim~\ref{clm:gathen}.
Observe that as the matrix $C$ in our case is an identity,
it is enough to consider the bound on the values of determinants
of square submatrices of $(n_{i,j})$, which was done by Ko\'scielski and Pacholski~\cite{Koscielski},
the rest of the proof is a simple recollection of their argument.

Recall the Hadamard inequality:
for any matrix $N = (n_{i,j})_{i,j = 1}^k$ we have
\begin{align*}
{\det}^2(N) &\leq \prod_{j=1}^k\sum_{i=1}^k n_{i,j}^2 \enspace .
\intertext{Therefore}
\det(N)
	&\leq
\left( \prod_{j=1}^k\sum_{i=1}^k n_{i,j}^2 \right)^{1/2} &\text{Hadamard inequality}\\
	&\leq
\left( \prod_{j=1}^k \left( \sum_{i=1}^k |n_{i,j}| \right) ^2 \right)^{1/2} &\text{trivial}\\
	&=
\prod_{j=1}^k \sum_{i=1}^k |n_{i,j}| &\text{simplification}\\
	&\leq
\left(\frac{\sum_{j=1}^k \left(\sum_{i=1}^k |n_{i,j}|\right)}{k}\right)^k &\text{inequality between means}\\
	&\leq
\left( \frac c k \right)^k &\text{by definition }
\sum_{j=1}^k \sum_{i=1}^k |n_{i,j}| = c\\
	&\leq
e^{c/e} \enspace &\text{calculus: sup at $k = c/e$}.
\end{align*}
Taking $N$ to be any submatrix of $(n_{i,j})$ yields that $M \leq e^{c/e}$
and consequently $q_i \leq (w+r) e^{c/e}$, as claimed.
\qedhere
\end{proof}

Now from Lemma~\ref{lem:diophantine upper bound} it can be easily concluded that
\begin{lemma}
\label{lem:small system solution}
In each minimal solution of the small system
of linear Diophantine equations for word equation $U = V$
all coordinates are $ \Ocomp((|U| + |V|) e^{2 n_v/e})$.
\end{lemma}
\begin{proof}
Recall that by the definition of the small system of linear Diophantine equations (for a word equation $U = V$),
this system has
\begin{itemize}
	\item at most twice as many variables as $U = V$, (so $r \leq 2n_v$ in terms of Lemma~\ref{lem:diophantine upper bound});
	\item the sum of coefficients at variables (so $c$ in the terms of Lemma~\ref{lem:diophantine upper bound}) is $2n_v$;
	\item the sum of values of constants of the equalities and inequalities
	(so $w$ in the terms of Lemma~\ref{lem:diophantine upper bound}) is $2(|U| + |V| + n_v)$
	(i.e.\ $2(|U| + |V|)$ for equations and $2n_v$ for the inequalities).
\end{itemize}
Hence from Lemma~\ref{lem:diophantine upper bound} it follows that each coordinate of a minimal solution of a small system of linear Diophantine
equations is at most
$$
2(|U| + |V| + n_v) e^{2n_v/e} = \Ocomp((|U| + |V|) e^{2 n_v/e}) \enspace,
$$
as claimed
\qedhere
\end{proof}

From Lemma~\ref{lem:small system solution} we can infer the upper-bound on the $\letters$-exponent of periodicity
of the length-minimal solution of the word equation.

\begin{lemma}[cf.~\cite{Koscielski}, cf.~Lemma~\ref{lem:periodicity bound original}]
\label{lem:diophantine solution word solution 2}
Consider a solution \solution{} of a word equation $U = V$,
the \solution-coherent Diophantine system $D$ and its solution $\{\ell_X,r_X\}_{X \in \variables}$
and the corresponding induced solutions $\solution[\{\ell_X,r_X\}_{X \in \variables}]$.
For a length-minimal $\solution'$ among them the $\letters$-exponent of periodicity of $\solution'(U)$ is $\Ocomp(n_v(|U|+|V|e^{2 n_v /e}))$, 
while $\per_\letters(\solution'(X))$ for any variable $X$ is $\Ocomp((|U|+|V|)e^{2 n_v /e})$.
\end{lemma}
\begin{proof}
By Lemma~\ref{lem:diophantine solution word solution} all solutions $\solution[\{\ell_X,r_X\}_{X \in \variables}]$ are similar.
Let, as in the statement, $\solution'$ be a length minimal among them, let it correspond to a solution 
$\{\ell_X',r_X'\}_{X \in \variables}$ of $D$.
Then by definition $\ell_X'$, ($r_X'$) are the lengths of the $a_X$-prefix ($b_X$-suffix) of $\solution'(X)$.
We show that $\{\ell_X',r_X'\}_{X \in \variables}$ is a minimal solution of $D$:
suppose for the sake of contradiction that it is not.
Then there is a solution $\{\ell_X'',r_X''\}_{X \in \variables}$ of $D$, such that
\begin{equation}
	\label{eq: minimal solution}
\ell_X'' \leq \ell_X' \quad \text{ and } \quad r_X'' \leq r_X' \quad \text{ for each } X \in \variables
\end{equation}
and at least one of those inequalities is strict, without loss of generality let $\ell_Y'' < \ell_Y'$.
By Lemma~\ref{lem:diophantine solution word solution} for each variable $X$ there is an arithmetic expression $e_X$
such that $|\solution'(X)| = e_X[\{\ell_X',r_X'\}_{X \in \variables}]$
and $|\solution''| = e_X[\{\ell_X'',r_X''\}_{X \in \variables}]$.
By~\eqref{eq: minimal solution} we obtain that $|\solution'(X)| \geq |\solution''(X)|$ for each variable.
Furthermore, Lemma~\ref{lem:diophantine solution word solution} also guarantees that each $e_X$
depends on $x_X$ and $y_X$ (if $y_X$ is used at all),
hence by the choice of $Y$ also $|\solution'(Y)| > |\solution''(Y)|$
and so $\solution'$ is not length-minimal, contradiction.

Then by the minimality of $\{\ell_X',r_X'\}_{X \in \variables}$ we obtain that each $\ell_X'$ and $r_X'$ is $\Ocomp((|U| + |V|) e^{2 n_v /e})$,
by Lemma~\ref{lem:small system solution}.
As the maximal $a$ block is a concatenation of explicit letters from the equation and $a_X$-prefixes and $b_X$-suffixes
of \sol X for various $X$, its length is at most
$$
n_v \cdot \max_{X \in \variables}(\ell_X,r_X) + (|U| + |V| - n_v) = \Ocomp(n_v (|U| + |V|) e^{2 n_v /e}) \enspace,
$$
which ends the proof.
\qedhere
\end{proof}

As a short corollary we obtain: 

\begin{theorem}[cf.~\cite{Koscielski}, cf.~Lemma~\ref{lem:periodicity bound original}]
\label{thm:diophantine solution word solution}
The $\letters$-exponent of periodicity of a word equation $U = V$ with $n_v$ occurrences of variables is $\Ocomp(n_v(|U|+|V|e^{2 n_v /e}))$.
\end{theorem}

\subsubsection*{General exponent of periodicity}
So far we have only shown that $\letters$-exponent of periodicity is at most exponential.
However judging by the work of Ko\'scielski and Pacholski~\cite{Koscielski},
the difficulty is elsewhere, in the case of exponent of periodicity for longer words.
We show that this is not the case: in the following lemma we show that employing the recompression technique
we obtain an exponential bound on the exponent of periodicity as a corollary of a similar bound for $\letters$-exponent of periodicity.
Unfortunately, our result is weaker than the one obtained by Ko\'scielski and Pacholski,
as they in fact had a $2 ^{cn}$ bound, for appropriate $c$.

\begin{lemma}
\label{lem: how per goes back}
Let $U = V$ with a solution \solution{} be transformed by some subprocedure of \algsolveeq{},
i.e. \algpairc{} or \algblocksc{} (or \algblocksi) into $U' = V'$ with $\solution'$.
Then $\per(\solution'(U')) \leq \per(\sol U)$. Furthermore
\begin{align}
	\tag{per 1} \label{per1} \per(\sol U) &= \per_\letters(\sol U) \text{, or} \\
	\tag{per 2} \label{per per} \per(\solution'(U')) &\geq \per(\sol U) - 1 \enspace .
\end{align}
\end{lemma}
\begin{proof}
Recall that by Lemmata~\ref{lem:pop preserves solutions} and~\ref{lem:cutpref cutsuff}
for \algpop{} and \algprefsuff{} it holds that $\sol U = \solution'(U')$ and so the claim trivially holds,
as $\per(\solution'(U')) = \per(\sol U)$.
So it is enough to show the claim for \algpairc{} and \algblocksc{}
restricted to compression (the analysis for \algblocksi{} is the same).

We first show that $\per(\sol U) \geq \per(\solution'(U'))$ for \algpairc.
By Lemma~\ref{lem: crossing pairs transforms} the corresponding inverse operator (when we restrict ourselves to compression) is
$\prod_{ab \in \letters_\ell\letters_r} \Morph {c^{(ab)}} {ab}$.
Hence $\sol U = \prod_{ab \in \letters_\ell\letters_r} \morph {c^{(ab)}} {ab} (\solution'(U'))$.
Let $w^k$ be a substring of $\solution'(U')$, then
$\prod_{ab \in \letters_\ell\letters_r} \morph {c^{(ab)}} {ab} (w^k) = (\prod_{ab \in \letters_\ell\letters_r} \morph {c^{(ab)}} {ab} (w) )^k $
is a substring of \sol U, hence $\per(\solution'(U')) \leq \per(\sol U)$.

Similarly, \algblocksc{} is a composition of \algprefsuff, which preserves the exponent of periodicity.
Hence it is enough to consider the inverse operator for \algblocksc{} restricted to the compression.
By Lemma~\ref{lem: consistent no crossing block transform}
it is $\prod_{a \in \letters} \Blockc$, where $\blockc$ replaces $a^\ell$ with $a_\ell$ for some maximal blocks $a^\ell$ and letters $a_\ell$.
Hence $\sol U = \prod_{a \in \letters} \blockc (\solution'(U'))$.
Consider any $w^k$ that is a substring of $\solution'(U')$.
Then $ \prod_{a \in \letters} \blockc (\solution'(w^k)) = (\prod_{a \in \letters} \blockc (\solution'(w)))^k$ is a substring of \sol U.
Thus $\per(\solution'(U')) \leq \per(\sol U)$.

We move to the second claim of the lemma, i.e.\ we are going to show that~\eqref{per1} or~\eqref{per per} holds.
Let $m = \per(\sol U)$. If there is $a \in \letters$ such that $a^m$
is a substring of \sol U, then~\eqref{per1} holds.
So assume that $w^m$ is a substring of \sol U, for some $w \notin \letters \cup \{ \epsilon \}$.
Moreover, we can assume that $w \neq a ^k$ for every $a$ and $k$,
as this clearly reduces to the case of $w = a$.

Consider first $\algpairc(\presentletters_1,\presentletters_2)$,
and let $w = bua$, recall that by the assumption $|w| > 1$ and so $|u| \geq 0$, i.e.\ it can be that $u = \epsilon$.
How does the image of $w^m$ looks like in $\solution' (U')$?
This depends on whether $a \in \presentletters_1$ and whether $b \in \presentletters_2$, in total there are four cases.
From Lemma~\ref{lem: crossing pairs transforms} we know that for
$h^{-1} = \prod_{a' \in \presentletters_\ell, b' \in \presentletters_r} \invmorph c {a'b'}$
we have $\solution'(U') = h^{-1} (\sol U))$.
\begin{description}
	\item[$b \notin \presentletters_2$]
	The further analysis depends on whether $a \in \presentletters_1$ or not
	\begin{description}
		\item[$a \notin \presentletters_1$]
		Consider any $w = bua$ in $w^m$. Observe that by case assumptions, the first letter of $w$ is never compressed with letter
		to the left and the last letter is never compressed with the letter to the right.
		So in this case $w^m$ after compression will be represented as $(h^{-1} (w))^m$,
		and so $(h^{-1} (w))^m$ is a substring of $h^{-1}(\sol U)$;
		thus, $\per(\solution'(U')) \geq \per(\sol U)$.
		\item[$a \in \presentletters_1$]
		Consider $w^m  =(bua)^m$. As in the previous case, the leading $b$ is never compressed with the previous letter
		In this case it might be that the last letter of $w^m$ is compressed with the following letter,
		however, each other last $a$ in $bua$ is not (as the following letter is $b \notin \presentletters_2$).
		Hence the compression of the prefix $w^{m-1}$ results in $(h^{-1} (w))^{m-1}$
		and so $\per(\solution'(U')) \geq \per(\sol U)-1$.
	\end{description}
	\item[$b \in \presentletters_2$]
	Similarly, the further analysis depends on whether $a \in \presentletters_1$ or not
	\begin{description}
		\item[$a \notin \presentletters_1$]
		The case is symmetric to the subcase above, in which $a \in \presentletters_1$ and $b\notin\presentletters_2$,
		in particular in a similar way we show that $\per(\solution'(U')) \geq \per(\sol U)-1$.
		\item[$a \in \presentletters_1$]
		Represent $w^m$ as $b (uab)^{m-1} ua$.
		Observe that each $ab$ in $(uab)^{m-1}$ is compressed and replaced with a new letter $c$.
		Furthermore, the first letter in each $u$ in $(uab)^{m-1}$ is not compressed with the letter to the left,
		as this is in each case $b \in \presentletters_2$.
		Hence, $(uab)^{m-1}$ is compressed into $(h^{-1} (uab))^{m-1}$ and so $\per(\solution'(U')) \geq \per(\sol U)-1$.
	\end{description}
\end{description}
The analysis for $\algblocksc$ (and similarly $\algblocksi$) is even simpler: let $w = b^\ell u a^r$, where $u$ does not begin with $b$ and does not end with $a$.
By Lemma~\ref{lem: consistent no crossing block transform} we know that for $\invblockc[] = \prod_{a \in \letters} \invblockc$
we have $\solution'(U') = \invblockc[] (\sol U)$
Then
$$
w^m = (b^\ell u a^r)^m = b^\ell (u a^r b^\ell)^{m-1} u a^r.
$$
As by the assumption $u$ does not start with $b$, hence $\solution'(U')$
contains $ \invblockc[] ( (u a^r b^\ell)^{m-1}) = (\invblockc[] (u a^r b^\ell))^{m-1}$,
and so $\per(\solution'(U')) \geq m-1$, as claimed.
\qedhere
\end{proof}

As a promised corollary we obtain the exponential bound on the exponent of periodicity.

\begin{theorem}[cf.~\cite{Koscielski}, cf.~Lemma~\ref{lem:periodicity bound original}]
The exponent of periodicity of equation of a length-minimal solution \solution{}
is single exponential in $|U| + |V|$.
\end{theorem}
\begin{proof}
Denote $U = V$ and some its length-minimal solution \solution{} by $U_1 = V_1$ and $\solution_1$.
Let $U_1 = V_1$, $U_2 = V_2$, \ldots, $U_m = V_m$ be all equations generated
during the run of \algsolveeq, in this order, let $\solution_1$ be transformed to $\solution_2$, \ldots, $\solution_m$ during this run,
and let $\phi_2$, \ldots, $\phi_m$ be the corresponding inverse operators.
We claim that if $\solution_1$ is length-minimal
then for each $i$ we have that
\begin{equation}
\label{eq: bound on per letters}
\per_\letters(\solution_i(U_i)) = \Ocomp(n_v (|U_i| + |V_i|) e^{2 n_v /e}) \enspace .
\end{equation}
Suppose that this is not the case.
Consider the $\solution_i$-coherent system of Diophantine equations,
let $\solution_i$ correspond to $\{\ell_X, r_X\}_{X \in \variables}$
and consider some minimal solution $\{\ell_X', r_X'\}_{X \in \variables}$ that is not larger than $\{\ell_X, r_X\}_{X \in \variables}$.
Then by Lemma~\ref{lem:diophantine solution word solution}
the $\solution_i' = \solution_i[\{\ell_X', r_X'\}_{X \in \variables}]$ is also a solution,
which is similar to $\solution_i$.
As those solutions are similar, $\solution_i'$ can be obtained by deleting some letters from $\solution_i$.
Then $\phi_2 \circ \phi_3 \circ \dots \circ \phi_i [\solution_i']$ is also a solution of $U_1 = V_1$
which is shorter than $\solution = \phi_2 \circ \phi_3 \circ \dots \circ \phi_i [\solution_i]$,
contradiction.

None of the equations $U_1 = V_1$, $U_2 = V_2$, \ldots, $U_m = V_m$ is repeated, and as each of them is of length at most $c' n$
(see Lemma~\ref{lem: space consumption}),
thus $m \leq (c'n)^{c'n+1} \leq n^{cn}$, for some constants $c$ and $c'$.
By Lemma~\ref{lem: how per goes back} it holds that $\per(\solution_i(U_i)) = \per_\letters(\solution_i(U_i))$
or $\per(\solution_i(U_i)) \leq \per(\solution_{i+1}(U_{i+1})) + 1$.
Observe that for $\solution_m(U_m)$ we have $\per(\solution_m(U_m)) = 1$:
since $|U_m|, |V_m| \leq 1$ we have two cases:
\begin{itemize}
	\item if any of $U_m$ or $V_m$ is $\epsilon$ then $\solution_m(U_m) = \epsilon$ and $\per(\solution_m(U_m)) = 0$;
	\item if one of $U_m$ or $V_m$ is a letter, say $a$, then $\solution_m(U_m) = a$ and clearly $\per(\solution_m(U_m)) = 1$;
	\item if both $U_m$ and $V_m$ are variables, say $U_m$ is $X$ and $V_m$ is $Y$ (we do not assume that $X \neq Y$).
	Suppose that $\solution_m(X) = \solution_m(Y)$ is longer than one letter, say it is $a w$.
	Consider $\solution_m'$, where $\solution_m'(X) = \solution_m'(Y) = w$ and it is equal to $\solution_m$ otherwise.
	Then $\solution_1' = \phi_2 \circ \phi_3 \circ \dots \circ \phi_m [\solution_m']$ is also a solution of $U_1 = V_1$
	and $\solution_1'(U_1)$ is shorter than $\solution_1(U_1)$,
	as $\solution_1 = \phi_2 \circ \phi_3 \circ \dots \circ \phi_m [\solution_m]$.
	This contradicts the assumption that $\solution_1$ is length minimal.
	Hence $\solution_m(X) = \solution_m(Y)$ has only one letter and so $\per(\solution_m(U_m)) = 1$.
\end{itemize}

Let $i$ be the smallest index among $1, 2, \ldots, m$ such that $\per_\letters(\solution_i(U_i)) = \per(\solution_i(U_i))$.
Note that such an $i$ exists, as $m$ satisfies this condition.
Recall that by~\eqref{eq: bound on per letters}
\begin{align*}
\per_\letters(\solution_i(U_i))
	&\leq 
c' n_v (|U_i| + |V_i|) e^{2 n_v /e} \enspace
\intertext{for some constant $c'$. By~\eqref{per per} for $j = 1, 2, \ldots, m-1$ we have
$\per(\solution_j(U_j)) \leq \per(\solution_{j+1}(U_{j+1})) + 1$ we conclude that}
\per(\solution_1(U_1))
	&=
(i-1) + c' n_v (|U_i| + |V_i|) e^{2 n_v /e}\\
	&= 
\Ocomp(n^{cn} + n_v n e^{2 n_v /e})\\
	&=
\Ocomp(n^{cn}) \enspace,
\end{align*}
for some constant $c$, in particular it is single exponential in $n = |U_1| + |V_1|$.
Since this holds for an arbitrary lenght-minimal solution $\solution_1$ of $U_1 = V_1$,
this yields the claim.
\end{proof}

\subsection*{Double exponential bound on minimal solutions}
It was shown by Plandowski~\cite{PlandowskiSTOC} that
the size of the length minimal solution of word equation
is always doubly exponential.
This result was achieved by careful and clever analysis of factorisations
of minimal solutions;
the proof is basically independent from the analysis of the \algPlandowskiFOCS,
though uses similar types of factorisations of words
(and in fact the doubly-exponential bound can be inferred from \algPlandowskiFOCS{}
after some simple modifications~\cite{PlandowskiPersonal}).

Since we know that on one hand the running time of \algsolveeq{} is polynomial in $n$ and $\log N$ (see Theorem~\ref{thm:main})
on the other the space consumption is $\Ocomp(n \log n)$ (see Lemma~\ref{lem:improved block compression}), the doubly exponential upper bound
on $\log N$ seems natural. However, both presented bounds are upper bounds and so cannot be directly compared.
To compare them we want to show that the running time is in fact also \emph{lower-bounded} in terms
of $\log N$.
\begin{lemma}
\label{lem: phase lower bound}
Let $N$ be the size of the length-minimal solution of a word equation of size $n$.
Then the number of phases of \algsolveeq{} is $\Omega(\log N / \poly(n))$
for every accepting run, regardless of the nondeterministic choices.
\end{lemma}
\begin{proof}
Suppose that the equation $U = V$ is transformed into
an unsatisfiable equation $U' = V'$;
then we are done, as it will never be turned into a satisfiable instance.
So in the following we consider only the case, in which
each of the equations is satisfiable.

The solution $\solution'$ of $U' = V'$ is obtained from $\solution$ of $U = V$
by two separate compression sub-phases:
in the first, some maximal blocks of letters may be compressed into one letter,
in the second, some pairs $ab$, for $a \neq b$ are replaced by a fresh letter
(see Lemmata~\ref{lem: consistent no crossing block transform} and \ref{lem: crossing pairs transforms}).
In the following we shall compare the lengths of the length-minimal solutions before and after one such compression subphase,
i.e.\ estimate  $N/N'$, where $N$ and $N'$ are the lengths of the length-minimal solutions before and after the subphase, respectively.

We begin with the second phase, as it is easier to analyse.
Notice, that if $c$ is introduced as a letter for a pair $ab$ then
$c$ is not compressed in the rest of this subphase,
hence at most two letters are compressed into one and those new letters are not further compressed.
Let \solution{} be a length minimal solution of $U = V$, i.e.\ of size $N$.
Then by Lemma~\ref{lem:crossing non crossing} the obtained equation $U' = V'$ has a solution $\solution'$ which implements the pair compression,
in particular, it is at most two times shorter than \sol U.
Hence
\begin{align*}
\frac{N}{N'} &= \frac{|\sol U|}{N'}\\
	&\leq
\frac{|\sol U|}{|\solution'(U')|}\\
	&\leq
2 \enspace.
\end{align*}
Thus, the second compression subphase shortens the shortest solution by a factor of at most $2$.
Let us return to the first sub-phase.

Consider any length-minimal solution $\solution'$ of $U' = V'$, let its length be $N'$.
Take any solution \solution{} that is transformed into $\solution'$ by \algblocksi.
Consider the \solution-coherent system of Diophantine equations $D$ and the solutions
$\solution[\{\ell_X,r_X\}_{X \in \variables}]$ induced by different solutions of $D$,
see Lemma~\ref{lem:diophantine solution word solution}.
Take the length-minimal among them, let it be $\solution_1$.
Then its $\letters$-exponent of periodicity is $\Ocomp((|U|+|V|)e^{2 n_v /e})$ by Lemma~\ref{lem:diophantine solution word solution 2}.
Now, note that as \solution{} and $\solution_1$ are similar,
the application of block compression to \sol U and $\solution_1(U)$ results in a string of the same length:
similar solutions have the same number of maximal blocks and each of those blocks is replaced with a single letter.
As the former is $\solution'(U')$, we get that $\solution'(U')$ is $\Ocomp((|U|+|V|)e^{2 n_v /e})$ times shorter than $\solution_1(U)$.
Consequently
\begin{align*}
\frac{|\solution_1(U)|}{N'}\\
	&=
\frac{|\solution_1(U)|}{|\solution'(U')|}
	&
\leq cn e^{2 n_v /e} \enspace .
\intertext{Since $|\solution_1(U)| \geq N$, where $N$ is the length of the length-minimal solution of $U = V$,
we obtain that}
\frac{N}{N'}
	&\leq
\frac{|\solution_1(U)|}{|\solution'(U')|}\\
	&\leq
cn_v e^{2 n_v /e} \enspace .
\intertext{Taking into the account the factor $2$ in the second sub-phase we obtain the upper bound}
2cn e^{2 n_v /e}
\end{align*}
on the proportion between length minimal solutions in the consecutive phases.

So let $N = N_1, N_2,\ldots,N_m$ be the lengths of length-minimal solutions in consecutive phases,
where $m$ is the last phase.
Then $N_{i}/N_{i+1} \leq 2 cn e^{2 n_v /e}$ and $N_m \leq 1$, hence
\begin{align*}
N &\leq (2cn e^{2 n_v /e})^{m}
\intertext{and so}
m &\geq \frac{\log N}{\poly(n)} \enspace,
\end{align*}
as claimed.
\qedhere
\end{proof}

\begin{corollary}[cf.~\cite{PlandowskiSTOC}]
The size of the length-minimal solution of a word equation of size $n$ is at most $2^{q(n) \cdot n_v ^{cn_v}}$
for some polynomial $q$ and constant $c$.
\end{corollary}
\begin{proof}
By Lemma~\ref{lem: space consumption}
the equation stored by \algsolveeq{} has at most $c n_v^{cn_v} \log n$ many phases.
On the other hand, by Lemma~\ref{lem: phase lower bound}, there are at least
$c' (\log N)/p(n)$ phases, for some constant $c'$ and polynomial $p$.
Thus,
$$
c' (\log N)/p(n) \leq c n_v^{cn_v} \log n \enspace,
$$
which yields the claim.
\qedhere
\end{proof}

\subsection*{Acknowledgements}
I would like to thank A.~Okhotin for his remarks about ingenuity of Plandowski's
result, which somehow stayed in my memory;
P.~Gawrychowski for initiating my interest in compressed membership problems
and compressed pattern matching, exploring which eventually led to this work
as well as for pointing to relevant literature~\cite{LohreySLP,MehlhornSU97};
J.~Karhum{\"a}ki, for his explicit question, whether the techniques of
local recompression can be applied to the word equations;
last not least, W.~Plandowski for his numerous comments and suggestions on the manuscript
as well as questions concerning the exact space consumption
that eventually led to results in Section~\ref{sec:linear space}.

\end{document}